%% file: arxiv.tex
\documentclass[10pt]{article}

\usepackage{fullpage}
\usepackage{microtype}
\usepackage{graphicx}
\usepackage{booktabs} %
\usepackage{adjustbox}
\usepackage{makecell}
\usepackage{color-edits}
\addauthor{vs}{red}
\usepackage{authblk}
\usepackage{bbm}
\usepackage{subcaption}
\usepackage{caption}
\usepackage{float}
\usepackage{appendix}

\usepackage[round]{natbib}
\input{math_commands}
\usepackage{algorithm, algpseudocode}

\usepackage{hyperref}
\hypersetup{
  colorlinks= true, 
  urlcolor= blue, 
  linkcolor= red, 
  citecolor= blue 
}

\def\one{\mathbbm{1}}

\usepackage{amsmath}
\usepackage{amssymb}
\usepackage{mathtools}
\usepackage{amsthm}

\usepackage[capitalize,noabbrev]{cleveref}
\usepackage{amsthm}
\newtheorem{theorem}{Theorem}[section]
\newtheorem{lemma}{Lemma}[section]

\newtheorem{assumption}{Assumption}
\newtheorem{proposition}{Proposition}[section]
\newtheorem{corollary}{Corollary}[section]

\newtheorem{example}{Example}[section]

\theoremstyle{remark}
\newtheorem{remark}[theorem]{Remark}

\usepackage[textsize=tiny]{todonotes}

\title{Shape-Adaptive Conditional Calibration for Conformal Prediction via Minimax Optimization}
\author{ Yajie Bao$^1$\thanks{The first two authors contributed equally to this work and are listed in alphabetical order.}, Chuchen Zhang$^{1*}$, Zhaojun Wang$^1$, Haojie Ren$^{2}$, and Changliang Zou$^1$\\
$^1$ {\small School of Statistics and Data Science, Nankai University}\\
$^2$ {\small School of Mathematical Sciences,  Shanghai Jiao Tong University}}
\date{\today}

\begin{document}

\maketitle

\begin{abstract}

Achieving valid conditional coverage in conformal prediction is challenging due to the theoretical difficulty of satisfying pointwise constraints in finite samples. Building upon the characterization of conditional coverage through marginal moment restrictions, we introduce Minimax Optimization Predictive Inference (MOPI), a framework that generalizes prior work by optimizing over a flexible class of set-valued mappings during the calibration phase, rather than simply calibrating a fixed sublevel set. This minimax formulation effectively circumvents the structural limitations of predefined score functions, achieving superior shape adaptivity while maintaining a principled connection to the minimization of mean squared coverage error. {Theoretically, we provide non-asymptotic oracle inequalities for arbitrary conditioning variables and show that the convergence rate of the coverage error attains the optimal order in the group-conditional coverage case.} Empirical results demonstrate that MOPI produces more efficient conditional prediction sets than existing baselines.
\end{abstract}
\noindent%


\section{Introduction}
Predictive inference aims to construct a prediction set for an unknown label based on a machine learning model, ensuring that the prediction set covers the true label with a pre-specified confidence level. This statistical guarantee is essential in high-stakes applications, such as medical diagnosis \citep{vazquez2022conformal}, safe planning in autonomous systems \citep{lindemann2023safe}, and robust decision-making under uncertainty \citep{johnstone2021conformal}, where reliable and uncertainty-aware predictions can enhance decision safety.

Conformal prediction \citep{vovk2005algorithmic} is a distribution-free and model-agnostic tool in predictive inference. Given the test covariate $X\in \gX$, it can issue an $(1-\alpha)$-level prediction set $C(X)$ for the unobserved label $Y \in \gY$ by using the collected labeled data and a machine learning model. Suppose the labeled data and the test data $(X,Y)$ are independent and identically distributed (i.i.d.) or exchangeable, then the conformal prediction set satisfies the \textit{marginal} coverage property $\sP\{Y \in C(X)\}\ge 1-\alpha$, see \citet{lei2018distribution}. Although marginal coverage offers a simple and attractive validity guarantee, it may fail to provide uniform reliability across individuals or under heterogeneous data distributions. For example, certain subgroups or regions of the covariate space may experience systematic under-coverage. To overcome this limitation, a more stringent and relevant objective is the \textit{test-conditional} coverage: $\sP\{Y \in C(X) \mid X\} \geq 1-\alpha$ almost surely.

Previous work has shown that achieving exact test-conditional coverage is impossible in a distribution-free setting without resorting to trivial prediction sets of infinite expected size \citep{vovk2012conditional,lei2013distribution}. Consequently, various conformal prediction methods have been proposed to achieve approximate or asymptotic conditional coverage by either modifying the calibration step \citep{guan2023localized,hore2024conformal, gibbs2023conformal} or using different score functions \citep{romano2019conformalized, chernozhukov2021distributional}.
 
Consider a generalized notion for conditional coverage:
\begin{align}\label{eq:Z_conditional_cover}
    \sP\{Y\in C(X)\mid Z = z\} = 1-\alpha,\quad \forall z\in \gZ,
\end{align}
where $Z \in \gZ$ is an \emph{arbitrary} conditioning variable. The formulation in \eqref{eq:Z_conditional_cover} encapsulates several important cases: test-conditional coverage when $Z = X$; group-conditional coverage when $Z = \one\{X\in G\}$ where $G$ is a specific group in the covariate domain $\gX$; the equalized coverage in \citet{Romano2020With} when $Z$ represents a sensitive attribute (e.g., gender, race). To achieve this conditional coverage target, a natural strategy is the ``conditional-to-marginal'' relaxation, which transforms the conditional constraint into a set of marginal moment restrictions \citep{andrews2013inference}:
\begin{align}\label{eq:cond_to_marg}
    \E\LRm{f(Z) \LRs{\mathbbm{1}\{Y\notin C(X)\} - \alpha}} = 0,\quad \forall f\in \gF,
\end{align}
where $\gF$ is a ``weight'' function class defined on the conditioning domain $\gZ$. When $\gF$ includes all measurable functions, the conditional coverage target \eqref{eq:Z_conditional_cover} is equivalent to the marginal relation \eqref{eq:cond_to_marg}. A seminal approach leveraging this relaxation is the \emph{conditional calibration} framework proposed by \citet{gibbs2023conformal}, which considers the prediction set of the sublevel form:
${C(X,Z)} = \{y\in \gY: s(X,y) \leq f(Z)\}$,
where $s$ is a fixed score function pretrained on a separate dataset and $f\in \gF$ is a threshold function to be determined. In this context, finding a prediction set that satisfies \eqref{eq:cond_to_marg} is equivalent to solving a functional quantile regression problem: $\min_{f\in \gF}\E[\ell_{\alpha}(s(X,Y), f(Z))]$, where $\ell_{\alpha}(u,v) = (u-v)(\mathbbm{1}\{u-v>0\}-\alpha)$ for $u,v\in \sR$.

While computationally tractable, this formulation inherently restricts the predictive region to the sublevel sets of a fixed score function. It can only adjust the volume of the set in the calibration phase but fails to rotate the shape or alter its aspect ratio to capture local correlations among output variables. This {\it geometric rigidity} is especially pronounced in multivariate regression with vector-valued labels.
Although recent efforts have focused on designing efficient scores during the training phase \citep{johnstone2021conformal,thurin2025optimal}, the resulting predictive sets still lack the flexibility to adapt to local distributional heteroscedasticity during calibration. 
Furthermore, this paradigm typically requires the threshold $f(\cdot)$ to be a direct mapping from observed features. Consequently, it struggles to handle scenarios where $Z$ represents sensitive attributes that can be masked at test time \citep{zafar2017fairness}. In such cases, a threshold function learned solely on $X$ cannot effectively incorporate the necessary conditioning information from $Z$ during calibration.

\subsection{A snapshot of our ideas and contributions}

To break these constraints, we introduce {\it Minimax Optimization Predictive Inference} (MOPI), a novel framework that reformulates conditional predictive inference as a minimax problem. We consider a general and structured set-valued function class:
\begin{equation}\label{eq:structure_set}
    \fC = \big\{ C(x;h) = \{ y\in \gY : T(h(x),y)\leq 0\} : h\in \gH \big\},
\end{equation}
where  $T(\cdot,\cdot)$ is a fixed function, $h(x)$ can be a vector or matrix-valued function encoding the geometric structure (e.g., mean $\mu(x)$ or covariance $\Sigma(x)$), and $\gH$ is a function class in the covariate domain $\gX$. Assuming $\gF$ is symmetric (i.e., $-f \in \gF$ whenever $f \in \gF$), the deviation of the left-hand side of \eqref{eq:cond_to_marg} from zero can be quantified by the maximum quantity $\max_{f\in \gF}\E[f(Z) (\mathbbm{1}\{Y\in C(X)\} - (1-\alpha))]$. Hence, we can construct the desired prediction set by \emph{minimizing this maximum quantity} over all set-valued functions $C \in \fC$, leading naturally to a minimax optimization problem \citep{dikkala2020minimax}:
\begin{equation}
    \min_{C \in \mathfrak{C}} \max_{f \in \mathcal{F}} \mathbb{E} [f(Z)(\mathbbm{1}\{Y \notin C(X)\} - \alpha) - f^2(Z)],\nonumber
\end{equation}
where the term $f^2(Z)$ is used for normalization (see Section \ref{sec:unify_minimax} for details).  This framework offers two fundamental advantages. First, by optimizing $C \in \fC$, MOPI can dynamically ``reshape'' the prediction set through finding the optimal $h\in \gH$ in \eqref{eq:structure_set}, such as adjusting the aspect ratio of a box or the orientation of an ellipsoid, to match local heteroscedasticity. 
Second, MOPI is inherently suited for masked conditioning variables. By decoupling the coverage enforcement (inner maximization {depending on $Z$}) from the predictive mapping (outer minimization {depending on $X$}), our framework allows information from $Z$ to guide the construction of prediction set $C(X)$ during the calibration phase without requiring $Z$ at prediction time.

The main contributions of this paper are summarized as follows:
\begin{itemize}

    \item[(1)] MOPI generalizes \citet{gibbs2023conformal}’s conditional calibration principle by allowing the prediction set's geometry to co-evolve with the local data density. Moreover, unlike existing conditional conformal methods, which typically require the conditioning variable $Z$ to be a subset of the test covariates (i.e., $Z \subseteq X$), MOPI allows $Z$ to be distinct from $X$. This decoupling enables our framework to enforce conditional coverage on sensitive attributes or latent groups that are available during calibration but masked at test time.

    \item[(2)] We demonstrate that MOPI objective possesses a rigorous statistical interpretation: under expressiveness conditions on the weight function class $\mathcal{F}$, solving the minimax problem is equivalent to minimizing the \textit{Mean Squared Coverage Error} (MSCE), defined as $\mathsf{MSCE}(C) = \mathbb{E}[(\mathbb{P}\{Y \notin C(X) \mid Z\} - \alpha)^2]$. We further establish non-asymptotic oracle inequalities for the MSCE of the empirical MOPI predictor. {Our analysis yields explicit convergence rates under both finite-dimensional and RKHS weight function classes, achieving the optimal order up to logarithmic factors for group-conditional coverage.}

    \item[(3)] We provide a principled guide for choosing the weight function class $\gF$ in \eqref{eq:cond_to_marg} according to the cardinality of the conditioning domain $\gZ$. Under practical choices of $\gF$, the inner maximization of the minimax problem admits a closed form, which greatly facilitates the overall optimization process. {To solve the empirical problem, we replace the miscoverage indicator with a smooth surrogate, and the corresponding MSCE bound is also established.}

    \item[(4)] We evaluate MOPI on extensive synthetic and real-world datasets, including those with multi-dimensional labels and masked sensitive attributes. The results demonstrate that MOPI yields more compact prediction sets while maintaining robust conditional coverage performance.

    
\end{itemize}

\subsection{Related works}

Several studies have examined how to achieve test-conditional coverage in the asymptotic regime, typically by designing improved nonconformity scores or by reweighting labeled data. For example, \citet{romano2019conformalized} proposed a conformalized quantile regression score that leverages quantile regression to capture local distributional characteristics of the label. \citet{chernozhukov2021distributional} introduced a nonconformity score based on joint distribution modeling. More recently, \citet{guan2023localized} developed a localized conformal prediction approach that assigned similarity-based weights to calibration scores, which ensured finite-sample marginal coverage and asymptotic conditional coverage under mild distributional assumptions. Building on this framework, \citet{hore2024conformal} proposed a randomly localized conformal prediction method that offers relaxed local coverage and marginal validity under certain covariate shifts. Further examples and developments can also be found in \citet{gyorfi2019nearest}, \citet{sesia2021conformal}, and \citet{kiyani2024conformal}.

Another line of work has focused on approximate notions of test-conditional coverage. \citet{lei2013nonparametric} proposed partitioning the covariate space into disjoint bins and constructing prediction sets using only calibration points within the same bin, thereby achieving group-conditional coverage within each bin of these partitions. \citet{foygel2021limits} studied the problem of obtaining approximate conditional coverage over every regular subset of the covariate space. Drawing on the idea of multiaccuracy from algorithmic fairness \citep{hebert2018multicalibration}, \citet{jung2022batch} developed a multivalid conformal prediction method that asymptotically guarantees conditional coverage over any user-specified finite collection of (possibly overlapping) groups. {Notably, \citet{gibbs2023conformal} introduced a conditional calibration (abbreviated as \texttt{CC}) framework via functional quantile regression. For the finite-dimensional case, their method provides a finite-sample group-conditional coverage guarantee and maintains marginal validity under a class of covariate shifts. In the infinite-dimensional setting, they characterized relaxed conditional coverage by bounding the absolute deviation in \eqref{eq:cond_to_marg}. In contrast, our work establishes an oracle inequality for the MSCE, which offers a more direct quantification of conditional coverage.}

For vector-valued responses, recent studies have increasingly focused on designing scoring functions that induce specific geometries for conformal prediction sets. \citet{johnstone2021conformal} and \citet{messoudi2022ellipsoidal} constructed ellipsoidal prediction regions using the Mahalanobis distance between the prediction and the response as a nonconformity score. \citet{feldman2023calibrated} first learned latent-space representations of the original data and then applied multi-output quantile regression to these representations, enabling the prediction set transformed to the original space to take a non-convex form. \citet{thurin2025optimal} employed optimal transport to define a ranking of multivariate scores, thereby generating prediction regions that go beyond predefined geometric families. \citet{braun2025minimum} proposed an optimization-based framework by minimizing the volume over a set-valued class defined by arbitrary norm balls, which induces a new nonconformity score. {Similar to the traditional pipeline of conformal prediction, these approaches primarily focused on first designing a nonconformity score and then calibrating the threshold, rather than adapting the shape of prediction sets during calibration. Furthermore, they only provided marginal coverage and did not explicitly address conditional coverage.}

The minimax formulations for moment-restriction problems have been explored in several related contexts, including solving nonparametric regression problems \citep{dikkala2020minimax},
off-policy evaluation in reinforcement learning \citep{shi2022minimax}, and parameter estimation under conditional moment constraints \citep{bennett2023variational}.
In contrast, these prior works focused on learning point‑prediction models. In the context of conformal prediction, \citet{kiyani2024length} studied the length minimization for sublevel sets of a pretrained score while ensuring the group-conditional validity, and transformed the constrained optimization problem to a minimax formulation. However, both the motivation and objective in \citet{kiyani2024length} differ from ours, and the theoretical analysis is fundamentally distinct.



\noindent\textbf{Notations.}
For Euclidean spaces $\gX$, $\gY$, and $\gZ$, we use $d_{\gX}$, $d_{\gY}$, and $d_{\gZ}$, respectively, to denote the corresponding dimension. When $\mathcal{Z}$ is finite, we write $|\mathcal{Z}|$ for its cardinality.
For a Euclidean vector $u\in \sR^{d}$, we write 
$\|u\|_{\infty} = \max_{1\leq j\leq d} |u_j|$. For a function $g:\gX\times \gY \times \gZ \to \sR$, we denote the $L_2$-norm as $\|g\|_{L_2} = \sqrt{\E[g^2(X,Y,Z)]}$. 

\noindent\textbf{Data and Code.}
{Codes for reproducing our simulation and real data experiments can be found at \texttt{https://github.com/Albert-zcc/MOPI}.}

\section{Minimax Optimization for Conditional Coverage}\label{sec:MOPI}

\subsection{Problem setup}\label{sec:set_up}
Let $(X,Y)\in\mathcal{X}\times\mathcal{Y}$ be a data pair drawn from an unknown distribution, and let $C(\cdot):\mathcal{X}\to 2^{\mathcal{Y}}$ denote a measurable set-valued function that assigns to each covariate $X$ a \emph{prediction set} $C(X)\subseteq\mathcal{Y}$. In particular, given a conditioning variable $Z \in \mathcal{Z}$, our goal is to find a prediction set $C(X)$ satisfying the conditional coverage property in \eqref{eq:Z_conditional_cover}.
The conditioning variable $Z$ typically depends on $(X, Y)$ and determines the scope of coverage enforcement, representing various forms of contextual, structural, or latent information about the data. Below, we provide several representative examples with different types of $Z$.

\begin{example}[Test-conditional coverage]\label{exam:test_conditional}
The most popular form of conditional coverage validity arises when $Z=X$, which is also referred to as test-conditional coverage in the conformal prediction literature \citep{angelopoulos2024theoretical}. In this case, the target \eqref{eq:Z_conditional_cover} reduces to:
$\mathbb{P}\{Y\in C(X)\mid X=x\}=1-\alpha, \quad \forall x\in\mathcal{X}.$
Another practical case is that $Z$ is a part of the covariate $X$, representing features of interest.
\end{example}

\begin{example}[Group-conditional coverage]\label{exam:group_conditional}
Let $\{G_1,\ldots,G_K\}$ be a finite collection of groups, where each $G_k\subseteq\mathcal{X}$ defines a subset of the covariate space. The groups can be disjoint or overlap. The goal of group-conditional coverage \citep{jung2022batch,gibbs2023conformal} is to ensure coverage within each group: $\mathbb{P}\{Y\in C(X)\mid X\in G_k\}= 1-\alpha, \quad \forall 1\leq k\leq K,$
which is a special case of the target \eqref{eq:Z_conditional_cover} by letting $Z=(\one\{X\in G_1\},\ldots,\one\{X\in G_K\})^{\top}$.
\end{example}

\begin{example}[Equalized coverage on sensitive attributes]\label{exam:equal_coverage}
    In the context of fairness, the conditioning variable $Z$ can represent sensitive attributes, and we require coverage parity of the prediction set $C(X)$ across different values of the sensitive attribute $Z$ \citep{Romano2020With}.
    In particular, the target \eqref{eq:Z_conditional_cover} allows the case where the variable $Z$ is masked in the test data (the prediction set depends only on the covariate $X$), which aligns with the requirement of algorithmic fairness literature \citep{zafar2017fairness, zhang2018mitigating}.
\end{example}

\subsection{A minimax optimization framework}\label{sec:unify_minimax}


Given a set-valued function $C$, recall from \eqref{eq:cond_to_marg} that the conditional coverage \eqref{eq:Z_conditional_cover} is equivalent to requiring $\Phi(C, f) :=\E[f(Z)(\mathbbm{1}\{Y\notin C(X)\} - \alpha)] = 0$ for all measurable $f$. Then, when the exact conditional coverage holds, and the weight function class $\gF$ is symmetric, we have $\max_{f\in \gF}\Phi(C, f) = 0$. This motivates us to construct the desired prediction set by solving the minimax problem: $\min_{C \in \mathfrak{C}} \max_{f \in \mathcal{F}} \mathbb{E} [f(Z)(\mathbbm{1}\{Y \notin C(X)\} - \alpha)]$,
where $\fC$ is the general class of set-valued functions defined in (\ref{eq:structure_set}). 

To make the optimization problem well-defined, the inner maximization domain (i.e., the function class $\gF$) must typically be bounded \citep{nemirovski2004prox}. To remove the dependence on the scale of $f$, we may add an $L_2$ norm constraint $\|f\|_{L_2} \leq 1$ when taking the maximum on $\Phi(C, f)$ over $f\in \gF$. To further facilitate optimization, we consider the following $L_2$ penalized objective:
\begin{align}\label{eq:obj_L2_penalty}
    \Psi(C, f):= \E\LRm{f(Z)(\mathbbm{1}\{Y\notin C(X)\} - \alpha) - f^2(Z)},
\end{align}
and define the {\it minimax optimization predictive inference (MOPI)} set as the solution
\begin{align}\label{eq:CP_minimax_general}
    C^*=\argmin_{C\in \mathfrak{C}}\max_{f\in \gF} \Psi(C, f).
\end{align}
The $L_2$ penalty makes $\Psi(C, f)$ strongly concave in $f$, and the corresponding minimax problem is easier to tackle; see \citet{lin2020gradient} and \citet{rafique2022weakly}.


\begin{remark}
    In the previous works such as \citet{dikkala2020minimax}, the objective function includes a tuning parameter $\lambda > 0$ for the $L_2$ penalty, taking the form $\Psi_{\lambda}(C, f) := \E\LRm{f(Z)(\mathbbm{1}\{Y\notin C(X)\} - \alpha) - \lambda f^2(Z)}$ in our case. 
   In fact, under the mild condition that $\lambda f\in \gF$ for all $f\in \gF$, the choice of $\lambda$ does not affect the resulting optimizer. This can be seen from the relation: $\lambda\max_{f\in \gF} \Psi_{\lambda}(C, f) = \max_{f\in \gF}\LRl{\lambda f(Z)(\mathbbm{1}\{Y\notin C(X)\} - \alpha) - \lambda^2 f^2(Z)} = \max_{f\in \gF} \Psi(C, f).$
    This allows us to set $\lambda=1$ without additional tuning. See Appendix~A.1 for detailed discussions.
\end{remark}

\subsection{Mean squared coverage error of minimax solution}\label{sec:MSCE}
Given any set-valued function $C$, we denote by $\alpha(Z;{C}) := \mathbb{P}\{Y\notin {C}(X)\mid Z\}$ the conditional miscoverage probability and recall the definition: $\mathsf{MSCE}(C) = \E\!\LRm{\LRs{\alpha(Z;{C}) - \alpha}^2}.$ 
This metric quantifies the average conditional coverage error. Next, we establish the theoretical connection between the objective (\ref{eq:obj_L2_penalty}) and the $\mathsf{MSCE}(C)$, providing a rigorous justification for statistical optimality of MOPI.

In fact, note that the conditional coverage \eqref{eq:Z_conditional_cover} can also be pursued by minimizing the \emph{oracle} least-squares problem
\begin{align}\label{eq:oracle_LS_general}
    C^{\rm ora} = \arg\min_{C\in\mathfrak{C}}\,\mathsf{MSCE}(C).
\end{align}
The term ``oracle'' indicates that the conditional miscoverage probability $\alpha(Z;C)$ cannot be computed if the conditional distribution is unknown. The oracle set-valued function $C^{\rm ora}$ is the optimal element within the class $\fC$, achieving the smallest MSCE, and the $\mathsf{MSCE}(C^{\rm ora})$ can be interpreted as the approximation error of the class $\fC$.
The next proposition connects the MSCEs of the minimax solution $C^{*}$ in \eqref{eq:CP_minimax_general} and the oracle solution $C^{\rm ora}$.

\begin{proposition}\label{pro:CP_LS2minimax_general}
    {If $(\alpha(\cdot;C)-\alpha)/2 \in \gF$ for any $C\in \mathfrak{C}$, we have $\max_{f\in \gF}\Psi(C,f) = \mathsf{MSCE}(C)/4$ for any $C\in \fC$ and $\mathsf{MSCE}(C^*) = \mathsf{MSCE}(C^{\rm ora})$.}
\end{proposition}
This result implies that if  $\gF$ is expressive enough to include the scaled miscoverage error function $(\alpha(\cdot;C)-\alpha)/2$, then solving the minimax problem \eqref{eq:CP_minimax_general} is equivalent to solving  \eqref{eq:oracle_LS_general}.
Proposition \ref{pro:CP_LS2minimax_general} also sheds light on a principled way to choose the function class $\gF$ based on the domain of the conditioning variable $Z$.





\begin{itemize}
    \item \textit{Finite-dimensional $\gF$.}
    When the domain $\gZ$ contains only finite elements, i.e., $|\gZ| < \infty$, the function $\alpha(\cdot;C)$ is a $|\gZ|$-dimensional function. To meet the condition of Proposition \ref{pro:CP_LS2minimax_general}, we can take a parametric function class $\gF = \{\sum_{z\in \gZ} \beta_z\mathbbm{1}\{Z = z\}: \beta_z \in \sR, z\in \gZ\}$. For the group-conditional coverage in Example \ref{exam:group_conditional}, it can be seen that $|\gZ| = K$ for the disjoint groups and $|\gZ| \leq 2^K$ for the overlap groups. In Example \ref{exam:equal_coverage}, the parametric class $\gF$ remains applicable when the sensitive attribute $Z$ is a categorical variable.

    \item \textit{Infinite-dimensional $\gF$.}
    When the domain $\gZ$ contains infinite points, e.g., $\gZ$ is a subset of the Euclidean space, we can choose $\gF$ as a reproducing kernel Hilbert space (RKHS). This choice is suitable for the test-conditional coverage in Example \ref{exam:test_conditional} and the continuous sensitive variable in Example \ref{exam:equal_coverage}.
\end{itemize}
By adjusting the weight function class $\mathcal{F}$ to the conditioning variable $Z$, MOPI can be tailored to enforce local or groupwise coverage, as well as fairness constraints.

\subsection{Examples of structured set-valued function class}\label{sec:shape_set}

Structured set-valued function class $\fC$ in \eqref{eq:structure_set} provides a flexible and expressive formulation of prediction sets, and can meet most downstream requirements. Now, we present two common examples of $\fC$.

\noindent\textbf{Sublevel set of pretrained score function.} This class is widely used in the conformal prediction literature~\citep{vovk2005algorithmic,lei2018distribution}. Given a pretrained score function $s(x,y)$, the sublevel prediction set takes the form
\begin{align}\label{eq:pretrain_level_set}
    \fC_{\rm sublevel} = \LRl{ C(x;h) = \{\, y \in \mathcal{Y} : s(x,y) \le h(x) \,\}: h \in \mathcal{H}},
\end{align}
where $h:\mathcal{X}\to\mathbb{R}$ is a one-dimensional threshold function. Let $T(h(x), y) = s(x,y) - h(x)$ in \eqref{eq:structure_set}, then $\fC_{\rm sublevel}$ is a special case. In this formulation, the score $s$ fixes the geometric structure of $C(x;h)$, leaving only the threshold $h(x)$ to be calibrated.

\begin{remark}
For test-conditional coverage with $Z = X$, the population quantile regression in the \texttt{CC} method sets $\mathcal{H} = \mathcal{F}$ in \eqref{eq:pretrain_level_set}, yielding the prediction set $C^{\text{qr}}(X) := \{y \in \mathcal{Y} : s(X,y) \leq f^{\text{qr}}(X)\}$. Here, $f^{\text{qr}} = \arg\min_{f \in \mathcal{F}} \mathbb{E}[\ell_{\alpha}(s(X,Y),f(X))]$ is the minimizer of the expected pinball loss. 
In addition, let $q(x) = \inf\{t\in \sR: \sP(s(X,Y) \leq t \mid X=x) \geq 1-\alpha\}$ be the ground truth conditional quantile function. If $q \in \gF$, the prediction set $C^{\rm qr}(X)$ achieves zero MSCE, as well as for the minimax solution $C^*$ under the condition of Proposition \ref{pro:CP_LS2minimax_general}. If $q \notin \gF$, the first conclusion of Proposition \ref{pro:CP_LS2minimax_general} guarantees that $C^*(X)$ achieves lower MSCE than $C^{\rm qr}(X)$. Hence, MOPI exhibits better performance in terms of MSCE. 
\end{remark}

\noindent\textbf{Geometrically structured prediction sets.} In some downstream tasks, e.g., the robust optimization problems \citep{ben2009robust, johnstone2021conformal}, the geometric structure of prediction sets in the multi-dimensional label space $\gY \subseteq \sR^{d_{\gY}}$ is critical to the decision making. Let us consider the following two commonly used geometrically set-valued classes:
\begin{itemize}
    \item [(i)] Let $h = (\mu, \sigma)$ with mean function $\mu:\gX \to \sR^{d_{\gY}}$ and variance function $\sigma: \gX \to \sR^{d_{\gY}}$, the box set-valued function class takes the form 
    $\fC_{\rm box} = \Big\{ C(x;h) = \left\{ y\in \sR^{d_{\gY}} : \|(y - \mu(x))/\sigma(x)\|_{\infty}\le 1\right\}: h\in \gH \Big\}$,
    where “$/ $” denotes element-wise division of vectors.
    By setting $T(h(x), y) = \|(y - \mu(x))/\sigma(x)\|_{\infty} - 1$ in the definition \eqref{eq:structure_set}, we recover the box class $\fC_{\rm box}$.
    
    \item [(ii)] Let $h = (\mu, \Sigma)$ with mean function $\mu:\gX \to \sR^{d_{\gY}}$ and covariance function $\Sigma: \gX \to \sR^{d_{\gY} \times d_{\gY}}$, the ellipsoidal set-valued function class takes the form
    $ \fC_{\rm ell} = \Big\{ C(x;h)
            = \big\{ y \in \mathbb{R}^{d_{\gY}} : (y - \mu(x))^\top \Sigma^{-1}(x) (y - \mu(x)) \le 1 \big\}
            : h \in \gH \Big\}$.
    This corresponds to the choice $T(h(x), y) = (y - \mu(x))^\top \Sigma^{-1}(x) (y - \mu(x)) - 1$ in the definition \eqref{eq:structure_set}.
\end{itemize}

{\begin{remark}
    It is also possible to construct a sublevel prediction set of the form \eqref{eq:pretrain_level_set} with geometric structures described above using a pretrained score. For example, \citet{johnstone2021conformal} and \citet{sun2023predict} built an ellipsoidal sublevel set via the score $s(x,y) = (y-\mu_0(x))^{\top}\Sigma_0^{-1}(y-\mu_0(x))$, where $\mu_0(x)$ is a regression model and $\Sigma_0$ is the covariance matrix obtained on a separate dataset. However, given such a score function, the \texttt{CC} method cannot further adjust the covariance matrix with respect to the conditioning variable. To embed the pretrained information into a more flexible structure, we consider setting $T(h(x), y) = \tilde{y}^{\top} \Sigma^{-1}(x)\tilde{y}$ in \eqref{eq:structure_set} where $\tilde{y} = \Sigma_0^{-1/2}(y-\mu_0(x))$ and $h(x) = \Sigma(x)$. This enables MOPI to adaptively refine the shape model $\Sigma(x)$ during calibration, leading to more accurate conditional coverage.
\end{remark}}

\section{Empirical Minimax Optimization Problem}\label{sec:empirical_opt}

Suppose we have i.i.d. labeled data $\gD_{n}:=\{(X_i, Y_i, Z_i)\}_{i=1}^{n}$. The goal of this section is to construct MOPI sets by solving the empirical version of the problem \eqref{eq:CP_minimax_general}.

\subsection{Empirical minimax problem}\label{sec:empirical_MOPI}

Formally, we consider the following sample averaging approximation to the population objective function $\Psi(C,f)$ defined in \eqref{eq:obj_L2_penalty}, 
\begin{align}
    \widehat{\Psi}(C, f) := \frac{1}{n}\sum_{i=1}^n \LRm{f(Z_i) \LRs{\mathbbm{1}\{Y_i\notin C(X_i)\} - \alpha} - f^2(Z_i)}.\nonumber
\end{align}
We assume the set-valued function class $\fC$ is finite-dimensional, and defer the infinite-dimensional case to Appendix D.
Suppose the weight function class $\gF$ is equipped with a norm $\|\cdot\|_{\gF}$, which measures the complexity of $\gF$.
Accordingly, we output the empirical MOPI prediction set by solving the minimax optimization:
\begin{align}\label{eq:sample_minimax_CP}
    \widehat{C} = \argmin_{C\in \fC}\max_{f\in \gF} \LRl{\widehat{\Psi}(C, f)  - \gamma\|f\|_{\gF}^2},
\end{align}
where $\gamma\ge 0$ acts as a regularization parameter when $\gF$ is an infinite-dimensional function class. We can simply set $\gamma = 0$ when $\gF$ is finite-dimensional.

In Section \ref{sec:MSCE}, we discussed how to choose the weight function class $\gF$ according to the cardinality of the conditioning domain $\gZ$. Now, we show that the inner maximization of the empirical problem \eqref{eq:sample_minimax_CP} admits a closed form under two practical choices of $\gF$.

\begin{lemma}\label{lemma:inner_close_form_finite}
    Let $\gF = \{\sum_{z\in \gZ} \beta_z\mathbbm{1}\{Z = z\}: \beta_z \in \sR, z\in \gZ\}$ when $|\gZ| < \infty$. The inner maximization in the problem \eqref{eq:sample_minimax_CP} with $\gamma = 0$ is:
    \begin{align}
    \max_{f\in \gF}\widehat{\Psi}(C,f) = \sum_{z\in\gZ}\frac{\big(\sum_{i=1}^n\mathbbm{1}\{Z_i=z\}\LRs{\mathbbm{1}\{Y_i\notin C(X_i)\}-\alpha}\big)^2}{4 n\sum_{i=1}^n\mathbbm{1}\{Z_i=z\}}.\nonumber
    \end{align}
\end{lemma}

\begin{lemma}\label{lemma:inner_close_form_rkhs}
    Let $\gF$ be a RKHS equipped with a kernel function $\gK(\cdot,\cdot):\gZ \times \gZ \to \sR$. Denote $\boldsymbol{\varphi}_n(C) = (\varphi_1(C),\ldots,\varphi_n(C))^\top$ where $\varphi_i(C) = \frac{1}{n}\bigl( \mathbbm{1}\{ Y_i\notin C(X_i) \} - \alpha \bigr)$ for $i=1,\ldots,n$ and $\mK_n = (\gK(Z_i,Z_j))_{i,j=1}^n$ be the empirical kernel matrix. Let $\mI_n$ be an $n\times n$ identity matrix. The inner maximization in the problem \eqref{eq:sample_minimax_CP} is:
    \begin{align}
        \max_{f\in \gF} \LRl{\widehat{\Psi}(C, f)  - \gamma\|f\|_{\gF}^2} &= 
       \frac{1}{4}\boldsymbol{\varphi}_n(C)^\top \mK_n \left( \frac{1}{n} \mK_n + \gamma \mI_n \right)^{-1} \boldsymbol{\varphi}_n(C).\nonumber
    \end{align}
\end{lemma}

Using the two lemmas above, we transform the corresponding empirical minimax problems \eqref{eq:sample_minimax_CP} into empirical risk \emph{minimization} problems. 

\subsection{Implementation of MOPI by smoothing}
\label{sec:smooth_MOPI}
Since the objective function $\widehat{\Psi}(C, f)$ contains the non-differentiable miscoverage indicator $\mathbbm{1}\{Y_i\notin C(X_i)\} = \mathbbm{1}\{T(h(X_i), Y_i) > 0\}$ for $C(X_i) = \{y\in \gY: T(h(X_i),Y_i) \leq 0\}$ as defined in \eqref{eq:structure_set}, direct optimization is infeasible. To overcome this non-differentiability, we replace the indicator $\mathbbm{1}\{u> 0\}$ with a smooth surrogate, such as the Sigmoid function $\tilde{\mathbbm{1}}\{u> 0\} =  (1+\exp(-u/r))^{-1}$ or the Gaussian error function
$\tilde{\mathbbm{1}}\{u> 0\} = \tfrac{1}{2}\!\left(1+\mathrm{erf}\!\big(u/\sqrt{2}r\big)\right)$, where
$\text{erf}(x) = \tfrac{2}{\sqrt{\pi}} \int_{0}^x e^{-t^2}\,dt$ and $r>0$ is the smoothing parameter. The smoothing strategy has been widely used in recent optimization works on conformal prediction, see e.g. \citet{kiyani2024length}, and \citet{wuerror2025}.

For any $(x,y) \in \gX \times \gY$ and $C\in \fC$, we denote the smoothed miscoverage indicator as $\tilde{\mathbbm{1}}\{y\notin C(x)\}:=\tilde{\mathbbm{1}}\{T(h(x),y)>0\}$. Then we have the smoothed objective: $\widetilde{\Psi}(C, f) := \frac{1}{n}\sum_{i=1}^n \LRm{f(Z_i) \LRs{\tilde{\mathbbm{1}}\{Y_i\notin C(X_i)\} - \alpha} - f^2(Z_i)}.$
Accordingly, the smoothed empirical MOPI set can be obtained by solving
\begin{align}\label{eq:sample_minimax_CP_smooth_finite}
     \widetilde{C} := \argmin_{C\in \fC}\max_{f\in \gF} \LRl{\widetilde{\Psi}(C, f)  - \gamma\|f\|_{\gF}^2}.
\end{align}
If $\gF$ is a parametric class or RKHS, we can also obtain the closed form of $\max_{f\in \gF} \LRl{\widetilde{\Psi}(C, f)  - \gamma\|f\|_{\gF}^2}$ by replacing miscoverage indicators with the smoothed ones. We can then solve the resulting empirical risk minimization problem using gradient-based methods. In Appendix B.2, we discuss and compare the computational efficiency of MOPI via the smoothed minimax optimization \eqref{eq:sample_minimax_CP_smooth_finite} and the quantile regression in the \texttt{CC} method \citep{gibbs2023conformal}.

\section{Theoretical Results}
In this section, we first establish an upper bound for the MSCE of the empirical MOPI set \eqref{eq:sample_minimax_CP}, and then derive specific convergence rates for examples discussed in Section \ref{sec:set_up}. After that, we analyze the set solution to the smoothed problem \eqref{eq:sample_minimax_CP_smooth_finite}. 

\subsection{Oracle inequality for MOPI}\label{sec:theory_finite_C}

We begin with the following two assumptions on the weight function class $\gF$.

\begin{assumption}\label{assum:shape_F}
    There exists a constant $U > 0$ such that $\sup_{f\in \gF_U}\sup_{z\in \gZ}|f(z)| \leq 1$, where $\gF_U := \{ f \in \gF : \|f\|_{\gF}^2 \le U \}$. In addition, $\gF_U$ is star-shaped around zero, i.e., $t f \in \gF_U$ whenever $f\in \gF_U$ for any $t\in [0,1]$.
\end{assumption}
\begin{assumption}\label{assum:appro_error}
    For any $C\in \fC$, it holds that $\E[(f_C(Z) - \{\alpha(Z;C)-\alpha(Z;C^{\rm ora})\})^2] \le \eta^2$, where $C^{\rm ora}$ is the minimizer of MSCE over $\fC$ defined in \eqref{eq:oracle_LS_general} and $f_C := \argmin_{f \in \gF_U}\E\LRm{\big(f(Z) - \{\alpha(Z;C)-\alpha(Z;C^{\rm ora})\}\big)^2}.$
\end{assumption}

In Assumption \ref{assum:shape_F}, $\gF_U$ is a ball with radius $U$ in the function class $\gF$, and the conditions are quite mild and can be satisfied for both a parametric function class and an RKHS. {In Assumption \ref{assum:appro_error}, the function $f_C(\cdot)$ is a projection of the conditional miscoverage probability difference $\alpha(\cdot;C)-\alpha(\cdot;C^{\rm ora})$ onto the function class $\gF$, and $\eta$ is an upper bound for the projection error.} Under the assumption of Proposition \ref{pro:CP_LS2minimax_general}: $\alpha(\cdot;C) - \alpha \in \gF$ for any $C\in \fC$, we have $\eta = 0$; see corollaries in Section \ref{sec:rate_VC_class}.

To characterize the excess risk of the empirical objective relative to the population objective, we introduce the function class: $\gG := \Big\{(x,y,z) \mapsto f_C(z)[\one\{y\notin C(x)\} - \one\{y\notin C^{\rm ora}(x)\}] : C \in \fC\Big\}$,
which depends on the set-valued class $\fC$.
Given $\delta > 0$, we define the localized Rademacher complexity \citep{Bartlett2005local} of the function class $\gG$ as $\overline{\mathcal{R}}_n(\delta;\gG) := \E\left[\sup_{\substack{g \in \gG,\|g\|_{L_2} \le \delta}} |n^{-1}\sum_{i=1}^n \varepsilon_i g(X_i,Y_i,Z_i)|\right]$, where $\{\varepsilon_i\}_{i=1}^n$ are i.i.d. Rademacher random variables taking values in $\{-1,+1\}$ with equal probability.
The \emph{critical radius} of $\gG$ is any positive solution to $\overline{\gR}_n(\delta;\gF_U) \le \delta^2$. The critical radius of $\gF_U$ is defined similarly. The next theorem gives a non-asymptotic upper bound for the MSCE of the empirical MOPI predictor in \eqref{eq:sample_minimax_CP}.

\begin{theorem}[Oracle inequality]\label{thm:mse}
    Suppose Assumptions \ref{assum:shape_F}-\ref{assum:appro_error} hold. Let $\delta_{n,\gF_U}$ and $\delta_{n,\gG}$ be the critical radius of $\gF_U$ and $\gG$, respectively. 
    (i) For finite-dimensional $\gF$, we choose the hyperparameter $\gamma = 0$ in \eqref{eq:sample_minimax_CP}; (ii) For infinite-dimensional $\gF$, we choose the hyperparameter $\gamma$ in \eqref{eq:sample_minimax_CP} such that $\gamma \|f\|_\gF^2 \leq \|f\|_{L_2}^2$. For any $\zeta \in (0,1)$, if $\LRs{\delta_{n,\gF_{U}} + \sqrt{\frac{\log(1 / \zeta)}{n}}} \|f\|_\gF \lesssim \sqrt{U} \|f\|_{L_2}$ holds all $f \in \gF$, with probability $1 - 4\zeta$:
        \begin{align}
            \E\LRm{\LRs{\alpha(Z; \widehat{C}) - \alpha}^2 \mid \gD_n}\lesssim \mathsf{MSCE}(C^{\rm ora})+ \eta^2+\delta_{n,\gF_{U}}^2+\delta_{n,\gG}^2 + {\frac{\log({\log_2(1/\delta_{n,\gG})}/\zeta)}{n}}.\nonumber
        \end{align}
\end{theorem}

In general, the approximation error $\mathsf{MSCE}(C^{\rm ora})$ is small if the capacity of $\mathfrak{C}$ is large enough. Moreover, the projection error $\eta$ can be negligible when $\gF$ is expressive enough. The critical values $\delta_{n,\gF_U}$ and $\delta_{n,\gG}$ decrease as the sample size $n$ grows, and their scales also depend on the complexities of $\gF$ and $\fC$.

\subsection{Upper bound for the approximation error}
Next, we investigate how the relationship between the covariate $X$ and the conditioning variable $Z$ influences the approximation error. 

\begin{assumption}\label{assum:oracle_exist}
    Suppose $\gH$ is a class of functions that map from $\gX$ to $\sR^m$.
    (i) There exist a measurable function $H^0(\cdot): \gZ \to \sR^m$ such that $\sP\{T(H^0(Z), Y) \leq 0\mid Z = z\}=1-\alpha$ for any $z\in \gZ$;
    {(ii)} $h^0(x) :=\E[H^0(Z) \mid X = x]$ belongs to $\gH$. (iii) There exists $\kappa > 0$ such that $\LRabs{\sP\{T(u_1,Y) \leq 0\mid X,Z\}-\sP\{T(u_2,Y) \leq 0\mid X,Z\}}\le \kappa \|u_1-u_2\|$ for any $u_1, u_2\in\R^m$.
\end{assumption}

    

Assumption~\ref{assum:oracle_exist} (i) can be viewed as an existence condition on the geometric model for achieving exact conditional coverage. Assumption~\ref{assum:oracle_exist} (ii) further requires that the conditional mean function $h^0(X)$ lies in $\gH$, which is satisfied whenever $\mathcal{H}$ is sufficiently rich. Let us consider the group-conditional coverage in Example~\ref{exam:group_conditional} and assume $\fC$ is the sublevel prediction set~\eqref{eq:pretrain_level_set} of the pretrained score $s$. Recalling that $Z = (\one\{X\in G_1\},\ldots,\one\{X\in G_K\})^{\top}$.
Let $q_z=\inf\{t:\sP(s(X,Y)\le t\mid Z=z)\ge 1-\alpha\}$ is the $1-\alpha$ conditional quantile given $Z=z$ and $H^0(Z)=\sum_{z\in\gZ} q_z \one\{Z=z\}$. 
Then, we can guarantee $h^0\in\gH$ by choosing $\gH=\{\sum_{z\in\gZ}\beta_z\mathbbm{1}\{Z=z\}:\beta_z\in\sR,z\in \gZ\}$ since $h^0(X)=H^0(Z)$. Assumption~\ref{assum:oracle_exist} (iii) requires Lipschitz continuity of conditional coverage probability on the geometric parameter.

{\begin{theorem}[Approximation error]\label{thm:LS2_projX}
    Under Assumption \ref{assum:oracle_exist} (i) and (ii), if $H^0(Z)$ is measurable with respective to the $\sigma$-algebra generated by $X$, then $\mathsf{MSCE}(C^{\rm ora}) = 0$.
    Furthermore, let $\rho^2= \frac{\mathsf{Tr}\LRs{\mathsf{Cov}\{h^0(X)\}}}{\mathsf{Tr}\LRs{\mathsf{Cov}\{H^0(Z)\}}}$, if Assumption \ref{assum:oracle_exist} (iii) also holds, then $\mathsf{MSCE}(C^{\rm ora})
        \le \kappa^2\big(1-\rho^2\big)\,\mathsf{Tr}\LRs{\mathsf{Cov}\!\left\{H^0(Z)\right\}}.$
\end{theorem}}

{Theorem \ref{thm:LS2_projX} shows that the approximation error of MOPI vanishes when $H^0(Z)$ is predictable with respect to the covariate $X$, covering both the test-conditional and group-conditional settings in Examples \ref{exam:test_conditional} and \ref{exam:group_conditional}. Otherwise, the approximation error is bounded by the variation in $H^0(Z)$ that cannot be explained by $h^0(X)$. The scalar $\rho$ summarizes this effect: when $\rho$ is close to one, the projection $h^0(X)$ captures most of the variability in $H^0(Z)$, leading to a small approximation error. As a special case, if $Z$ is independent of $(X,Y)$, then the conditional coverage $\sP\{Y\in C(X)\mid Z\}$ reduces to marginal coverage $\sP\{Y\in C(X)\}$, and we can be take $H^0(z)$ as a constant function, then Theorem~\ref{thm:LS2_projX} implies $\mathsf{MSCE}(C^{\rm ora})=0$.
}

\subsection{Convergence rates of MSCE}\label{sec:rate_VC_class}
In what follows, we derive the specific convergence rates of the MSCE under the following Vapnik–Chervonenkis (VC) complexity assumption on the set-valued class $\fC$, which is satisfied by basic linear models \citep{wainwright2019high} and deep neural networks \citep{bartlett2019nearly}.
{For infinite-dimensional class $\fC$, since VC-based arguments are no longer applicable, we establish the corresponding convergence rate using a different technique; the details are deferred to Appendix~D.}



\begin{assumption}\label{assum:VC_class_C}
    The function class $\{(x,y)\mapsto T(h(x), y): h\in\mathcal{H}\}$ has a VC-dimension $d_{\fC}$, where $T$ is defined in \eqref{eq:structure_set}.
\end{assumption}

%
In the following analysis, we upper bound the MSCE when the domain $\gZ$ is finite. In this case, one can consider the parametric function class $\gF$, which naturally satisfies Assumption \ref{assum:shape_F}. In addition, given any $C \in \fC$, the function $\alpha(z;C)$ takes at most $|\gZ|$ distinct values. Consequently, $\alpha(z;C)\in\gF$ and $\eta=0$ for Assumption \ref{assum:appro_error}.

\begin{theorem}\label{thm:mse_finite}
    Suppose $|\gZ|< \infty$ and Assumption \ref{assum:VC_class_C} hold for the set-valued class $\fC$ in \eqref{eq:structure_set}.
    Choosing $\gF =\{\sum_{z\in \gZ} \beta_z\mathbbm{1}\{Z = z\}: \beta_z \in \sR,\forall z\in \gZ\}$ and setting $\gamma=0$ in the optimization problem \eqref{eq:sample_minimax_CP}, then
    \[
    \mathsf{MSCE}(\widehat{C}) \lesssim \mathsf{MSCE}(C^{\rm ora}) +{\frac{d_{\fC}+|\gZ|}{n}}\log\LRs{\frac{n}{d_{\fC}+|\gZ|}}.
    \]
\end{theorem}



In particular, we can obtain the group-conditional coverage guarantee in Example \ref{exam:group_conditional} based on Theorem \ref{thm:mse_finite}. Let $\{G_1, \ldots, G_K\}$ be a collection of groups in $\gX$, and we recall the conditioning variable $Z = (\one\{X\in G_1\},\ldots,\one\{X\in G_K\})^{\top}$. The choice of $Z$ implies that $\rho = 1$ in Theorem \ref{thm:LS2_projX}, and hence $\mathsf{MSCE}(C^{\rm ora}) = 0$. By choosing the function class $\gH = \gF$ in the set-valued class $\fC$, we have the following corollary.


\begin{corollary}\label{cor:mse_finite_group_X}
    Considering Example \ref{exam:group_conditional}, we let $Z = (\one\{X\in G_1\},\ldots,\one\{X\in G_K\})^{\top}$ and choose $\gH = \gF = \{\sum_{z\in \gZ} \beta_z\mathbbm{1}\{Z = z\}: \beta_z \in \sR,\forall z\in \gZ\}$. Suppose the condition (i) in Assumption \ref{assum:oracle_exist} holds, then for any $k\in [K]$,
    \[
    \left|\sP\LRl{Y\in\widehat{C}(X)\mid X\in G_k}-(1-\alpha)\right| \lesssim \sqrt{\frac{d_{\fC}+|\gZ|}{n\cdot \sP(X\in G_k)}\log\LRs{\frac{n}{d_{\fC}+|\gZ|}}}.
    \]
\end{corollary}

{For the sublevel set class \eqref{eq:pretrain_level_set} under the choice of $\gH$ in Corollary \ref{cor:mse_finite_group_X}, we have $d_\fC=O(|\gZ|)$, where $|\gZ| = K$ for disjoint groups and $|\gZ| \leq 2^K$ for overlap groups. And our upper bound significantly improves the dependence on sample size from $n^{-1/4}$ in \citet{jung2022batch} to $n^{-1/2}$, and matches the optimal convergence rate in \citet{areces2024two} up to a logarithmic factor.}

Next, we provide equalized coverage guarantee in Example \ref{exam:equal_coverage}.

\begin{corollary}\label{cor:mse_finite_equal}
    Considering Example \ref{exam:equal_coverage} with $|\gZ| < \infty$, and suppose Assumption \ref{assum:oracle_exist} holds. Let
    $\Delta_{H^0}:=\max_{z,z'\in\gZ}\|H^0(z)-H^0(z')\|$ for $H^0$ defined in Assumption \ref{assum:oracle_exist} and {$\pi_z(X) = \sP(Z = z\mid X)$}, then for any $z\in \gZ$,
    \[
    \left|\sP\LRl{Y\in\widehat{C}(X)\mid Z=z}-(1-\alpha)\right| \lesssim \kappa\sqrt{\frac{\Delta_{H^0}^2\sum_{z\in\gZ}\E[\pi_z(X)(1- \pi_z(X))]}{\sP(Z=z)}}+\sqrt{\frac{d_{\fC}+|\gZ| }{n\cdot \sP(Z=z)}\log\LRs{\frac{n}{d_\fC+|\gZ|}}}.
    \]
\end{corollary}

In Corollary \ref{cor:mse_finite_equal}, the quantity $\Delta_{H^0}^2$ measures the fluctuation of the oracle threshold function $H^0(z)$, while $\sum_{z\in\gZ}\E[\pi_z(X)(1- \pi_z(X))]$ is Gini impurity of the conditioning variable $Z$ given the covariate $X$. Therefore, the first term in the coverage bound is near to zero in two scenarios: (i) $\Delta_{H^0}^2\approx 0$ (i.e., $H^0(Z)$ is approximately constant), which implies that $Z$ barely influences the coverage probability by the definition in Assumption \ref{assum:oracle_exist}; or (ii) the conditional distribution $Z\mid X$ is nearly degenerate on some category $z^{\prime}\in \gZ$ (i.e., $\pi_{z^{\prime}}(X) \approx 1$ for some $z^{\prime} \in \gZ$).



 We now establish upper bounds for the MSCE when the domain $\gZ$ is infinite. For simplicity, let us consider a bounded Euclidean domain $\gZ = [0,1]^{d_{\gZ}}$, where $d_{\gZ}$ is the dimension. In this case, we choose $\gF$ as the RKHS equipped with a Gaussian kernel $\gK(z,z^\prime)=\exp\LRs{-\frac{1}{2\vartheta^2}\|z-z^\prime\|^2}$ with the bandwidth $\vartheta$. Then Assumption \ref{assum:shape_F} holds with $U=1$ since $\sup_{z\in \gZ}|f(z)|\le \sup_{z\in \gZ}\sqrt{\gK(z,z)}\|f\|_\gF\le 1$ for any $f\in\gF_U$.


\begin{theorem}\label{thm:mse_rkhs}  
     Suppose Assumptions \ref{assum:appro_error} and \ref{assum:VC_class_C} hold for $\fC$ in \eqref{eq:structure_set}.
    If the hyperparameter $\gamma$ in optimization problem \eqref{eq:sample_minimax_CP} is chosen such that 
    $\gamma\|f\|^2_{\gF}\le\|f\|^2_{L_2}$ for all $f\in\gF$, then
    \[
    \mathsf{MSCE}(\widehat{C}) \lesssim \mathsf{MSCE}(C^{\rm ora})+\eta^2+\frac{d_{\fC}\log(n/d_{\fC})}{n}+ \frac{d_{\gZ}\log^{d_{\gZ}+1}(n/d_{\gZ})}{n}.
    \]
\end{theorem}



{In conjunction with Theorem \ref{thm:LS2_projX}, if Assumptions \ref{assum:appro_error}, \ref{assum:oracle_exist}, and \ref{assum:VC_class_C} hold for the test-conditional coverage case described in Example \ref{exam:test_conditional} and $\alpha(X;C) - \alpha \in \gF_U$ with $U=1$ for any $C\in \fC$, we have $\mathsf{MSCE}(\widehat{C}) \lesssim \frac{d_{\fC}\log(n/d_{\fC})}{n}+ \frac{d_{\gX}\log^{d_{\gX}+1}(n/d_{\gX})}{n}$, where the first term matches the standard VC-class convergence rate in \cite{duchi2025sampleconditional}, and the second term agrees with the rate of Gaussian-kernel RKHS in \cite{wainwright2019high}.}



\subsection{Oracle inequality for MOPI via smoothed optimization}

Since the miscoverage indicator is replaced by the smooth surrogate in the objective \eqref{eq:sample_minimax_CP_smooth_finite}, we define a new function class as $\widetilde{\gG} := \big\{(x,y,z) \mapsto f_C(z)[\tilde{\mathbbm{1}}\{y\notin C(x)\} - \tilde{\mathbbm{1}}\{y\notin C^{\rm ora}(x)\}] : C \in \fC\big\}$ to get the bound of the corresponding excess risk. In the following, we focus on the Gaussian error function as the smooth surrogate: $\tilde{\mathbbm{1}}\{y\notin C(x)\} = \frac{1}{2}\LRs{1 + \operatorname{erf}\LRl{T(h(x), y)/\sqrt{2}r}}$ for $C(x) = \{y\in \gY: T(h(x), y) \leq 0\}$, where $r > 0$ is the smoothing parameter. 

The next theorem provides the oracle inequality for MOPI set in \eqref{eq:sample_minimax_CP_smooth_finite}.
Compared with Theorem \ref{thm:mse}, there will be an additional error regarding the smoothing bias, which depends on the scale of the smoothing parameter $r$.
\begin{theorem}\label{thm:mse_smooth_finite}
    Suppose Assumptions \ref{assum:shape_F}-\ref{assum:appro_error} hold. Assume the density of $T(h(X),Y)\mid Z$ is upper bounded for all $h\in\gH$, and $\sup_{f\in \gF}\sup_{z\in \gZ}|f(z)|$ is also upper bounded.
    Under the same choice of $\gamma$ in Theorem \ref{thm:mse},
    if $\LRs{\delta_{n,\gF_{U}} + \sqrt{\frac{\log(n)}{ n}}} \|f\|_\gF \lesssim \sqrt{U} \|f\|_{L_2}$ holds for any $f \in \gF$, then
    {\small\begin{align*}
       \mathsf{MSCE}(\widetilde{C})\lesssim\mathsf{MSCE}(C^{\rm ora}) + \eta^2+\delta_{n,\gF_{U}}^2+\delta_{n,\widetilde{\gG}}^2+\frac{r^2}{\LRs{\delta_{n,\gF_{U}}\!+ \!\sqrt{\log n/n}}^2} +{\frac{\log\LRs{n{\log_2 \delta_{n,\widetilde{\gG}}^{-1}}}}{n}}.
    \end{align*}}
    where $\delta_{n,\gF_U}$ and $\delta_{n,\widetilde{\gG}}$ are critical radius of $\gF_U$ and $\widetilde{\gG}$, respectively.
\end{theorem}
According to the result of Theorem~\ref{thm:mse_smooth_finite}, the MSCE will be dominated by complexities of $\gF$ and $\fC$ provided that the smoothing parameter is chosen such that $r \lesssim \min\{{\delta}_{n,\gF_U}^2,\log n/n\}$. Similar to Section~\ref{sec:rate_VC_class}, we can get specific coverage rates on MSCE for different choices of classes $\gF$ and $\fC$. In Appendix A.2, we conduct an empirical sensitivity analysis on the choice of $r$, indicating that the performance of MOPI is stable as long as $r$ is not too large.

\section{Simulation Results}\label{sec:experiments}
We conduct simulations on synthetic data to compare the performance of our proposed \texttt{MOPI} with split conformal prediction (\texttt{SCP}) in \citet{vovk2005algorithmic} and \citet{lei2018distribution}, {the conditional calibration (\texttt{CC}) in \citet{gibbs2023conformal}}, and randomly-localized conformal prediction (\texttt{RLCP}) in \citet{hore2024conformal}. 
For \texttt{MOPI}, we replace the indicator function $\mathbbm{1}\{u > 0\}$ with a sigmoid smoothing surrogate $\tilde{\mathbbm{1}}\{u > 0\} = (1+\exp(-u/r))^{-1}$, with the smoothing parameter set to $r=0.1$.
The nominal coverage level is set to $1-\alpha=90\%$, and all simulation results are averaged over $100$ replications. In each replication, we independently generate three datasets from the same distribution: pretraining set $\{({X}_i,{Y}_i,{Z}_i)\}_{i\in \gD_{\rm pre}}$, calibration set $\{(X_i,Y_i,Z_i)\}_{i\in \gD_{\rm cal}}$ and test set $\{(X_{i},Y_{i},Z_{i})\}_{i\in \gD_{\rm test}}$. We provide additional simulation results in Appendix~F.

\subsection{Geometrically structured sets for multi-dimensional label}\label{subsec:multi_Y}

For a multivariate label $Y$, we design a simulation setting involving cross-dimensional dependence, in which the strength of heterogeneity varies across dimensions.
Let $X \sim U(0,5)$, the multi-dimensional label is generated as follows: $Y\mid X \sim\mathrm{N}(\mu^*(X), \Sigma^*(X))$, where {$\mu^*_j(X) = \sin(X+j) +0.3j$} and $\Sigma_{jj}^*(X) = \left(0.05 + 1.5\sqrt{j} \cdot |\sin(X + 2j)|\right)^{\sqrt{j}}$ for $j=1,\ldots,d_{\gY}$ and $\Sigma_{12}^*(X) = \Sigma_{21}^*(X) = 0.6\sqrt{\Sigma_{11}^*(X)\Sigma_{22}^*(X)}$, while all other entries of $\Sigma^*(X)$ are zero. In this simulation, we set the sample size of test data as $|\gD_{\rm test}| = 10,000$.

In this experiment, we study test-conditional coverage ($Z=X$) for ellipsoidal prediction sets. The baseline methods (\texttt{SCP}, \texttt{CC}, and \texttt{RLCP}) construct sublevel sets \eqref{eq:pretrain_level_set} using the score $s(x,y)=(y-\mu_0(x))^{\top}\Sigma_0^{-1}(y-\mu_0(x))$, where $\mu_0(x)$ is a pretrained random forest and $\Sigma_0$ is the empirical covariance estimated from a pretraining set of size $|\gD_{\rm pre}|=1{,}500$. The thresholds are then calibrated on a separate calibration set.
To ensure a fair comparison while leveraging the same pretrained components, \texttt{MOPI} first normalizes $y$ as $\tilde{y}=\Sigma_0^{-1/2}(y-\mu_0(x))$, and then constructs structured sets $\fC_{\rm ell} = \{ C(x;\Sigma) = \{ y \in \mathbb{R}^{d_{\gY}} : \tilde{y}^\top \Sigma^{-1}(x)\tilde{y} \le 1 \} : \Sigma \in \gH \}$, where $\gH$ is a neural network class mapping covariates to covariance matrices.

For \texttt{MOPI} and \texttt{CC}, we choose the function class $\gF$ to be the RKHS on $\gX$ equipped with a Gaussian kernel. \texttt{RLCP} also employs the same Gaussian kernel for its localization. To evaluate different methods, we compute the following metrics:
\begin{itemize}
    \item[(a)] Marginal Coverage (\textit{Marginal}): $\frac{1}{|\gD_{\rm test}|}\sum_{i\in \gD_{\rm test}} \mathbbm{1}\{Y_{i} \in \widehat{C}(X_{i})\}$.
    \item[(b)] Root of MSCE ($\sqrt{{\rm MSCE}}$): $ \LRl{\frac{1}{|\gJ|}\sum_{J\in\gJ}\LRs{\frac{1}{n_J} \sum_{i\in \gD_{\rm test}} \mathbbm{1}\{X_{i} \in J, Y_{i} \notin \widehat{C}\left(X_{i}\right)\}-\alpha}^2}^{1/2}$, where $\gJ$ is the set of equidistant partition sets of $\gX$ and $n_J = \sum_{i\in \gD_{\rm test}} \mathbbm{1}\{X_{i} \in J\}$.
    \item[(c)] Worst-case conditional coverage (\textit{Worst-case}): $\min_{B \in \mathcal{B}} \frac{1}{n_B} \sum_{i\in \gD_{\rm test}} \mathbbm{1}\{X_{i} \in B, Y_{i} \in \widehat{C}\left(X_{i}\right)\}$, where $n_B=\sum_{i\in \gD_{\rm test}} \mathbbm{1}\{X_{i} \in B\}$, and $\mathcal{B}$ denotes a collection of balls in $\mathcal{X}$. Specifically, each ball is centered at a randomly selected test point, and its radius is drawn independently from a uniform distribution on $(0.1,\,0.25)$.
    \item[(d)] Set volume: $\mathrm{Median}\{|\widehat{C}(X_i)|: i\in\gD_{\rm test}\}$, which serves as a robust empirical metric in case of sets with infinite volume.
\end{itemize}

We first evaluate the competing methods by varying the calibration sample size. The results in Figure \ref{fig:cond_metric_sample_size_varies_E} show that, compared to the baseline methods, \texttt{MOPI} achieves worst-case coverage rates closest to the nominal $90\%$ level across all calibration set sizes, while consistently attaining a lower MSCE. Moreover, during the calibration stage, our method can adapt the shape of the prediction sets via the learned $\Sigma(x)$, whereas the baseline methods are limited to simply rescaling the pretrained sets by adjusting a scalar quantile. Consequently, \texttt{MOPI} produces substantially smaller prediction sets. As shown in Table~\ref{tab:multi_shape}, we further compare four methods across different label-space dimensions $d_{\gY}\in\{2,4,6\}$ and different nominal coverage levels $1-\alpha \in \{95\%, 90\%, 85\%\}$. The results indicate that \texttt{MOPI} consistently outperforms the baselines in both conditional coverage metrics and set volumes.

\begin{figure}[ht]
    \centering
    \includegraphics[width=\linewidth]{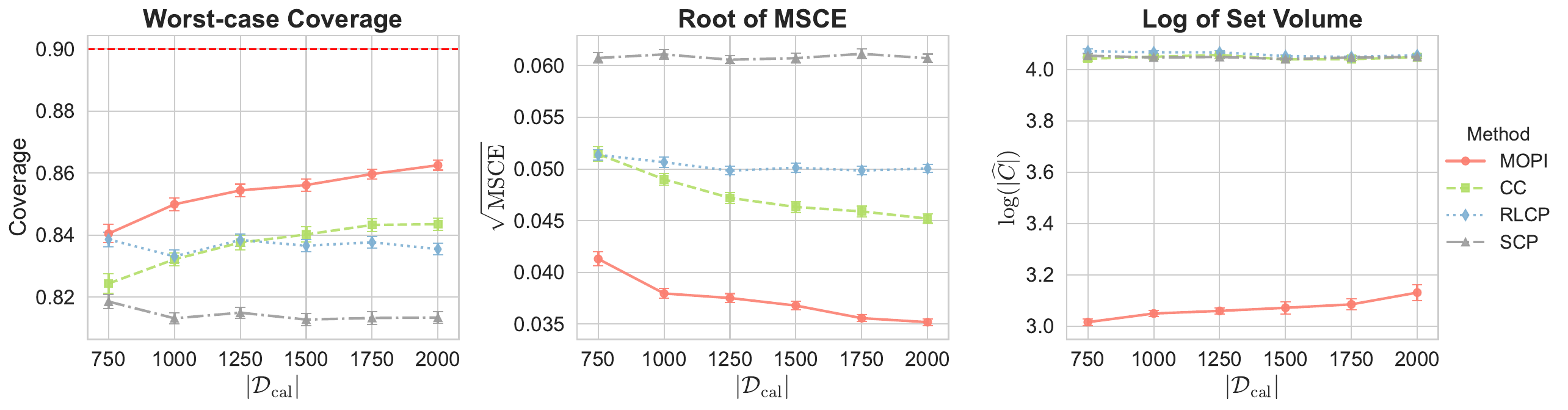}
    \caption{Conditional coverage metrics and log of set volumes (denoted by $\log|\widehat{C}|$) versus sample sizes of calibration set under {ellipsoidal} sets and $d_{\gY} = 2$. The error bars represent the standard deviation of the metrics over $100$ replications.}
    \label{fig:cond_metric_sample_size_varies_E}
\end{figure}

\begin{table}[ht]
\centering
\setlength{\tabcolsep}{3pt}
\caption{Coverage metrics and log of set volumes (denoted by $\log|\widehat{C}|$) for {ellipsoidal} prediction sets under different label-space dimensions and nominal coverage levels.}
\resizebox{\textwidth}{!}{
\begin{tabular}{lcccccccccccc}
\toprule
& \multicolumn{4}{c}{$d_{\gY}=2$} 
& \multicolumn{4}{c}{$d_{\gY}=4$} 
& \multicolumn{4}{c}{$d_{\gY}=6$} \\
Methods & \textit{Marginal} & \textit{Worst-case} & $\sqrt{\rm MSCE}$ & $\log|\widehat{C}|$ 
& \textit{Marginal} & \textit{Worst-case} & $\sqrt{\rm MSCE}$ & $\log|\widehat{C}|$
& \textit{Marginal} & \textit{Worst-case} & $\sqrt{\rm MSCE}$ & $\log|\widehat{C}|$\\
\midrule

\multicolumn{13}{c}{\textit{Panel A: $1-\alpha = 95\%$}} \\
MOPI & 0.952 & \textbf{0.918} & \textbf{0.026} & \textbf{3.301}
     & 0.955 & \textbf{0.920} & \textbf{0.027} & \textbf{7.752}
     & 0.953 & \textbf{0.914} & \textbf{0.027} & \textbf{13.311} \\
CC   & 0.948 & 0.905 & 0.034 & 4.393
     & 0.947 & 0.907 & 0.032 & 10.228
     & 0.948 & 0.905 & 0.033 & 18.845 \\
RLCP & 0.952 & 0.908 & 0.034 & 4.431
     & 0.951 & 0.892 & 0.040 & 10.248
     & 0.952 & 0.900 & 0.036 & 18.844 \\
SCP  & \textbf{0.950} & 0.892 & 0.040 & 4.414
     & \textbf{0.950} & 0.858 & 0.054 & 10.177
     & \textbf{0.951} & 0.872 & 0.047 & 18.718 \\

\addlinespace
\multicolumn{13}{c}{\textit{Panel B: $1-\alpha = 90\%$}} \\
MOPI & 0.903 & \textbf{0.856} & \textbf{0.037} & \textbf{3.072}
     & 0.902 & \textbf{0.856} & \textbf{0.035} & \textbf{7.301}
     & 0.902 & \textbf{0.853} & \textbf{0.037} & \textbf{12.714} \\
CC   & 0.897 & 0.840 & 0.046 & 4.040
     & 0.898 & 0.847 & 0.044 & 9.636
     & 0.898 & 0.844 & 0.045 & 17.960 \\
RLCP & 0.901 & 0.837 & 0.050 & 4.052
     & 0.901 & 0.817 & 0.059 & 9.586
     & 0.902 & 0.828 & 0.054 & 17.865 \\
SCP  & \textbf{0.900} & 0.813 & 0.061 & 4.041
     & \textbf{0.900} & 0.767 & 0.083 & 9.500
     & \textbf{0.900} & 0.787 & 0.073 & 17.722 \\

\addlinespace
\multicolumn{13}{c}{\textit{Panel C: $1-\alpha = 85\%$}} \\
MOPI & 0.853 & \textbf{0.799} & \textbf{0.044} & \textbf{2.901}
     & 0.853 & \textbf{0.801} & \textbf{0.042} & \textbf{7.064}
     & 0.852 & \textbf{0.795} & \textbf{0.043} & \textbf{12.313} \\
CC   & 0.846 & 0.780 & 0.056 & 3.780
     & 0.848 & 0.788 & 0.053 & 9.204
     & 0.848 & 0.785 & 0.053 & 17.320 \\
RLCP & \textbf{0.850} & 0.774 & 0.061 & 3.777
     & 0.850 & 0.751 & 0.073 & 9.117
     & 0.851 & 0.764 & 0.066 & 17.188 \\
SCP  & 0.849 & 0.746 & 0.074 & 3.759
     & \textbf{0.850} & 0.692 & 0.103 & 9.020
     & \textbf{0.850} & 0.713 & 0.092 & 16.992 \\

\bottomrule
\end{tabular}
}
\label{tab:multi_shape}
\end{table}

We next consider group-conditional coverage for box-shaped prediction sets, with groups
$G_k=[k-1,k)$, $k=1,\ldots,5$, and
$Z=(\mathbbm{1}\{X\in G_1\},\ldots,\mathbbm{1}\{X\in G_5\})$.
For \texttt{MOPI} and \texttt{CC}, the corresponding function classes are all chosen as
$\{\beta^\top z:\beta\in\mathbb{R}^5\}$.
Let $\sigma_{0,k}$ denote the empirical standard deviation of
$\{Y_i:X_i\in G_k\}_{i\in\mathcal{D}_{\rm pre}}$, with
$|\mathcal{D}_{\rm pre}|=1,500$, and define
$\sigma_0(X)=\sum_{k=1}^5\sigma_{0,k}\mathbbm{1}\{X\in G_k\}$.
The baselines use
$s(X,y)=\|(y-\mu_0(X))/\sigma_0(X)\|_\infty$, where $\mu_0$ is a pretrained random forest.
After the same normalization $\tilde y=(y-\mu_0(x))/\sigma_0(x)$, \texttt{MOPI} uses the box class
$\fC_{\rm box}=\{C(x;\sigma)=\{y:\|\tilde y/\sigma(x)\|_\infty\le1\}:\sigma\in\gH\}$.
We evaluate group-wise coverage by averaging
$\mathbbm{1}\{Y_i\in\widehat C(X_i)\}$ over test points with $X_i\in G_k$.
As shown in Figure~\ref{fig:group_cond_overlapping_dy3}, for
$|\gD_{\rm cal}|=1,500$, \texttt{MOPI} achieves coverage closer to the nominal level and smaller box sizes across groups than the baselines.

\begin{figure}[ht]
    \centering
    \includegraphics[width=\linewidth]{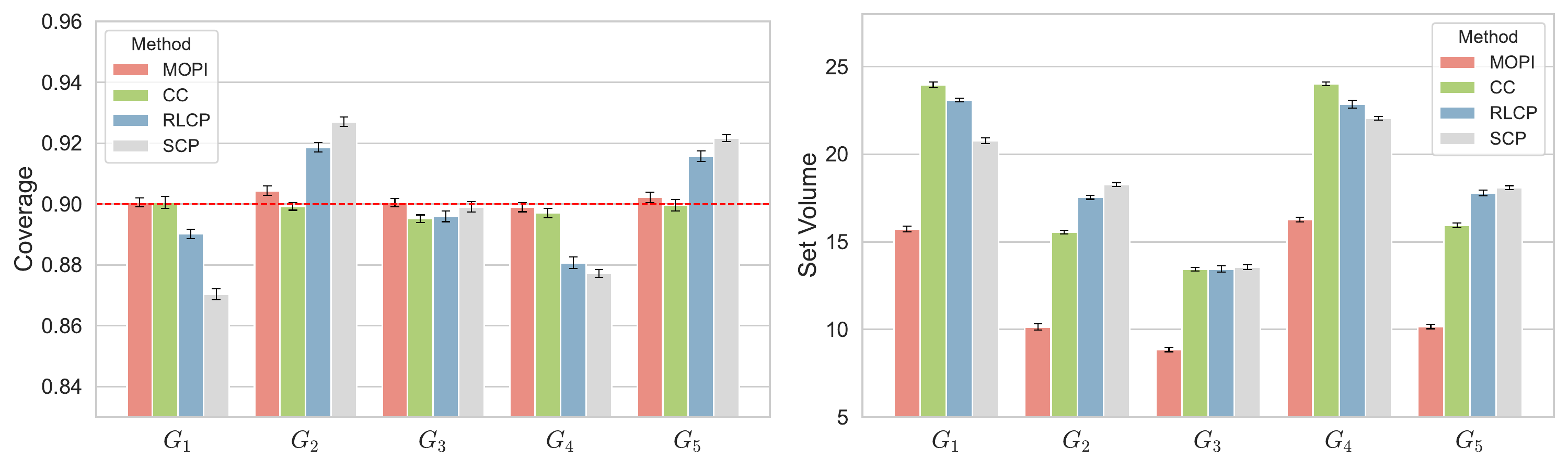}
    \caption{Group-conditional coverage rates and set sizes under {box} sets and $d_{\gY} = 2$. The error bars represent the standard deviation of the metrics over $100$ replications.}
    \label{fig:group_cond_overlapping_dy3}
\end{figure}

\subsection{Sublevel sets for equalized coverage on sensitive attributes}
\label{sec:simu_equal_coverage}
In this subsection, we study the equalized coverage setting in Example~\ref{exam:equal_coverage}, where the sensitive attribute $Z$ is unobserved at test time. Prediction sets are constructed through the sublevel-set formulation \eqref{eq:pretrain_level_set} using a common pretrained score. Specifically, we fit a linear model $\mu$ on a pretraining set of size $|\gD_{\rm pre}|=3,000$ and use the absolute residual score $s(x,y)=|y-\mu(x)|$. All methods are then calibrated on $|\gD_{\rm cal}|=1,500$ samples and evaluated on $|\gD_{\rm test}|=500$ test samples. The data are generated from the heteroscedastic model
$Y=\mu^*(X,Z)+\varepsilon(X,Z)$, with $\varepsilon(X,Z)\sim\mathrm{N}(0,\sigma^*(X,Z)^2)$, where $X=(X_1,\ldots,X_{11})^\top$, $X_1\sim\mathrm{Unif}(0,5)$, and $X_2,\ldots,X_{11}\stackrel{\rm i.i.d.}{\sim}\mathrm{N}(0,1)$.
We consider the following two settings:
\begin{itemize}
    \item \textbf{Setting 1}: $\mu^*(X,Z)=\sum_{j=1}^{11} \frac{1}{\sqrt{11}} X_{j} + 0.5 Z$ and $\sigma^{*}(X,Z)^2 = 1+Z^2 +0.5\times\sin(\sum_{j=2}^{11}X_{j})$ with $Z=\one\{X_{1} \leq 2.5\}$ and $\gZ = \{0,1\}$.
    \item \textbf{Setting 2}: $\mu^*(X,Z) = \mu_Z^*(X)$ and $\sigma^*(X,Z) = \sigma_Z^*(X)$ for $Z\in \gZ = \{1,2,3,4\}$, where $Z\mid X\sim \text{Cat}(\pi_1^*(X),\ldots,\pi_4^*(X))$, and $\text{Cat}(\cdot,\ldots,\cdot)$ denotes the categorical distribution. The functions $\{\mu_z^*(X),\sigma_z^*(X), \pi_z^*(X)\}_{z=1}^4$ are given in Appendix F.2.
\end{itemize}
For \texttt{MOPI}, we take the weight function class as $\gF = \left\{ \sum_{z \in \gZ} \beta_z \mathbbm{1}\{Z = z\} : \beta_z \in \sR \right\}$. For the comparison, we let $\gH$ in \texttt{MOPI} and $\gF$ in \texttt{CC} be the same RKHS on $\gX$ equipped with a Gaussian kernel.
We report the conditional coverage rate for each sensitive attribute value $z \in \gZ$, which should reach the nominal level $1-\alpha=90\%$.

\begin{figure}[ht]
    \centering
    \includegraphics[width=\linewidth]{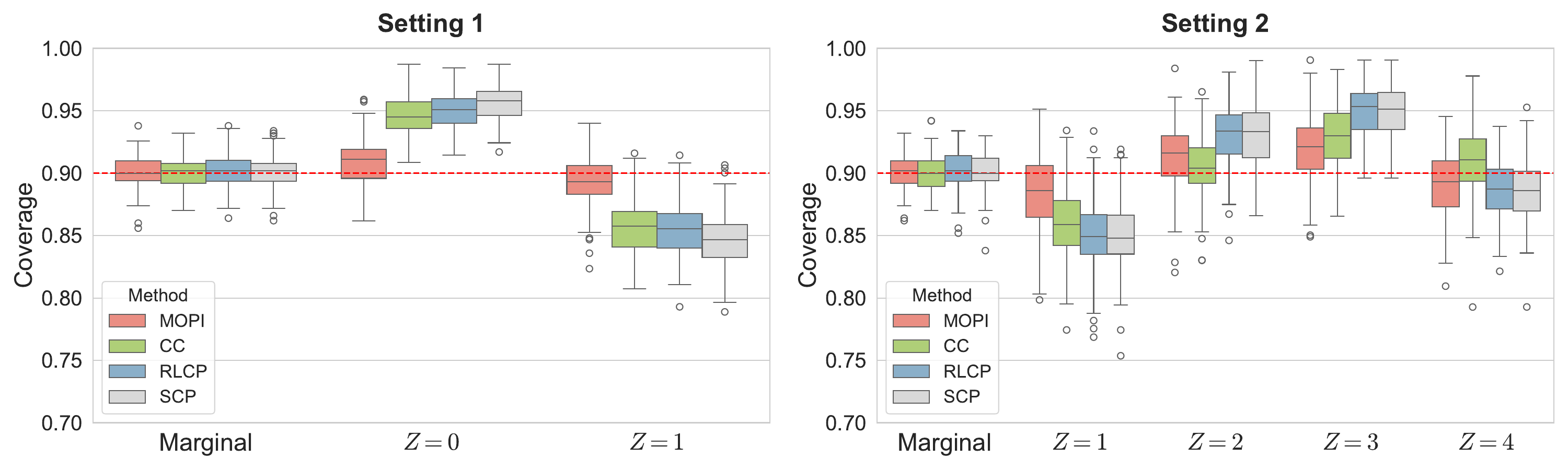}
    \caption{Marginal coverage rate and conditional coverage rates on sensitive attributes. }
    \label{fig:equalized_cov_full}
\end{figure}

The experiment results are reported in Figure \ref{fig:equalized_cov_full}. 
We can see that \texttt{MOPI} attains the best equalized coverage in both settings, again without sacrificing marginal coverage. In \textbf{Setting~1}, where $Z$ is a deterministic function of $X$, both \texttt{CC} and \texttt{RLCP} show moderate improvements in equalized coverage over \texttt{SCP}. However, in \textbf{Settings~2}, where $Z$ is correlated with $X$ but not a deterministic function of it, \texttt{CC} and \texttt{RLCP} largely fail to leverage the information in $Z$ to improve equalized coverage. This limitation arises because these methods cannot incorporate information about $Z$ during the calibration stage. In contrast, \texttt{MOPI} effectively utilizes the available calibration-stage information about $Z$, thereby achieving substantially better equalized coverage. 

Now, we empirically investigate the impact of the dependence between $X$ and $Z$ on conditional coverage performance under the following setting.
\begin{itemize}
    \item \textbf{Setting 1$^\prime$}: $\mu^*(X,Z)$ and $\sigma^{*}(X,Z)^2$ are the same as \textbf{Setting 1}, where $Z=\one\{\rho^* X_{1} + (1-\rho^*) V\leq 2.5\rho^*\}$, where $V\sim \mathrm{N}(0,1)$ and $\rho^*\in \{0,0.25, 0.5, 0.75, 1\}$.
\end{itemize}
The dependence between the conditioning variable $Z$ and the covariate $X$ increases monotonically with $\rho^*$, which plays a similar role to $\rho^2$ in Theorem \ref{thm:LS2_projX}.
When $\rho^*=1$, $Z$ is a deterministic function of $X$; when $\rho^*=0$, $Z$ is independent of $X$.
As shown in Figure \ref{fig:equalized_cov_rho}, for all values of $\rho^*$, the proposed \texttt{MOPI} method consistently achieves better conditional coverage than the three baseline approaches. The simulation results are consistent with the theoretical insight in Theorem \ref{thm:LS2_projX} and Corollary \ref{cor:mse_finite_equal} that a larger correlation of $Z$ and $X$ leads to smaller MSCE and conditional coverage gap.

\begin{figure}[ht]
    \centering
    \includegraphics[width=\linewidth]{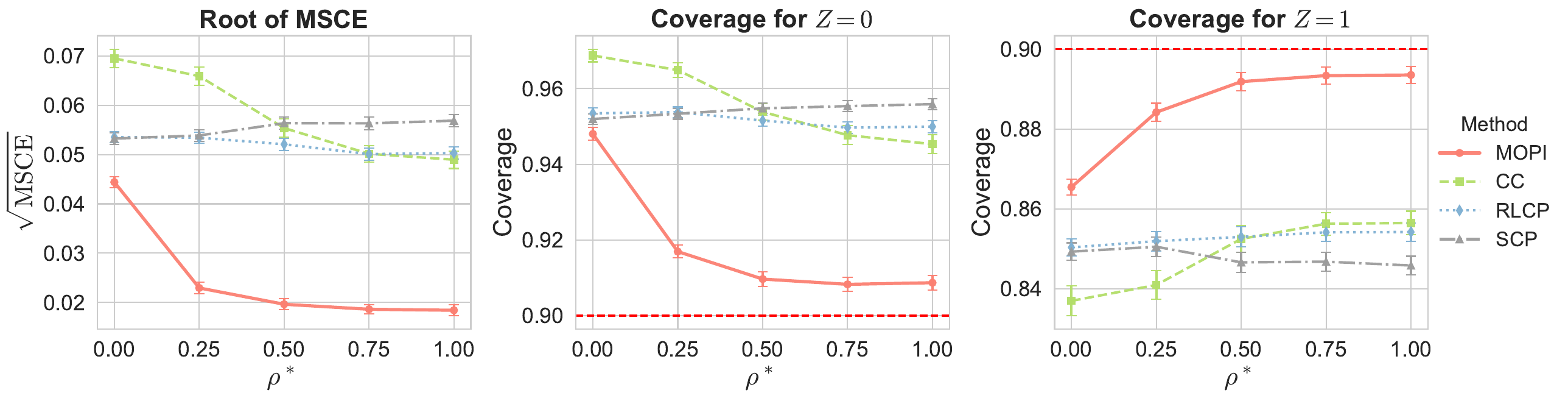}
    \caption{Root of MSCE and conditional coverage rates under {\bf Setting $1^{\prime}$} as $\rho^*$ varies. The error bars represent the standard deviation of the metrics over $100$ replications.}
    \label{fig:equalized_cov_rho}
\end{figure}

\section{Real Data Analysis}\label{sec:real_data}
In this section, we examine the performance of \texttt{MOPI} and three baseline methods on two real-world datasets. The additional experiment results are deferred to Appendix~\ref{appen:additional_realdata}. 


\subsection{Households dataset}\label{subsec:Households}
In this subsection, we consider the Households Dataset, a multivariate response dataset previously analyzed in \citet{dheur2025unified}.
The response dimension is $d_{\gY}=2$, corresponding to expenditures on housing and health. 
The covariates include five household characteristics, expenditures on transportation, entertainment, food, utilities, and household income. The full dataset consists of 7,207 observations.
Due to the wide numerical ranges, we apply a logarithmic transformation to both the labels and the covariates.

In this experiment, we focus on constructing ellipsoidal prediction sets.
In each repetition, we split the dataset into three parts: pretraining set with size $|\gD_{\rm pre}| = 3,460$, calibration set with size $|\gD_{\rm cal}| = 2,307$, and test set with size $|\gD_{\rm test}| = 1,440$. The experiment is repeated 100 times with random splits, and the target coverage level is $1-\alpha = 90\%$. 
The baseline methods use the score function $s(x,y)=(y-\mu_0(x))^\top \Sigma^{-1}_0(x)(y-\mu_0(x))$, where $\mu_0(x)$ is a random forest predictor and $\Sigma_0(x)$ is a neural network fitted on the pretraining data. The detailed formulation of the pretraining loss is provided in Appendix~\ref{subsec:Imple_multiY}. For \texttt{CC}, the function space is chosen to be an RKHS with a Gaussian kernel, and \texttt{RLCP} is implemented using the same Gaussian kernel.
For \texttt{MOPI}, we take $\gF$ to be an RKHS equipped with a Gaussian kernel, and $\gH$ to be the class of neural networks with the same architecture used to estimate $\Sigma_0(x)$. 

\begin{figure}[H]
    \centering
    \begin{subfigure}[b]{0.49\textwidth}
        \centering
        \includegraphics[width=\linewidth]{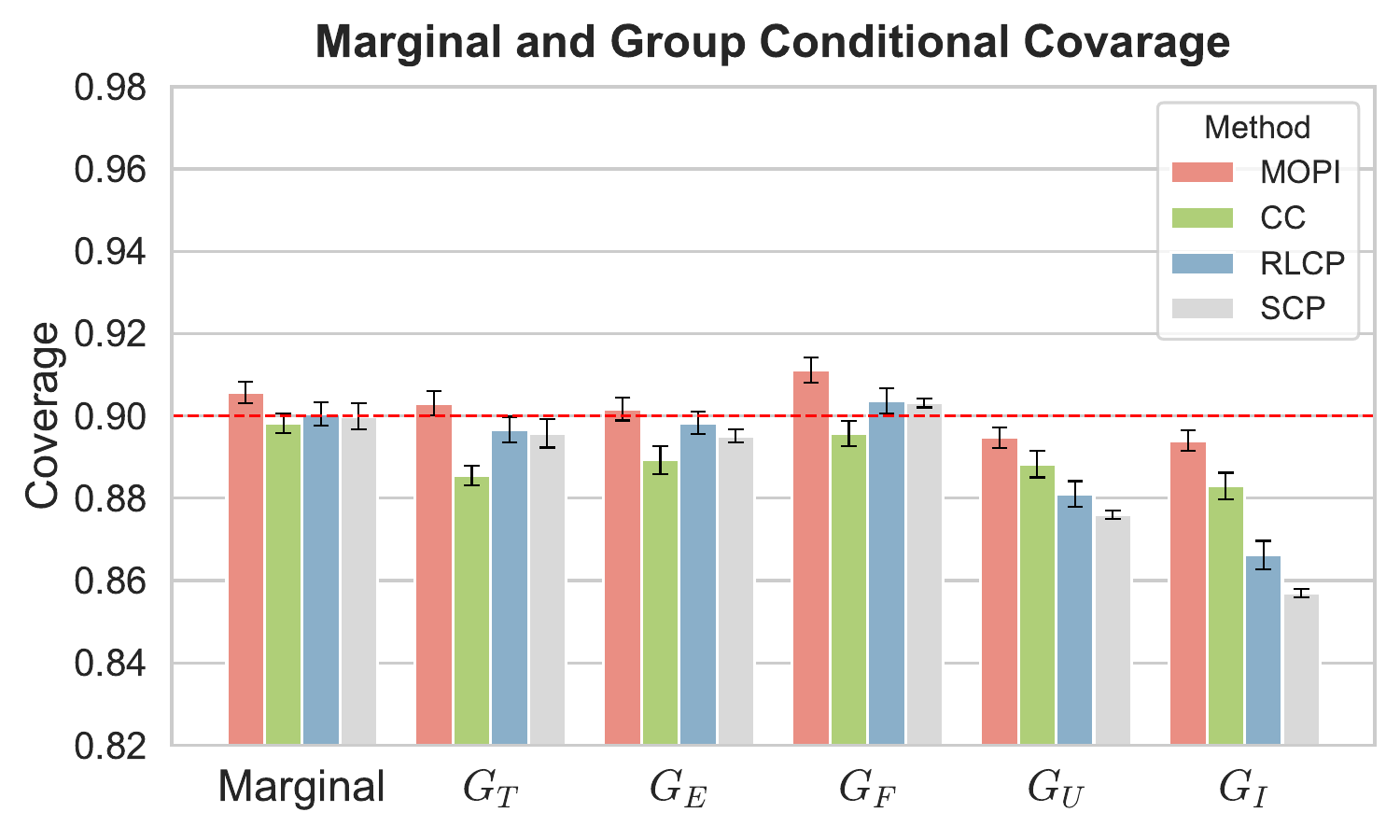}
    \end{subfigure}
    \begin{subfigure}[b]{0.49\textwidth}
        \centering
        \includegraphics[width=\linewidth]{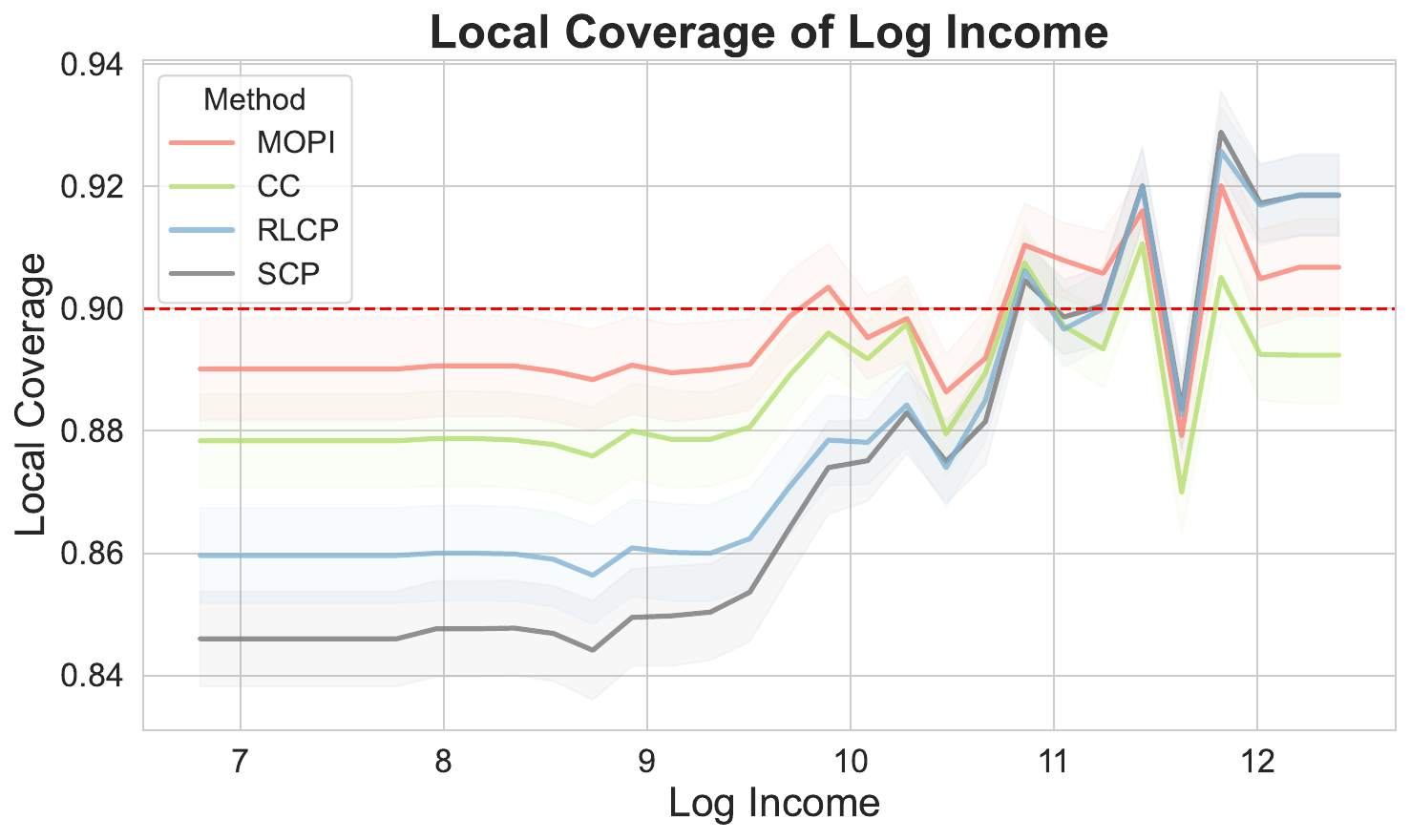}
    \end{subfigure}
    \caption{Coverage comparison for the Households dataset. The error bars (left) show the standard deviation, and the confidence bands (right) represent approximate $95\%$ normal confidence intervals, both computed over $100$ replications.}
    \label{fig:Households}
\end{figure}

We evaluate two metrics to assess test-conditional coverage. 
(i) \emph{Group-conditional coverage}: we consider subpopulations corresponding to low expenditure or income levels across five covariates described above. Specifically, for each covariate, we define the ``low'' group as those observations whose value falls below the $10$th percentile of that covariate over the full sample, denoted by $\{G_T, G_E, G_F, G_U, G_I\}$, where the subscripts denote the initials of the five covariates described above, respectively.
(ii) \emph{Local coverage of log income}: we partition the range of the household income into 30 evenly spaced grid points and compute the coverage rate at each point using the nearest 5\% of test samples. 
From Figure~\ref{fig:Households}, we observe that \texttt{MOPI} achieves coverage that is generally close to the nominal level across most groups while maintaining marginal validity.
Moreover, the local coverage of log income shows that \texttt{MOPI} provides more uniform coverage across income levels.


\subsection{Communities and Crime dataset}
In this subsection, we focus on the Communities and Crime dataset \citep{redmond2002data}, in which the objective is to predict each community's per-capita violent crime rate from its demographic and socioeconomic characteristics. Following \citet{gibbs2023conformal}, we consider a concise set of explanatory variables, including population size, unemployment rate, median household income, percentage of residents with limited English proficiency, racial composition (the proportions of Black, White, Asian, and Hispanic categories, denoted as $X_{\rm \%Black}$, $X_{\rm \%White}$, $X_{\rm \%Asian}$, and $X_{\rm \%Hispanic}$, respectively), and age structure (the shares of individuals aged 12--21 and those over 65). Let $X^{\rm all}$ denote the collection of all covariates listed above, and define the racial composition variables as $X^{\rm racial} = (X_{\rm \%Black}, X_{\rm \%White}, X_{\rm \%Asian}, X_{\rm \%Hispanic})$. This vector $X^{\rm racial}$ serves as the conditioning variable $Z$. In this experiment, we consider the sublevel sets defined in \eqref{eq:pretrain_level_set} with the pretrained score $s(x,y)=|y-{\mu}_0(x)|$, where $\mu_0$ is fitted on the pretraining set. The experiment is repeated 100 times with random splits. In each repetition, we split the original dataset into three parts: a pretraining set of size $|\gD_{\rm pre}| = 650$, a calibration set of size $|\gD_{\rm cal}| = 650$, and a test set of size $|\gD_{\rm test}| = 694$.

 
We set both $\gH$ and $\gF$ in \texttt{MOPI} to be the same RKHS with Gaussian kernels. \texttt{RLCP} is implemented using the same Gaussian kernel.
We compute the linear re-weighting coverage of each racial group \citep{gibbs2023conformal} defined as
$\sum_{i\in \gD_{\rm test}} \frac{X_{i,j}^{\rm racial}}{\sum_{i\in \gD_{\rm test}} X_{i,j}^{\rm racial}}\;\one\{Y_{i}\in\widehat{C}(X_{i})\},\;j=1,\ldots,4,$
where $X_{i,j}^{\rm racial}$ denotes the proportion of racial group $j$ in community $i$. We consider two experimental settings. In the \textit{unmasked} case, both the calibration and test datasets contain the full set of covariates, including racial and demographic attributes. This corresponds to the case $X = X^{\rm all}$ and $Z = X^{\rm racial}$.
In this case, we implement \texttt{CC} under the same configuration as in \citet{gibbs2023conformal}. 
Specifically, we choose the function class $\gF = \{ \beta_0+ Z^{\top}\beta + f(X) : f \in \sH, \beta_0 \in \sR, \beta \in \sR^4 \},$ where $\sH$ is an RKHS equipped with a Gaussian kernel.
In the \textit{masked} case, the sensitive attribute is unobserved at test time. This corresponds to the scenario where $X = X^{\rm all} \setminus X^{\rm racial}$ and $Z = X^{\rm racial}$. 
Since $Z$ is unobserved in test data, we adopt the function class $\gF=\{ \beta_0 + f(X) : f \in \sH,\beta_0\in \sR \}$ for \texttt{CC}. 
The hyperparameters of each method are selected via cross-validation on the calibration data by minimizing the linearly reweighted coverage error, meaning that baseline methods can also leverage information from sensitive attributes during calibration.


\begin{figure}[h]
    \centering
    \begin{subfigure}[b]{0.49\textwidth}
        \centering
        \includegraphics[width=\linewidth]{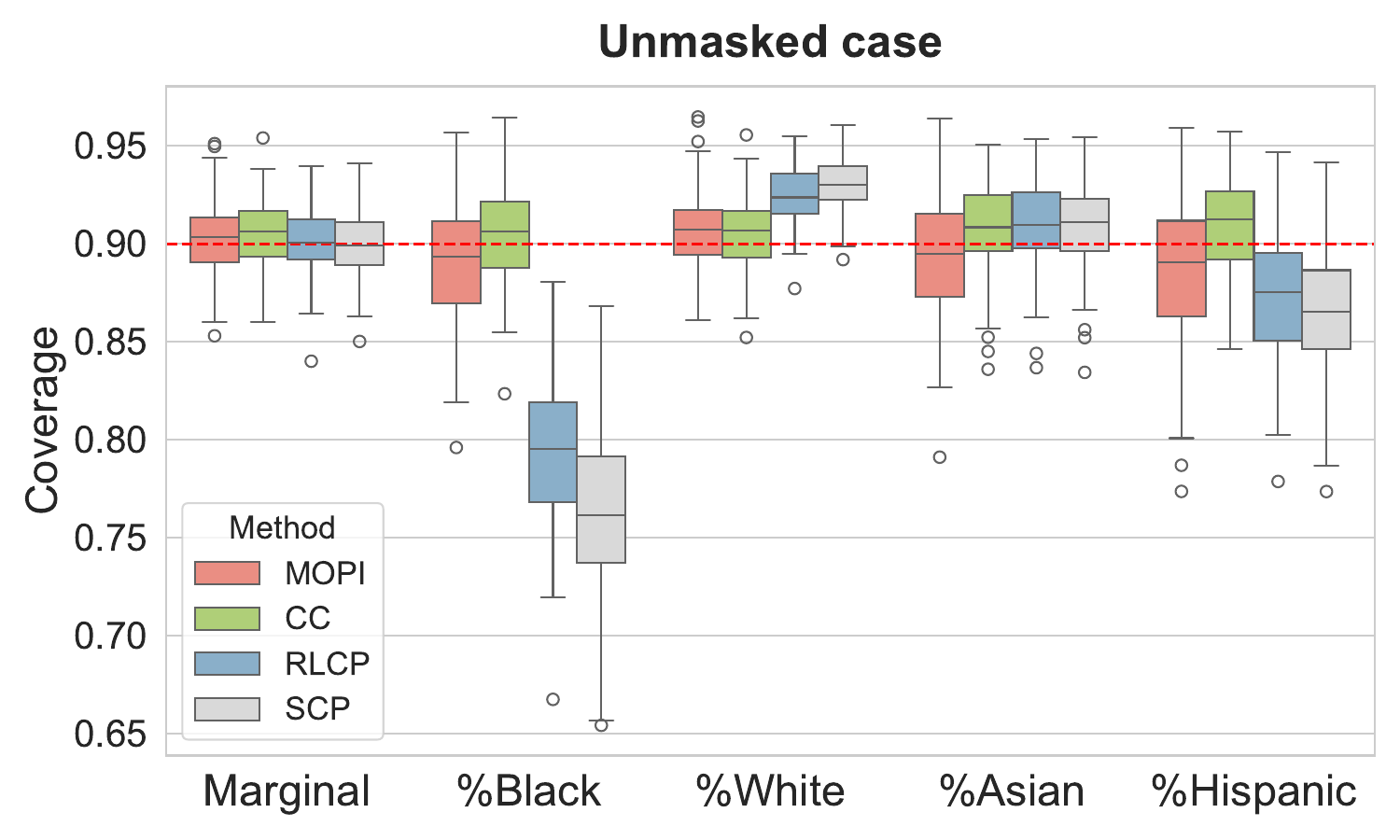}
    \end{subfigure}
        \begin{subfigure}[b]{0.49\textwidth}
        \centering
        \includegraphics[width=\linewidth]{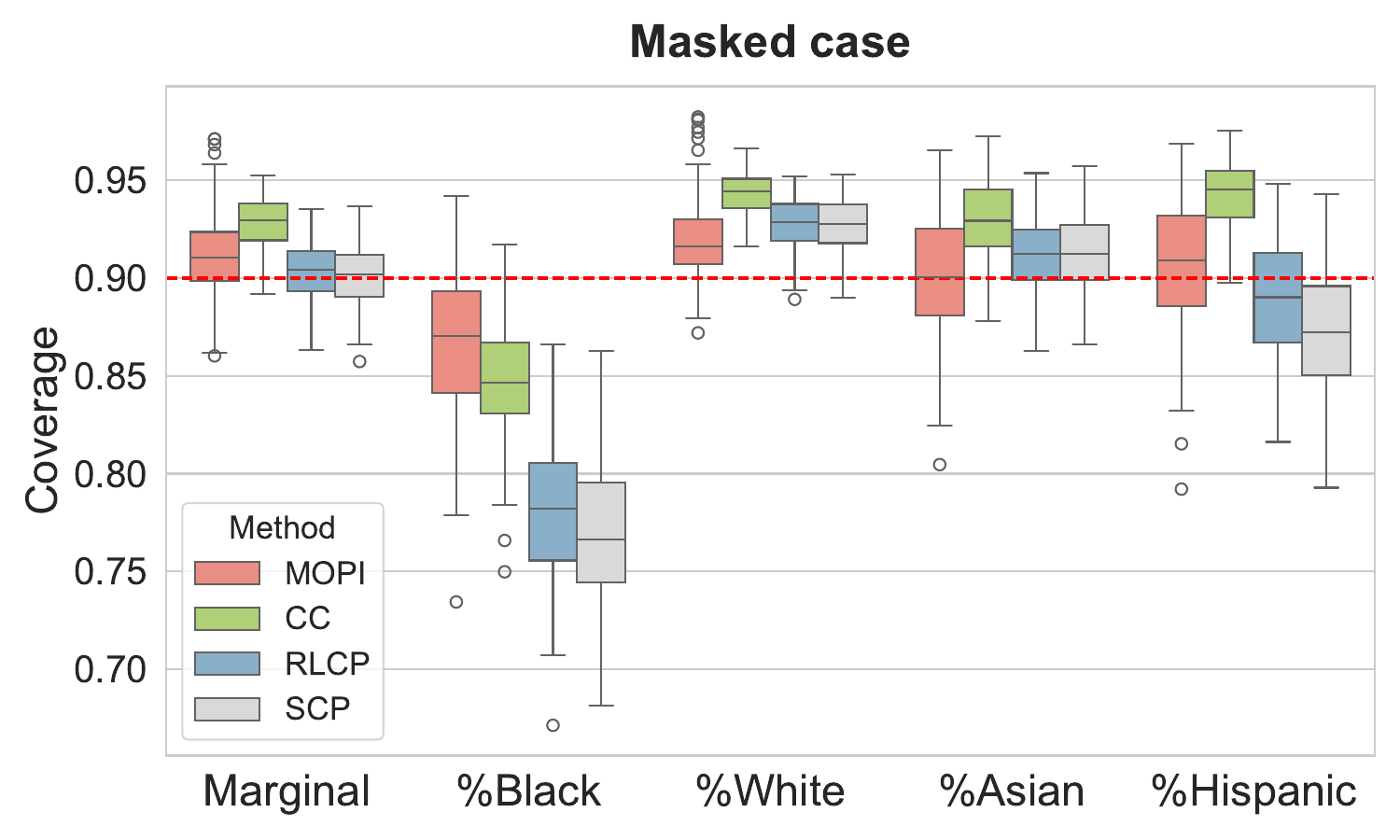}
    \end{subfigure}
    \caption{Marginal and linear re-weighting coverage for the Communities and Crime dataset. The box plots are obtained from $100$ replications with random data splits.}
    \label{fig:Crime_observe_90}
\end{figure}

The experiment results are reported in Figure \ref{fig:Crime_observe_90}.
In the {unmasked} case, both \texttt{MOPI} and \texttt{CC} successfully achieve the target coverage level across all four racial groups.
In the masked case, the baseline methods fail to maintain the desired coverage across different racial subgroups. By contrast, since \texttt{MOPI} can leverage $Z$ during the minimax optimization phase, it achieves better and more balanced coverage across racial groups than baselines.

\section{Concluding Remarks}\label{sec:conclusion}


Two promising directions remain for future work. First, the miscoverage indicator in objective \eqref{eq:sample_minimax_CP} results in a nonconvex optimization problem; it would thus be meaningful to explore convex surrogate losses and analyze the corresponding theoretical properties, as has been done in classification contexts \citep{bartlett2006convexity}. Second, while the prediction sets in this paper are constructed to satisfy $Z$-conditional coverage, a challenging open problem is how to simultaneously optimize the efficiency of a downstream task under the same conditional coverage constraint, such as volume \citep{braun2025minimum} and decision loss \citep{bao2025optimal}.

\bibliographystyle{apalike}
\bibliography{references}

\newpage
\appendix
\input{appendix}

\end{document}

%% file: math_commands.tex
\usepackage{amsmath,amsfonts,bm}

\def\one{\mathbbm{1}}

\def\1{\bm{1}}

\def\vf{{\bm{f}}}

\def\vu{{\bm{u}}}

\def\mI{{\bm{I}}}

\def\mK{{\bm{K}}}

\DeclareMathAlphabet{\mathsfit}{\encodingdefault}{\sfdefault}{m}{sl}
\SetMathAlphabet{\mathsfit}{bold}{\encodingdefault}{\sfdefault}{bx}{n}

\def\gD{{\mathcal{D}}}
\def\gE{{\mathcal{E}}}
\def\gF{{\mathcal{F}}}
\def\gG{{\mathcal{G}}}
\def\gH{{\mathcal{H}}}

\def\gJ{{\mathcal{J}}}
\def\gK{{\mathcal{K}}}

\def\gM{{\mathcal{M}}}

\def\gR{{\mathcal{R}}}

\def\gT{{\mathcal{T}}}
\def\gU{{\mathcal{U}}}
\def\gV{{\mathcal{V}}}

\def\gX{{\mathcal{X}}}
\def\gY{{\mathcal{Y}}}
\def\gZ{{\mathcal{Z}}}

\def\sH{{\mathbb{H}}}

\def\sK{{\mathbb{K}}}

\def\sP{{\mathbb{P}}}

\def\sR{{\mathbb{R}}}

\def\fC{{\mathfrak{C}}}

\newcommand{\E}{\mathop{\mathbb{E}}}

\newcommand{\R}{\mathbb{R}}

\newcommand{\Cov}{\mathsf{Cov}}

\DeclareMathOperator*{\argmin}{arg\,min}

\newcommand{\LRabs}[1]{\left|#1\right|}

\newcommand{\LRs}[1]{\left(#1\right)}
\newcommand{\LRm}[1]{\left[#1\right]}
\newcommand{\LRl}[1]{\left\{#1\right\}}

%% file: appendix.tex
\setcounter{table}{0}
\renewcommand{\thetable}{S\arabic{table}}
\setcounter{figure}{0}
\renewcommand{\thefigure}{S\arabic{figure}}
\setcounter{lemma}{0}
\renewcommand{\thelemma}{S\arabic{lemma}}
\setcounter{theorem}{0}
\renewcommand{\thetheorem}{S\arabic{theorem}}

\numberwithin{equation}{section}
\allowdisplaybreaks

\section{Implementation of MOPI}

\subsection{The effect of $L_2$ penalty coefficient}
\label{appen:parameter_select}
In this section, we consider a more general form of \eqref{eq:obj_L2_penalty}:
\begin{align}\label{eq:obj_L2_penalty_general}
    \Psi_\lambda(C, f):= \E\LRm{f(Z)(\mathbbm{1}\{Y\notin C(X)\} - \alpha) - \lambda f^2(Z)},
\end{align}
and therefore \eqref{eq:obj_L2_penalty} is a special case of \eqref{eq:obj_L2_penalty_general} with $\lambda=1$. Depending on the cardinality of $\gZ$, we choose different $\gF$ and admit close forms of $\max_{f\in\gF} \Psi_\lambda(C,f)$.
\begin{lemma}\label{lemma:inner_max_finite_class}
    If the cardinality of $\gZ$ is finite and $\gF = \{\sum_{z\in \gZ} \beta_z\mathbbm{1}\{Z = z\}: \beta_z \in \sR, z\in \gZ\}$, then for any $C\in \mathfrak{C}$ we have
    \begin{align}
        \max_{f\in \gF} \Psi_\lambda(C,f) &= \sum_{z\in |\gZ|} \frac{\LRs{\E\big[\mathbbm{1}\{Z = z\} (\mathbbm{1}\{Y\notin C(X)\} -\alpha)\big]}^2}{4\lambda\sP(Z = z)}.\nonumber
    \end{align}
\end{lemma}

\begin{lemma}
    If the cardinality of $\gZ$ is infinite and $\gF = \sK$ as the RKHS with the kernel function $\gK(\cdot,\cdot):\gZ\times\gZ \to \sR$. For $u,v,h\in\sK$, let $\langle u(\cdot), v(\cdot) \rangle_{\sK}$ be the inner product in RKHS and define $(u\otimes v)(h)=\langle h, v \rangle_{\sK}u$. Suppose $\alpha(\cdot;C)-\alpha\in\gF$ for all $C\in \mathfrak{C}$, then we have
    \begin{align*}
        \max_{f\in \gF} \Psi_\lambda(C,f)  = \frac{1}{4\lambda}\Big\langle \E[\gK(\cdot,Z)\{\alpha(Z;C)-\alpha\}], \gT^{-1}\E[\gK(\cdot,Z)\{\alpha(Z;C)-\alpha\}]\Big\rangle_{\sK},
    \end{align*}
    where $\gT(\cdot)=\E[\gK(\cdot,Z)\otimes\gK(\cdot,Z)]$.
\end{lemma}
From the two lemmas above, we see that when the function class $\gF$ is restricted to a finite-dimensional space or an RKHS according to the cardinality of $\gZ$, the inner maximization problem admits a closed-form solution. Moreover, the parameter $\lambda$ only rescales the objective function and does not affect the optimizer. 

Align with \eqref{eq:obj_L2_penalty_general}, consider the following empirical objective function:
\begin{align*}
    \widehat{\Psi}_\lambda(C,f)=\frac{1}{n}\sum_{i=1}^n f(Z_i)\LRs{\one\{Y_i\notin C(X)\}-\alpha}-\frac{\lambda}{n}\sum_{i=1}^n f^2(Z_i),
\end{align*}
where $\lambda>0$ is a hyperparameter. The empirical minimax optimization problem can be written as:
\begin{align}
    \label{eq:sample_minimax_CP_lambda}
    \argmin_{C\in\fC}\max_{f\in\gF}\LRl{\widehat{\Psi}_\lambda(C,f)-\gamma\|f\|_\gF^2}.
\end{align}
 Hence, problem \eqref{eq:sample_minimax_CP} is a special case of \eqref{eq:sample_minimax_CP_lambda} by letting $\lambda=1$. For finite dimensional $\gF$, we let $\gamma=0$ and $\gF=\{\sum_{z\in\gZ}\beta_z\one\{Z=z\}:\beta_z\in\R, z\in\gZ\}$, by Lemma \ref{lemma:inner_close_form_finite_lambda}, the \eqref{eq:sample_minimax_CP_lambda} can be rewritten as:
 $$
 \argmin_{C\in\fC} \LRl{\sum_{z\in\gZ}\frac{\big(\sum_{i=1}^n\mathbbm{1}\{Z_i=z\}\LRs{\mathbbm{1}\{Y_i\notin C(X_i)\}-\alpha}\big)^2}{4\lambda n\sum_{i=1}^n\mathbbm{1}\{Z_i=z\}}}.
 $$
 Since $\lambda$ only acts as a multiplicative scaling factor in the objective, its value does not affect the optimizer.
 Similarly, For $\gF=\sK$ is a RKHS, by Lemma \ref{lemma:inner_close_form_rkhs_lambda}, the \eqref{eq:sample_minimax_CP_lambda} can be rewritten as: 
 $$
 \argmin_{C\in\fC} \LRl{\frac{1}{4\lambda}\boldsymbol{\varphi}_n(C)^\top \mK_n \left( \frac{1}{n} \mK_n +\underbrace{\gamma/\lambda}_{:=\tilde{\gamma}} \mI_n \right)^{-1} \boldsymbol{\varphi}_n(C)}.
 $$
 Thus, by reparameterizing $\gamma$ as $\tilde{\gamma}=\gamma/\lambda$, the parameter $\lambda$ still only rescales the objective function and does not affect the resulting optimal solution. Consequently, throughout this paper we set $\lambda=1$ both theoretically and empirically.

\subsection{Sensitive analysis of smoothing parameter}
\label{appen:sensitive_smooth}
To study the effect of the smoothing parameter $r$ on the resulting prediction sets, we consider \textbf{Setting~1($x$)} and apply the sigmoid surrogate with different values of $r$. The experimental setup is identical to that in Section~\ref{sec:test_cond_one_dim}. 

As shown in Figure~\ref{fig:sensitive_smooth}, the Root of MSCE decreases rapidly as $r$ decreases and attains its minimum around $r=0.1$. At the same time, both the marginal coverage and the worst-case conditional coverage decrease from over-coverage to values close to the nominal level. As $r$ is further reduced, the Root of MSCE increases slowly with a mild slope, and the changes in marginal and worst-case conditional coverage remain limited.
Moreover, the error bars indicate that smaller values of $r$ lead to larger variability in coverage. This is expected, since a smaller $r$ shrinks the region where the smooth surrogate has non-negligible gradients, which in turn makes the optimization more unstable.

\begin{figure}[H]
    \centering
    \includegraphics[width=0.75\linewidth]{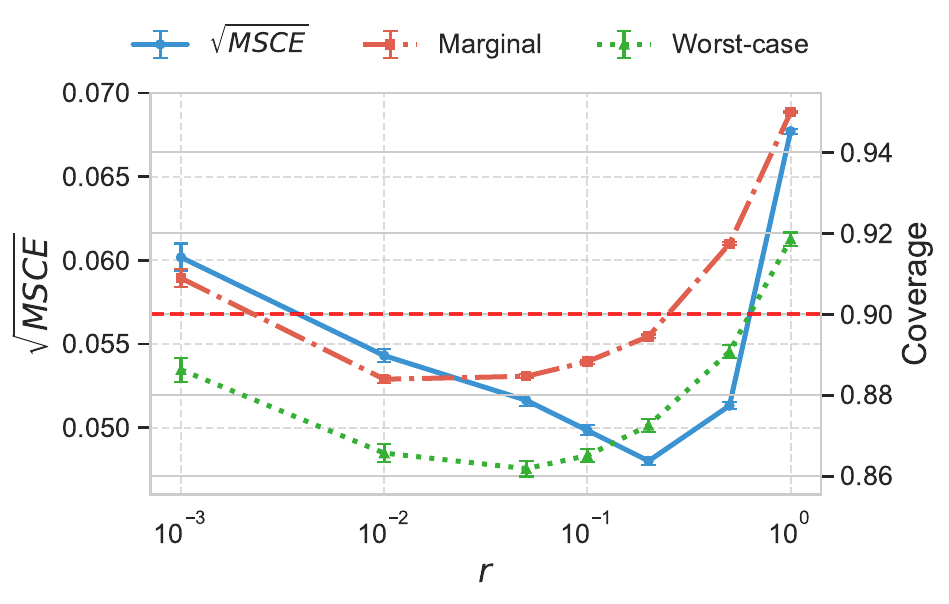}
    \caption{Root of MSCE, Marginal coverage and Worst-case conditional coverage as smoothing parameter $r$ varies. Dots and error bars represent the means and confidence intervals over 300 trials.}
    \label{fig:sensitive_smooth}
\end{figure}

Overall, these results suggest that in practice one should choose a relatively small—but not excessively small—smoothing parameter, depending on the available labeled sample size. Moreover, MOPI appears to be reasonably robust to the choice of $r$ over a moderate range, without requiring precise tuning of the optimal value.

\section{More Comparison with Conditional Calibration}\label{appen:comparison_with_QR}

\subsection{Theoretical comparison under sublevel sets}
In this section, we compare the MSCE performance of the quantile regression approach and MOPI. For test-conditional coverage with $Z=X$ consider the sublevel set \eqref{eq:pretrain_level_set}, and set $\gH=\gF$ in MOPI for a fair comparison. For a sufficiently large $\gF$ such that $\alpha(X,C)-\alpha\in\gF$ for all $C\in\fC=\{C(X;f)=\{y:s(X,y)\le f(X)\},f\in\gF\}$, by Proposition \ref{pro:CP_LS2minimax_general}, we have
$\max_{f\in\gF}\Psi(C,f)=\mathsf{MSCE}(C)/4$. Therefore,
\begin{align*}
C^*:=\argmin_{C\in\fC}\max_{f\in\gF}\Psi(C,f)=\argmin_{C\in\fC}\frac{1}{4}\mathsf{MSCE}(C):=C^{\rm ora}.
\end{align*}
On the other hand, the quantile regression predictor is given by $C^{\rm qr}(X) := \{y\in \gY: s(X,y) \leq f^{\rm qr}(X)\}$ where $f^{\rm qr} = \argmin_{f\in \gF}\E[\ell_{\alpha}(s(X,Y), f(X))]$. By the optimality of $C^{\rm ora}$, it follows that
$$
\mathsf{MSCE}(C^*)=\mathsf{MSCE}(C^{\rm ora})\le \mathsf{MSCE}(C^{\rm qr}).
$$
Therefore, the minimax solution $C^*$ achieves
lower MSCE compared with $C^{\rm qr}$.
Let $q(x) = \inf\{t\in \sR: \sP(s(X,Y) \leq t \mid X=x) \geq 1-\alpha\}$ be the ground truth quantile function.
Next, we investigate when the MSCE of $C^*$ and $C^{\rm qr}$ equals zero.
If $q\in\gF$, by the definition of $q$, we have $\mathsf{MSCE}(C(X;q))=0$. Hence, under the condition of Proposition \ref{pro:CP_LS2minimax_general}:
$$
0\leq\mathsf{MSCE}(C^*)=\mathsf{MSCE}(C^{\rm ora})\le\mathsf{MSCE}(C(X;q))=0.
$$
Let $f_x=f(x)$ for any fixed $x\in\gX$.
Note that $\E[\ell_{\alpha}(s(X,Y), f(X))\mid X=x]=\E[\ell_{\alpha}(s(X,Y), f_x)\mid X=x]$, consider the following point-wise optimization:
$$
\min_{f_x\in\R}\E[\ell_{\alpha}(s(X,Y), f_x)\mid X=x].
$$
Then, the first order condition of the above optimization problem is
$$
\frac{\partial}{\partial f_x}\E[\ell_{\alpha}(s(X,Y), f_x)\mid X=x]=\sP\{s(X,Y)<f_x\mid X=x\}-(1-\alpha)=0,
$$
Hence, the optimal solution is $q(x)$ for a fixed $x\in\gX$ by the definition of $q(x)$. If $q\in\gF$, note that 
\begin{align*}
    \min_{f\in \gF}\E[\ell_{\alpha}(s(X,Y), f(X))]&\le \E[\ell_{\alpha}(s(X,Y), q(X))]\\
    &=\E_X\LRm{\E[\ell_{\alpha}(s(X,Y), q(x))\mid X=x]}\\
    &=\E_X\LRm{\min_{f_x\in\R}\E[\ell_{\alpha}(s(X,Y), f_x)\mid X=x]}.
\end{align*}
On the other hand, $\E[\ell_{\alpha}(s(X,Y), f_x)\mid X=x]\geq \E[\ell_{\alpha}(s(X,Y), q(x))\mid X=x]$, by the optimality of $q(x)$. Taking expectation on both side and minimize over $f\in\gF$, we have
\begin{align*}
     \min_{f\in \gF}\E[\ell_{\alpha}(s(X,Y), f(X))]&\geq \E[\ell_{\alpha}(s(X,Y), q(X))].
\end{align*}
Therefore, $f^{\rm qr}(x)=q(x)$ a.s. for all $x\in\gX$. Thus, when $q\in\gF$, $\mathsf{MSCE}(C^{*})=\mathsf{MSCE}(C^{\rm qr})=0$.

\subsection{Computational efficiency of MOPI}
\label{appen:compute_eff}

In this section, we study the computational efficiency of \texttt{MOPI}. 
We focus on \textbf{Setting~1}, with experimental configurations following Section \ref{sec:simu_equal_coverage}. 

We consider two variants of \texttt{CC}: \texttt{CC(full)} and \texttt{CC(non-full)}. 
The \texttt{CC(full)} variant follows \citet{gibbs2023conformal} exactly, solving qunatile regression jointly over $\gD_{\rm cal}$ and each test point. 
In contrast, \texttt{CC(non-full)} estimates the score quantile using only $\gD_{\rm cal}$ via quantile regression. 
Both \texttt{CC(non-full)} and \texttt{MOPI} are optimized using the Adam algorithm with the same learning rate and number of iterations. 
All experiments are conducted on a MacBook Pro equipped with an M4 Pro CPU and 24GB of RAM.
\begin{figure}[H]
    \centering
    \includegraphics[width=0.75\linewidth]{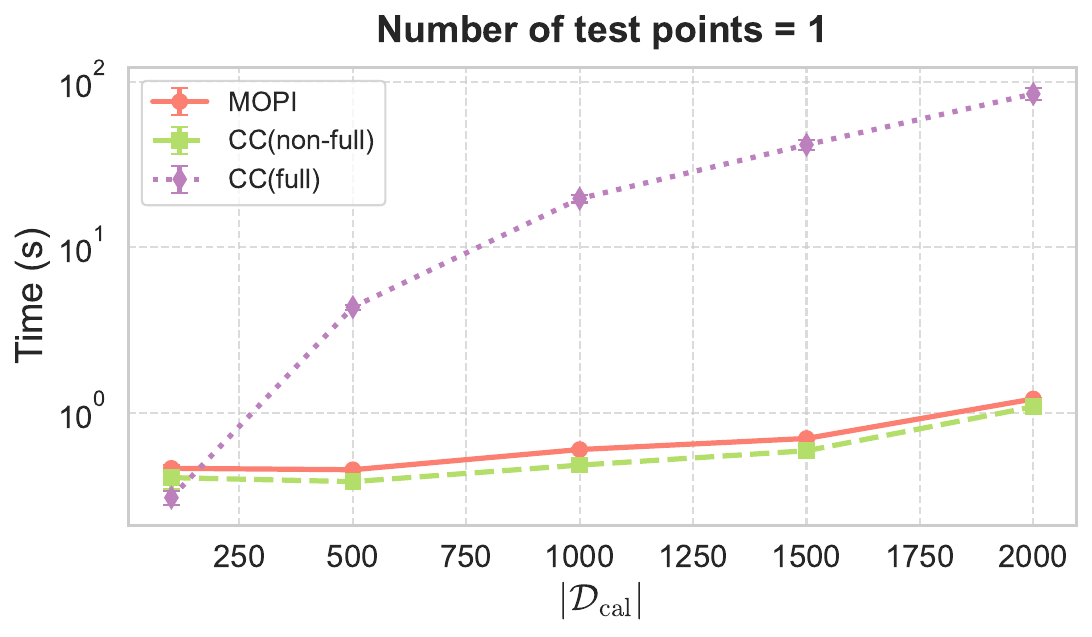}
    \caption{Comparison of the computational efficiency of MOPI and CC with a single test point. Dots and error bars show means and confidence intervals from 100 trials.}
    \label{fig:compute_time}
\end{figure}

As shown in Figure~\ref{fig:compute_time}, the computational time of \texttt{MOPI} is comparable to that of \texttt{CC(non-full)}, since both methods solve an optimization problem solely on the calibration dataset. 
In contrast, \texttt{CC(full)}, as a full conformal method, requires solving the optimization problem jointly over the calibration data and each test point, resulting in a computational cost that is several times higher—typically 10$\times$ to 100$\times$ that of \texttt{MOPI} and \texttt{CC(non-full)}. 
This observation is broadly consistent with the simulation findings reported in \citet{duchi2025sampleconditional}.

\section{Proofs of Main Results}
In the following, we denote $\epsilon^2 := \mathsf{MSCE}(C^{\rm ora})$ and consider $\fC=\{C(x;h):h\in\gH\}$ be the structure prediction set \eqref{eq:structure_set}. 

\subsection{Proof of Proposition \ref{pro:CP_LS2minimax_general}}
Proposition \ref{pro:CP_LS2minimax_general} is the special case of Lemma \ref{lemma:CP_LS2minimax_general} with $\lambda=1$.

\begin{lemma}\label{lemma:CP_LS2minimax_general}
    Let $\Psi_{\lambda}(C,f) = \E\LRm{f(Z)\LRs{\one\{Y\notin C(X)\}-\alpha}-\lambda f^2(Z)}$. Then we have $\Psi_{\lambda}(C,f)\le \frac{1}{4\lambda}\E\LRm{\LRs{\alpha(Z;C) - \alpha}^2}$. 
    If for any $C\in \mathfrak{C}$, there exists some $f^*\in \gF$ such that $\alpha(Z;C) - \alpha = 2\lambda f^*(Z)$, then it holds that $\max_{f\in \gF}\Psi_{\lambda}(C,f) = \frac{1}{4\lambda}\E\LRm{\LRs{\alpha(Z;C) - \alpha}^2}$ for any $C\in \mathfrak{C}$.
    In addition, we also have
    \begin{align}\label{eq:minimax_obj_difference_general}
        \mathsf{MSCE}(C^{*}) = \mathsf{MSCE}(C^{\rm ora}).
    \end{align}
\end{lemma}
\begin{proof}
    By the definition of $\Psi_{\lambda}$, we have
    \begin{align}\label{eq:obj_minus_general}
        \max_{f\in \gF} \Psi_{\lambda}(C,f) &= \max_{f\in \gF} \E\LRm{f(Z)(\mathbbm{1}\{Y\notin {C}(X)\} - \alpha) - \lambda f^2(Z)}\nonumber\\
        &= \max_{f\in \gF} \E\LRm{f(Z)(\alpha(Z;C) - \alpha) - \lambda f^2(Z)}\nonumber\\
        &= \frac{1}{4\lambda}\E\LRm{\LRs{\alpha(Z;C) - \alpha}^2} - \frac{1}{\lambda}\min_{f\in \gF} \E\LRm{\LRs{\frac{1}{2}(\alpha(Z;C) - \alpha) - \lambda f(Z)}^2}
    \end{align}
    According to the assumption, we know for any $C\in\fC$, $\exists f^*\in \gF$ such that $\alpha(Z;C) - \alpha= 2\lambda f^*(Z)$. Therefore,
    \begin{align}
        0\le \min_{f\in \gF} \E\LRm{\LRs{\frac{1}{2}(\alpha(Z;C) - \alpha) - \lambda f(Z)}^2}\le \E\LRm{\LRs{\frac{1}{2}(\alpha(Z;C) - \alpha) - \lambda f^*(Z)}^2}=0\nonumber.
    \end{align}
    Hence by \eqref{eq:obj_minus_general} we have $\max_{f\in \gF} \Psi_{\lambda}(C,f)=\frac{1}{4\lambda}\E\LRm{\LRs{\alpha(Z;C) - \alpha}^2}$ for all $C\in\fC$. Taking $\argmin$ on both side w.r.t. $C$,
    \begin{align}
        C^{*}:=\argmin_{C\in\fC}\max_{f\in \gF} \Psi_{\lambda}(C,f)=\argmin_{C\in\fC}\frac{1}{4\lambda}\E\LRm{\LRs{\alpha(Z;C) - \alpha}^2}=: C^{\rm ora},\nonumber
    \end{align}
    which indicate $\mathsf{MSCE}(C^{*}) = \mathsf{MSCE}(C^{\rm ora})$.
\end{proof}


\subsection{Proof of closed forms of inner maximization}\label{appen:cloded_forms}
We consider the closed forms of inner maximization of the empirical minimax problem \eqref{eq:sample_minimax_CP_lambda} and state the following two Lemmas.

\begin{lemma}\label{lemma:inner_close_form_finite_lambda}
    {Assume $\sum_{i=1}^n\one\{Z_i=z\}>0$ for every $z\in\gZ$.} Let $\gF = \{\sum_{z\in \gZ} \beta_z\mathbbm{1}\{Z = z\}: \beta_z \in \sR, z\in \gZ\}$. The inner maximization in the problem \eqref{eq:sample_minimax_CP_lambda} with $\gamma = 0$ is:
    \begin{align}
    \max_{f\in \gF}\widehat{\Psi}_\lambda(C,f) = \sum_{z\in\gZ}\frac{\big(\sum_{i=1}^n\mathbbm{1}\{Z_i=z\}\LRs{\mathbbm{1}\{Y_i\notin C(X_i)\}-\alpha}\big)^2}{4\lambda n\sum_{i=1}^n\mathbbm{1}\{Z_i=z\}}.\nonumber
    \end{align}
\end{lemma}

\begin{lemma}\label{lemma:inner_close_form_rkhs_lambda}
    Let $\gF =\sK$, where $\sK$ is a RKHS equipped with a kernel function $\gK(\cdot,\cdot):\gZ \times \gZ \to \sR$. Denote $\boldsymbol{\varphi}_n(C) = (\varphi_1(C),\ldots,\varphi_n(C))^\top$ where $\varphi_i(C) = \frac{1}{n}\bigl( \mathbbm{1}\{ Y_i\notin C(X_i) \} - \alpha \bigr)$ for $i=1,\ldots,n$ and $\mK_n = (\gK(Z_i,Z_j))_{i,j=1}^n$ be the empirical kernel matrix. Then inner maximization in the problem \eqref{eq:sample_minimax_CP_lambda} reduces to:
    \begin{align}
        \max_{f\in \gF} \LRl{\widehat{\Psi}_\lambda(C, f)  - \gamma\|f\|_{\gF}^2} = 
       \frac{1}{4\lambda}\boldsymbol{\varphi}_n(C)^\top \mK_n \left( \frac{1}{n} \mK_n +{\gamma/\lambda} \mI_n \right)^{-1} \boldsymbol{\varphi}_n(C),\nonumber
    \end{align}
    where $\mI_n$ is an $n\times n$ identity matrix.
\end{lemma}

\begin{proof}[Proof of Lemma \ref{lemma:inner_close_form_finite_lambda}]
    Since $\gF=\{\sum_{z\in\gZ}\beta_z \one\{Z=z\}:\beta_z\in\R, z\in\gZ\}$, let $\boldsymbol{\beta}=(\beta_1,\ldots,\beta_{|\gZ|})^\top$, $\boldsymbol{1}(Z)=(\one\{Z=z_1\},\ldots,\one\{Z=z_{|\gZ|}\})^\top$ and $\varphi_i(C) = \frac{1}{n}\bigl( \mathbbm{1}\{ Y_i\notin C(X_i) \} - \alpha \bigr)$, then by choosing $\gamma=0$ we can rewrite the inner problem of (8) as:
    \begin{align*}
        \sup_{f\in\gF}\widehat{\Psi}_\lambda(C,f)&=\sup_{{\boldsymbol\beta}\in\R^{|\gZ|}}\frac{1}{n}\sum_{i=1}^n \LRm{\boldsymbol{\beta}^\top\boldsymbol{1}(Z_i)n\varphi_i(C)-\lambda\boldsymbol{\beta}^\top\boldsymbol{1}(Z_i)\boldsymbol{1}(Z_i)^\top\boldsymbol{\beta}}\\
        &:=\sup_{{\boldsymbol\beta}\in\R^{|\gZ|}}\boldsymbol{\beta}^\top\LRs{\sum_{i=1}^n \boldsymbol{1}(Z_i)\varphi_i(C)} - \lambda\boldsymbol{\beta}^\top\LRs{\frac{1}{n}\sum_{i=1}^n\boldsymbol{1}(Z_i)\boldsymbol{1}(Z_i)^\top}\boldsymbol{\beta},
    \end{align*}
    which is a quadratic form w.r.t. $\boldsymbol{\beta}$. Moreover, since for $z,z^\prime\in\gZ$, $z\ne z^\prime$, we have $\one\{Z=z\}\cdot\one\{Z=z^\prime\}=0$. Hence, 
    $$
    \frac{1}{n}\sum_{i=1}^n\boldsymbol{1}(Z_i)\boldsymbol{1}(Z_i)^\top=\mathrm{Diag}\LRs{\frac{1}{n}\sum_{i=1}^n\one\{Z_i=z_1\},\ldots,\frac{1}{n}\sum_{i=1}^n\one\{Z_i=z_{|\gZ|}\}}.
    $$
    Assume $n$ is sufficiently large such that $\sum_{i=1}^n\one\{Z_i=z\}>0$ for all $z\in\gZ$.  Taking the first order condition, we have the optimizer:
    \begin{align*}
        \boldsymbol{\beta}^*=\frac{1}{2\lambda}\LRs{\frac{1}{n}\sum_{i=1}^n\boldsymbol{1}(Z_i)\boldsymbol{1}(Z_i)^\top}^{-1}\LRs{\sum_{i=1}^n \boldsymbol{1}(Z_i)\varphi_i(C)}
    \end{align*}
    Thus, 
    \begin{align*}
        \sup_{f\in\gF} \widehat{\Psi}(C,f) &=\sup_{\boldsymbol{\beta}\in\R^{|\gZ|}}\boldsymbol{\beta}^\top\LRs{\sum_{i=1}^n \boldsymbol{1}(Z_i)\varphi_i(C)} - \lambda\boldsymbol{\beta}^\top\LRs{\frac{1}{n}\sum_{i=1}^n\boldsymbol{1}(Z_i)\boldsymbol{1}(Z_i)^\top}\boldsymbol{\beta}\\
        &={\sum_{z\in\gZ}\frac{\big(\sum_{i=1}^n\mathbbm{1}\{Z_i=z\}\varphi_i(C)\big)^2}{4\lambda\LRs{\frac{1}{n}\sum_{i=1}^n\mathbbm{1}\{Z_i=z\}}}}\\
        &={\sum_{z\in\gZ}\frac{\big(\sum_{i=1}^n\mathbbm{1}\{Z_i=z\}\LRs{\mathbbm{1}\{ Y_i\notin C(X_i) \} - \alpha}\big)^2}{4\lambda n\sum_{i=1}^n\mathbbm{1}\{Z_i=z\}}}
    \end{align*}
\end{proof}

\begin{proof}[Proof of Lemma \ref{lemma:inner_close_form_rkhs_lambda}]
    Since we let $\gF=\sK$ and the inner maximization of \eqref{eq:sample_minimax_CP_lambda} takes form 
    \begin{align*}
        \sup_{f\in\gF} \widehat{\Psi}_\lambda(C,f)-\gamma\|f\|_\gF^2.
    \end{align*}
    By the generalized representer theorem (E.g., \cite{scholkopf2001generalized}, Theorem 1), the optimal solution of the inner maximization of \eqref{eq:sample_minimax_CP_lambda} takes form 
    \begin{align*}
        f^*(z)=\sum_{i=1}^n\alpha^*_i\gK(Z_i,z)
    \end{align*}
    for some weight $\boldsymbol{\alpha}^*\in\R^n$.
    Let $\mK_n = (\gK(Z_i,Z_j))_{i,j=1}^n$ be the empirical kernel matrix. For any $\boldsymbol{\alpha}=(\alpha_1,\ldots,\alpha_n)^\top\in\R^n$, consider a function $f(z)=\sum_{i=1}^n\alpha\gK(Z_i,z)$, then we have $\|f\|_{\sK}^2=\boldsymbol{\alpha}^\top\mK_n\boldsymbol{\alpha}$, $f(Z_i)=\boldsymbol{e}_i^\top\mK_n\boldsymbol{\alpha}$, where $\boldsymbol{e}_i \in \{0,1\}^n$ is the vector with a $1$ in the $i$-th position and $0$ elsewhere. Then 
    \begin{align*}
        \frac{1}{n}\sum_{i=1}^n f(Z_i)=\frac{1}{n}\sum_{i=1}^n \boldsymbol{\alpha}^\top\mK_n\boldsymbol{e}_i \boldsymbol{e}_i^\top\mK_n\boldsymbol{\alpha}=\frac{1}{n}\boldsymbol{\alpha}^\top\mK_n^2\boldsymbol{\alpha}.
    \end{align*}    
    Let $\boldsymbol{\varphi}_n(C) = (\varphi_1(C),\ldots,\varphi_n(C))^\top$ where $\varphi_i(C) = \frac{1}{n}\bigl( \mathbbm{1}\{ Y_i\notin C(X_i) \} - \alpha \bigr)$ for $i=1,\ldots,n$. Then we can rewrite the inner problem of \eqref{eq:sample_minimax_CP_lambda} as:
    \begin{align*}
        \sup_{f\in\gF}\widehat{\Psi}_\lambda(C,f) - \gamma\|f\|_{\gF}^2= \sup_{\boldsymbol{\alpha}\in\R^n} \boldsymbol{\varphi}_n^\top(C)\mK_n\boldsymbol{\alpha}-\boldsymbol{\alpha}^\top\mK_n\LRs{\frac{\lambda}{n}\mK_n+\gamma\mI_n}\boldsymbol{\alpha}.
    \end{align*}
    Taking the first order condition, we have the optimizer:
    \begin{align*}
        \boldsymbol{\alpha}^*=\frac{1}{2}\LRs{\frac{\lambda}{n}\mK_n+\gamma\mI_n}^{-1}\boldsymbol{\varphi}_n(C),
    \end{align*}
    and the optimal value:
    \begin{align*}
        \frac{1}{4}\boldsymbol{\varphi}_n(C)^\top\mK_n\LRs{\frac{\lambda}{n}\mK_n+\gamma\mI_n}^{-1}\boldsymbol{\varphi}_n(C)=
        \frac{1}{4\lambda}\boldsymbol{\varphi}_n(C)^\top\mK_n\LRs{\frac{1}{n}\mK_n+\gamma/\lambda\mI_n}^{-1}\boldsymbol{\varphi}_n(C).
    \end{align*}
\end{proof}

\subsection{Proof of Theorem \ref{thm:mse}}

In the following, we denote $\epsilon^2 := \mathsf{MSCE}(C^{\rm ora})$, $\widehat{\Psi}_\gamma(C,f)=\widehat{\Psi}(C,f)-\gamma\|f\|_\gF^2$. Consider $\fC=\{C(x;h):h\in\gH\}$ be the structure prediction set \eqref{eq:structure_set} and the following more general penalized minimax problem:
\begin{align}\label{eq:sample_minimax_CP_general}
    \widehat{C}:=\argmin_{C\in\fC}\max_{f\in\gF}\LRl{\widehat{\Psi}_\gamma(C,f)+\nu\|C\|_{\gH}^2}.
\end{align}
Thus, Eq. \eqref{eq:sample_minimax_CP} is the special case of \eqref{eq:sample_minimax_CP_general} with $\nu=0$, and 
Theorem \ref{thm:mse} is the special case of the following Theorem \ref{thm:mse_complete}.
\begin{theorem}\label{thm:mse_complete}
    Suppose the Assumptions \ref{assum:shape_F}, \ref{assum:appro_error} holds.
    Furthermore, assume the constant $U$ satisfies $\sup_{f \in \gF_U} \|f\|_{\infty} \leq 1$. Let $\delta_{n,\gF_U}$ and $\delta_{n,\gG}$ be the critical radius of $\gF_U$ and $\gG$, respectively. For any $\zeta \in (0,1)$, let $\tilde{\delta}_{n,\gF} := \delta_{n,\gF_{U}} + \sqrt{\frac{\log(c_1 / \zeta)}{c_2 n}}$ for some positive constants $c_1$ and $c_2$. 
    \begin{itemize}
        \item For finite-dimensional $\gF$, let the hyperparameter $\gamma = 0$ in \eqref{eq:sample_minimax_CP_general}. Then, for sufficiently large $n$, such that $\tilde{\delta}_{n,\gF}^2 \|f\|_\gF^2  \leq \frac{U}{2}\|f\|_{L_2}^2$ for any $f \in \gF$, conditioning on the labeled data $\gD_n$, with probability at least $1 - 4\zeta$, we have:
        \begin{align}
            \E\left[\left(\alpha(Z; \widehat{C}) - \alpha\right)^2 \mid \gD_n\right]
            &\lesssim \epsilon^2 + \eta^2 + \delta_{n,\gF_U}^2 + \delta_{n,\gG}^2\nonumber\\
            &\qquad+ \frac{\log\left(\log_2\left(1/\delta_{n,\gG}\right) / \zeta\right)}{n} + \frac{\nu^2 \|C^{\rm ora}\|_{\gH}^4}{\tilde{\delta}_{n,\gF}^2}.\nonumber
        \end{align}

        \item For infinite-dimensional $\gF$, if we choose the hyperparameter $\gamma$ in \eqref{eq:sample_minimax_CP_general} such that $\gamma \|f\|_\gF^2 \leq \|f\|_{L_2}^2$, then for sufficiently large $n$, such that $\tilde{\delta}_{n,\gF}^2 \|f\|_\gF^2 \leq \frac{U}{2}\|f\|_{L_2}^2$ for any $f \in \gF$, conditioning on the labeled data $\gD_n$, with probability at least $1 - 4\zeta$, we have:
    \begin{align}
        \E\left[\left(\alpha(Z; \widehat{C}) - \alpha\right)^2 \mid \gD_n\right]
        &\lesssim \epsilon^2 + \eta^2 + \delta_{n,\gF_U}^2 + \delta_{n,\gG}^2\nonumber\\
        &\qquad+ \frac{\log\left(\log_2\left(1/\delta_{n,\gG}\right) / \zeta\right)}{n} + \frac{\nu^2 \|C^{\rm ora}\|_{\gH}^4}{\tilde{\delta}_{n,\gF}^2}.\nonumber
    \end{align}
    \end{itemize}
    
\end{theorem}

\begin{proof}[Proof of Theorem \ref{thm:mse_complete}]
    Let $\widehat{C} = \argmin_{C\in \fC}\max_{f\in \gF} \widehat{\Psi}_{\gamma}(C, f)+\nu\|C\|_{\gH}^2$. Notice that
    \begin{align}
        \E\LRm{\LRs{\alpha(Z;\widehat{C}) - \alpha}^2} \leq 2\epsilon^2 + 2\E\LRm{\LRs{\alpha(Z;\widehat{C}) - \alpha(Z; C^{\rm ora})}^2},
    \end{align}
    where $\E[(\alpha(Z;C^{\rm ora})-\alpha)^2] = \epsilon^2$. 
    Let $\gF_{U} = \{f\in \gF: \|f\|_{\gF}^2 \leq U\}$ and $\delta_{n,\gF_{U}}$ be the critical radius of $\gF_U$. Let $Pf^2:=\|f\|_{L_2}^2=\E_P[f^2(Z)],\;P_nf^2:=\|f\|_{L_2,n}^2=\frac{1}{n}\sum_{i=1}^nf^2(Z_i)$.
    By Lemma \ref{lemma:quadratic_EP_bound} and our assumption $\sup_{f\in \gF_{U}}\|f\|_{\infty} \leq 1$, for any $t \geq \delta_{n,\gF_{U}}$, we have
    \begin{align}
        \sP\LRs{\forall f\in \gF_{U}, \left|P_n f^2 - P f^2 \right| \leq \frac{P f^2}{2} + \frac{t^2}{2}}\geq 1 - c_1 e^{-c_2 n t^2},\nonumber
    \end{align}
    where $c_1,c_2 > 0$ are universal constants. For $f\notin \gF_U$, we apply the same exponential inequality to $\sqrt{U} f/\|f\|_{\gF}$, then we have
    \begin{align}
        \sP\LRs{\forall f\in \gF, \left|P_n f^2 - P f^2 \right| \leq \frac{P f^2}{2} + \frac{t^2 \|f\|_{\gF}^2}{2U}}\geq 1 - c_1 e^{-c_2 n t^2}.\nonumber
    \end{align}
    Combining the results above and taking $ t = \delta_{n,\gF_{U}} + \sqrt{\frac{\log(c_1/\zeta)}{c_2 n}}$ for any $\zeta > 0$, we have w.p. $1-\zeta$,
    \begin{align*}
        \forall f\in \gF, \left|P_n f^2 - P f^2 \right| \leq \frac{P f^2}{2} + \frac{1\vee (\|f\|_{\gF}^2/U) }{2}\LRs{ \delta_{n,\gF_{U}} + \sqrt{\frac{\log(c_1/\zeta)}{c_2 n}}}^2.
    \end{align*}
    Since we choose $\tilde{\delta}_{n,\gF}=\delta_{n,\gF_{U}} + \sqrt{\frac{\log(c_1/\zeta)}{c_2 n}}$,
    and assume $\tilde{\delta}_{n,\gF}^2 \|f\|_\gF^2  \leq \frac{U}{2}\|f\|_{L_2}^2$,
    then w.p. $1-\zeta$:
    \begin{align}
    \label{eq:variance_empi_process_lower}
        \forall f\in \gF: \gamma\|f\|_{\gF}^2+\|f\|_{L_2,n}^2 & \ge\gamma\|f\|_{\gF}^2+\LRs{\frac{\|f\|_{L_2}^2}{2}-\LRs{1\vee \frac{\|f\|_{\gF}^2}{U}}\frac{\tilde{\delta}_{n,\gF}^2}{2}}\nonumber\\
        &\ge {\gamma}\|f\|_{\gF}^2+\frac{1}{4}\|f\|_{L_2}^2-\frac{\tilde{\delta}_{n,\gF}^2}{2},
    \end{align}
    and
    \begin{align}
    \label{eq:variance_empi_process_upper}
        \forall f\in \gF: \gamma\|f\|_{\gF}^2+\|f\|_{L_2,n}^2 &\le\gamma\|f\|_{\gF}^2+\LRs{\frac{3\|f\|_{L_2}^2}{2}+\LRs{1\vee \frac{\|f\|_{\gF}^2}{U}}\frac{\tilde{\delta}_{n,\gF}^2}{2}}\nonumber\\
        &\le {\gamma}\|f\|_{\gF}^2+\frac{7}{4}\|f\|_{L_2}^2+\frac{\tilde{\delta}_{n,\gF}^2}{2}.
    \end{align}
    For $f\in \gF, C\in \fC$, we denote the function
    \begin{align}
        \psi(C, f) &= f(Z) \LRs{\mathbbm{1}\{Y\not\in C(X)\} - \alpha},\nonumber
    \end{align}
    Since $\psi(C, f)$ is 1-Lipschitz w.r.t. $f$, using Lemma 14 in \citet{foster2023orthogonal}, then w.p. $1-\zeta$,
    \begin{align}\label{eq:multiaccu_empi_process_bound}
         \forall f\in \gF,  \left|(P_n-P) {\psi}(C^{\rm ora}, f)\right| \leq 18 \LRs{\tilde{\delta}_{n,\gF} \|f\|_{L_2} + \LRs{1\vee \frac{\|f\|_{\gF}}{\sqrt{U}}}\tilde{\delta}_{n,\gF}^2}.
    \end{align}
    Note that $\Psi_{\lambda}(C,f) = P \psi(C,f) - \lambda P f^2$ and $\widehat{\Psi}_{\gamma}(C,f) = P_n \psi(C,f) -  P_n f^2 - \gamma\|f\|_{\gF}^2$ for any $C\in \fC, f\in \gF$.
    By the optimality of $\widehat{C}$.
    \begin{align}\label{eq:upper_hhat_horacle}
        \sup_{f\in\gF}\widehat{\Psi}_{\gamma}(\widehat{C},f) &\leq \sup_{f\in \gF}\widehat{\Psi}_{\gamma}(C^{\rm ora},f) +\nu\LRs{\|C^{\rm ora}\|_{\gH}^2-\|\widehat{C}\|_{\gH}^2}
    \end{align}
    Combining \eqref{eq:variance_empi_process_lower} and \eqref{eq:multiaccu_empi_process_bound}, for $\gamma\geq 0$, we can get w.p. $1-2\zeta$:
    \begin{align}\label{eq:empirical_sup_upper_general}
        \sup_{f\in \gF} &\widehat{\Psi}_{\gamma}(C^{\rm ora},f) =\sup_{f\in \gF}\LRl{P_n \psi(C^{\rm ora},f) -  P_n f^2-\gamma\|f\|_\gF^2}\nonumber\\
        &\leq \sup_{f\in \gF}\LRl{P_n \psi(C^{\rm ora},f) - \LRs{\gamma\|f\|_{\gF}^2+\frac{1}{4}\|f\|_{L_2}^2-\frac{\tilde{\delta}_{n,\gF}^2}{2}}}\nonumber\\
        &\leq \sup_{f\in \gF}\LRl{P \psi(C^{\rm ora}, f) + 18 \LRs{\tilde{\delta}_{n,\gF} \|f\|_{L_2} + \LRs{1\vee \frac{\|f\|_{\gF}}{\sqrt{U}}}\tilde{\delta}_{n,\gF}^2} - \LRs{\gamma\|f\|_{\gF}^2+\frac{1}{4}\|f\|_{L_2}^2-\frac{\tilde{\delta}_{n,\gF}^2}{2}}} \nonumber\\
        &\overset{(i)}{\leq} \sup_{f\in \gF}\LRl{P \psi(C^{\rm ora}, f) + 18 \LRs{2\tilde{\delta}_{n,\gF} \|f\|_{L_2} + \tilde{\delta}_{n,\gF}^2} - \LRs{\frac{1}{4}\|f\|_{L_2}^2-\frac{\tilde{\delta}_{n,\gF}^2}{2}}} \nonumber\\
        &\leq \sup_{f\in \gF}\LRl{P \psi(C^{\rm ora}, f) -\frac{1}{8}P f^2} + \sup_{f\in \gF}\LRl{36\tilde{\delta}_{n,\gF} \|f\|_{L_2} - \frac{1}{8}\|f\|_{L_2}^2} +18\tilde{\delta}_{n,\gF}^2+\frac{\tilde{\delta}_{n,\gF}^2}{2}\nonumber\\
        &\overset{(ii)}{\leq} \sup_{f\in \gF}\Psi_{\frac{1}{8}}(C^{\rm ora}, f) + {2\cdot36^2\tilde{\delta}_{n,\gF}^2}  + 18\tilde{\delta}_{n,\gF}^2+\frac{\tilde{\delta}_{n,\gF}^2}{2}\nonumber\\
        &\leq \sup_{f\in \gF}\Psi_{\frac{1}{8}}(C^{\rm ora}, f) + {2611\tilde{\delta}_{n,\gF}^2},
    \end{align}
    where (i) we used the assumption $\tilde{\delta}_{n,\gF} \|f\|_{\gF} \leq \sqrt{\frac{U}{2}}\|f\|_{L_2}$ and $\gamma\geq 0$, and (ii) we used the fact $\sup_{f\in\gF}(a\|f\|-b\|f\|^2)\le\frac{a^2}{4b}$ for any norm $\|\cdot\|$ and $a,b> 0$.
    In addition, we also have the lower bound
    \begin{align}
        \sup_{f\in \gF} \widehat{\Psi}_{\gamma}(\widehat{C},f) &= \sup_{f\in \gF} \LRl{P_n {\psi}(\widehat{C}, f) - P_n {\psi}(C^{\rm ora}, f) + P_n {\psi}(C^{\rm ora}, f) - P_n f^2 - \gamma\|f\|_{\gF}^2}\nonumber\\
        &\geq \sup_{f\in \gF}\LRl{P_n {\psi}(\widehat{C}, f) - P_n {\psi}(C^{\rm ora}, f) - 2\LRs{ P_n f^2 + \gamma\|f\|_{\gF}^2}}\nonumber\\
        &\qquad +\inf_{f\in \gF} \LRl{P_n {\psi}(C^{\rm ora}, f) +  P_n f^2 + \gamma\|f\|_{\gF}^2}\nonumber\\
        &= \sup_{f\in \gF}\LRl{P_n {\psi}(\widehat{C}, f) - P_n {\psi}(C^{\rm ora}, f) - 2\LRs{ P_n f^2 + \gamma\|f\|_{\gF}^2}} - \sup_{f\in \gF}\widehat{\Psi}_{\gamma}(C^{\rm ora}, f),\nonumber
    \end{align}
    where we used the fact that $\gF$ is symmetric. It follows that w.p. $1-2\zeta$:
    \begin{align}\label{eq:penalized_EP_upper}
        &\sup_{f\in \gF}\LRl{P_n {\psi}(\widehat{C}, f) - P_n {\psi}(C^{\rm ora}, f) - 2\LRs{ P_n f^2 + \gamma\|f\|_{\gF}^2}} \nonumber\\
        &\qquad \leq \LRl{\sup_{f\in \gF} \widehat{\Psi}_{\gamma}(\widehat{C},f) + \sup_{f\in \gF} \widehat{\Psi}_{\gamma}(C^{\rm ora},f)}\nonumber\\
        &\qquad\overset{\rm (i)}{\leq} 2\sup_{f\in \gF} \widehat{\Psi}_{\gamma}(C^{\rm ora},f)+{\nu}\LRs{\|C^{\rm ora}\|_{\gH}^2-\|\widehat{C}\|_{\gH}^2}\nonumber\\
        &\qquad\overset{\rm (ii)}{\leq} 2\sup_{f\in \gF}\Psi_{\frac{1}{8}}(C^{\rm ora}, f) + {2\cdot 2611\tilde{\delta}_{n,\gF}^2} + {\nu}\LRs{\|C^{\rm ora}\|_{\gH}^2-\|\widehat{C}\|_{\gH}^2}\nonumber\\
        &\qquad\overset{\rm (iii)}{\leq} {4\epsilon^2} + {2\cdot 2611\tilde{\delta}_{n,\gF}^2} +{\nu}\|C^{\rm ora}\|_{\gH}^2,
    \end{align}
    where (i) holds due to \eqref{eq:upper_hhat_horacle}, (ii) holds due to \eqref{eq:empirical_sup_upper_general} and (iii) holds due to Lemma \ref{lemma:CP_LS2minimax_general} and $\E[(\alpha(Z;C^{\rm ora}) - \alpha)^2] = \epsilon^2$.
    
    For any $C\in \fC$, we define
    \begin{align}\label{eq:def_fh}
        f_C = \argmin_{f\in \gF_U} \E \LRm{(f(Z) - (\alpha(Z;C)-\alpha(Z;C^{\rm ora})))^2}.
    \end{align}
    By the assumption $\E \LRm{(f_C(Z) - (\alpha(Z;C) - \alpha(Z;C^{\rm ora})))^2} \leq \eta^2$ for any $C\in \fC$, we have
    \begin{align}
        P \psi(C, f_C) - P \psi(C^{\rm ora}, f_C) &= 
        \E\LRm{f_C(Z) (\alpha(Z;C) - \alpha(Z;C^{\rm ora}))}\nonumber\\
        &= \E\LRm{f_C^2(Z)} + \E\LRm{f_C(Z) (\alpha(Z;C) - \alpha(Z;C^{\rm ora}) - f_C(Z))}\nonumber\\
        &\geq \|f_C\|_{L_2}^2 - \|f_C\|_{L_2} \|\alpha(C) - \alpha(C^{\rm ora}) - f_C\|_{L_2}\nonumber\\
        &\geq \|f_C\|_{L_2}^2 - \|f_C\|_{L_2} \eta,\nonumber
    \end{align}
    where the first inequality follows from Cauchy's inequality.
    It follows that
    \begin{align}\label{eq:pop_weightcover_diff}
        \frac{P \psi(C, f_C) - P \psi(C^{\rm ora}, f_C)}{\|f_C\|_{L_2}} \geq \|f_C\|_{L_2} - \eta \geq \|\alpha(C) - \alpha(C^{\rm ora})\|_{L_2} - 2\eta.
    \end{align}
    Let $r \in [0,1]$, then $r f_{\widehat{C}} \in \gF_U$ since $\gF_U$ is star-shaped.
    Let $v^2(C, C^{\rm ora}) := \E[|\alpha(Z;C) - \alpha(Z;C^{\rm ora})|^2]$ and define the function class
    \begin{align}\label{eq:def_G}
        \gG = \LRl{(x,z,y)\mapsto f_C(z)(\one\{y\notin C(x)\}-\one\{y\notin C^{\rm ora}(x)\}): C\in \fC}.
    \end{align}
    Then we have w.p. $1-2\zeta$:
    \begin{align}\label{eq:excess_risk_lower}
        &\sup_{f\in \gF}\LRl{P_n {\psi}(\widehat{C}, f) - P_n {\psi}(C^{\rm ora}, f) - 2\LRs{ P_n f^2 + \gamma\|f\|_{\gF}^2}}\nonumber\\
        & \geq r P_n {\psi}(\widehat{C}, f_{\widehat{C}}) - r P_n {\psi}(C^{\rm ora}, f_{\widehat{C}}) - 2r^2 \LRs{ P_n f_{\widehat{C}}^2 + \gamma\|f_{\widehat{C}}\|_{\gF}^2}\nonumber\\
        &\overset{\rm (i)}{\geq}  r P_n {\psi}(\widehat{C}, f_{\widehat{C}}) - r P_n {\psi}(C^{\rm ora}, f_{\widehat{C}}) - 2r^2  \LRs{{\gamma}\|f_{\widehat{C}}\|_{\gF}^2+\frac{7}{4}\|f_{\widehat{C}}\|_{L_2}^2+\frac{\tilde{\delta}_{n,\gF}^2}{2}} \nonumber\\
        &\overset{\rm (ii)}{\geq} r \LRl{P {\psi}(\widehat{C}, f_{\widehat{C}}) - P {\psi}(C^{\rm ora}, f_{\widehat{C}}) - v(\widehat{C},C^{\rm ora})\LRs{4 \delta_{n,\gG} + 2\sqrt{\frac{2 \log(R_n/\zeta)}{n}}} - \frac{\log(R_n/\zeta)}{3n}}\nonumber\\
        &\qquad-2r^2  \LRs{{\gamma}\|f_{\widehat{C}}\|_{\gF}^2+\frac{7}{4}\|f_{\widehat{C}}\|_{L_2}^2+\frac{\tilde{\delta}_{n,\gF}^2}{2}} \nonumber\\
        &\overset{\rm (iii)}{\geq} r \LRl{(v(\widehat{C},C^{\rm ora}) - 2\eta)\|f_{\widehat{C}}\|_{L_2} - v(\widehat{C},C^{\rm ora})\LRs{4 \delta_{n,\gG} + 2\sqrt{\frac{2 \log(R_n/\zeta)}{n}}} - \frac{\log(R_n/\zeta)}{3n}}\nonumber\\
        &\qquad -2r^2  \LRs{{\gamma}\|f_{\widehat{C}}\|_{\gF}^2+\frac{7}{4}\|f_{\widehat{C}}\|_{L_2}^2+\frac{\tilde{\delta}_{n,\gF}^2}{2}},
    \end{align}
    where (i) use \eqref{eq:variance_empi_process_upper}; (ii) holds due to Lemma \ref{lemma:weight_cover_EP_bound}; and (iii) holds due to \eqref{eq:pop_weightcover_diff}. $R_n = \lceil\log_2(1/\delta_{n,\gG})\rceil+1$.
    If $4 \delta_{n,\gG} + 2\sqrt{\frac{2 \log(R_n/\zeta)}{n}} \geq \|f_{\widehat{C}}\|_{L_2}/2$, we have
    \begin{align}
        v(\widehat{C}, C^{\rm ora}) \leq \|f_{\widehat{C}}\|_{L_2} + \eta \lesssim \delta_{n,\gG} + \sqrt{\frac{\log(R_n/\zeta)}{n}} + \eta.\nonumber
    \end{align}
    Otherwise, i.e., $4 \delta_{n,\gG} + 2\sqrt{\frac{2 \log(R_n/\zeta)}{n}} \leq \|f_{\widehat{C}}\|_{L_2}/2$, combining \eqref{eq:excess_risk_lower} and \eqref{eq:penalized_EP_upper}, w.p. $1-4\zeta$, we get
    \begin{align}
        v(\widehat{C},C^{\rm ora}) &\leq  \frac{2}{r \|f_{\widehat{C}}\|_{L_2}}\LRs{{4\epsilon^2} + {2\cdot2611\tilde{\delta}_{n,\gF}^2} +{\nu}\|C^{\rm ora}\|_{\gH}^2} + 2\eta\nonumber\\
        &\qquad + \frac{2\log(R_n/\zeta)}{3n \|f_{\widehat{C}}\|_{L_2}} + \frac{4r}{\|f_{\widehat{C}}\|_{L_2}}\LRs{{\gamma}\|f_{\widehat{C}}\|_{\gF}^2+\frac{7}{4}\|f_{\widehat{C}}\|_{L_2}^2+\frac{\tilde{\delta}_{n,\gF}^2}{2}}.\nonumber
    \end{align}
    If $\|f_{\widehat{C}}\|_{L_2} \geq \tilde{\delta}_{n,\gF} \geq \sqrt{\frac{\log(c_1/\zeta)}{c_2 n}}$ , choosing $r = \max\{\epsilon, \tilde{\delta}_{n,\gF}\}/\|f_{\widehat{C}}\|_{L_2}$,  by $\gamma\|f\|_\gF^2\leq \|f\|_{L_2}^2$, we have
     \begin{align}
         v(\widehat{C},C^{\rm ora}) &\lesssim {\epsilon} +{\tilde{\delta}_{n,\gF}} +\frac{\nu\|C^{\rm ora}\|_{\gH}^2}{\tilde{\delta}_{n,\gF}}+ \eta + \sqrt{\frac{\log(R_n/\zeta)}{n}} + \tilde{\delta}_{n,\gF}+\tilde{\delta}_{n,\gF} + \tilde{\delta}_{n,\gF}\nonumber\\
         &\lesssim {\epsilon} +{\tilde{\delta}_{n,\gF}} + \frac{\nu\|C^{\rm ora}\|_{\gH}^2}{\tilde{\delta}_{n,\gF}}+ \eta+ \sqrt{\frac{\log(R_n/\zeta)}{n}}.\nonumber
     \end{align}
     Therefore,
     \begin{align}
         v(\widehat{C},C^{\rm ora}) \lesssim\epsilon+\tilde{\delta}_{n,\gF}+\eta+\sqrt{\frac{\log(R_n/\zeta)}{n}}+\frac{\nu\|C^{\rm ora}\|_{\gH}^2}{\tilde{\delta}_{n,\gF}}.
     \end{align}
     Otherwise, i.e. $\|f_{\widehat{C}}\|_{L_2} \leq \tilde{\delta}_{n,\gF}$,  we have
     \begin{align}
         v(\widehat{C}, C^{\rm ora}) \leq \|f_{\widehat{C}}\|_{L_2} + \eta \lesssim \tilde{\delta}_{n,\gF} + \eta.\nonumber
     \end{align}
    Combining the relations above, we can conclude that
    \begin{align}\label{eq:mse_bound}
        v(\widehat{C},C^{\rm ora}) \lesssim \epsilon+\delta_{n,\gF_{U}}+\delta_{n,\gG} + \eta+\frac{\nu\|C^{\rm ora}\|_{\gH}^2}{\tilde{\delta}_{n,\gF}}+ \sqrt{\frac{\log(R_n/\zeta)}{n}},
    \end{align}
    where we used the definition of $\tilde{\delta}_{n,\gF}$.
\end{proof}

\subsection{Proof of Theorem \ref{thm:LS2_projX}}

\begin{proof}
    Firstly, if Assumption \ref{assum:oracle_exist} (i)-(ii) holds and $H^0(Z)$ is measurable w.r.t. $\sigma$-algebra generated by $X$, then we have $h^0(X)=\E[H^0(Z)\mid X]=H^0(Z)$ a.s.. 
    Consider the structured set-valued function class $\fC=\{C(x;h):h\in\gH\}$ defined in \eqref{eq:structure_set}. Let 
    \begin{align*}
        \alpha(Z;h) := \sP\{Y\notin C(X;h)\}=1-\E\LRm{\sP\{T(h(X),Y)\leq 0\mid X,Z\}\mid Z},
    \end{align*}
    and $C^{\rm ora}(X):=\{y\in\gY: T(h^{\rm ora}(X),Y)\le 0\}$ be the optimal solution of \eqref{eq:oracle_LS_general}. Hence by Assumption \ref{assum:oracle_exist}(ii) and the definition of $h^{\rm ora}$,
    \begin{align*}
        \E\LRm{(\alpha(Z;h^{\rm ora}) - \alpha)^2}\le\E\LRm{(\alpha(Z;h^0) - \alpha)^2}=0.
    \end{align*}
    Additionally by Assumption \ref{assum:oracle_exist}(iii), for any $\vu_1,\vu_2\in\R^m$, there exist a constant $\kappa>0$ such that
    $$\LRabs{\sP\{T(\vu_1,Y)\leq 0\mid X,Z\}-\sP\{T(\vu_2,Y)\leq 0\mid X,Z\}}\le \kappa \|\vu_1-\vu_2\|,$$
    where $\|\vu\|=\sqrt{\vu^\top \vu}$. Then under Assumption 3, we have:
    
    \begin{align}
        &\E\LRm{(\alpha(Z;h^0) - \alpha)^2}\nonumber\\
        &=\E\LRm{\LRl{\E\LRm{\sP\{T(H^0(Z),Y)\leq 0\mid X,Z\}-\sP\{T(h^0(X),Y)\leq 0\mid X,Z\}\mid Z}}^2}\nonumber\\
        &\le \kappa^2\E\LRm{\LRl{\E\LRm{\|h^0(X)-H^0(Z)\|\mid Z}}^2}\nonumber\\
        &\le \kappa^2\E\LRm{\|h^0(X)-H^0(Z)\|^2}\label{eq:mse_approx_upper_bound}\\
        &= \kappa^2\mathsf{Tr}\LRs{\E\LRm{\LRs{h^0(X)-H^0(Z)}\LRs{h^0(X)-H^0(Z)}^\top}}\nonumber\\
        &= \kappa^2\mathsf{Tr}\LRs{\E\LRm{\E\LRm{\LRs{h^0(X)-H^0(Z)}\LRs{h^0(X)-H^0(Z)}^\top\mid X}}}\nonumber\\
        &= \kappa^2\mathsf{Tr}\LRs{\E\LRm{\Cov(H^0(Z)\mid X)}}\nonumber\\
        &=\kappa^2\mathsf{Tr}\LRs{\Cov(H^0(Z))}(1-\rho^2),\nonumber
    \end{align}
    where the second inequality uses Jensen's inequality, and
    $\rho^2=\frac{\mathsf{Tr}(\Cov(h^0(X)))}{\mathsf{Tr}\LRs{\Cov(H^0(Z))}}$. By the definition of $h^{\rm ora}$:
    \begin{align*}
        \E\LRm{(\alpha(Z;h^{\rm ora}) - \alpha)^2}\le\E\LRm{(\alpha(Z;h^0) - \alpha)^2}\le \kappa^2\mathsf{Tr}\LRs{\Cov(H^0(Z))}(1-\rho^2).
    \end{align*}

\end{proof}

\subsubsection{Discuss of Assumption \ref{assum:oracle_exist}}
\label{appen:discuss_h^0}
For the test-conditional coverage setting in Example~\ref{exam:test_conditional}, we take $Z=X$ and consider the pretrained sublevel prediction sets in \eqref{eq:pretrain_level_set}, then $T(h(X),Y)=s(X,Y)-h(X)$. In this case, if the distribution of $s(X,Y)\mid X$ is a Lipschitz function, then for any $h_1,h_2\in\gH$,
\begin{align}
    &\quad\LRabs{\sP\{T(h_1(X),Y)\leq 0 \mid X,Z\}- \sP\{T(h_2(X),Y) \leq 0\mid  X,Z\}}\nonumber\\
    &=\LRabs{\sP\{s(X,Y)\le h_1(X)\mid X\}-\sP\{s(X,Y)\le h_2(X)\mid X\}}\lesssim \LRabs{h_1(X)-h_2(X)},\nonumber
\end{align}
Hence, Assumption \ref{assum:oracle_exist}(iii) holds. Let 
$$
h^0(X)=\E\LRm{H^0(X)\mid X}=H^0(X)=\inf\bigl\{t:\sP\{s(X,Y)\le t\mid X\}\ge 1-\alpha\bigr\}.
$$
Therefore, $\sP\{s(X,Y)\le H^0(X)\mid X\}=1-\alpha$ by the distribution of $s(X,Y)\mid X$ is a Lipschitz continuous and then Assumption \ref{assum:oracle_exist}(i) holds.
Then $h^0(x)\in\gH$ provided that $\gH$ is sufficiently rich. 
For example, $\gH$ may be chosen as a function class capable of approximating arbitrary continuous functions, such as an RKHS with a universal kernel.

For the group-conditional coverage setting in Example~\ref{exam:group_conditional}, 
we take $Z= \bigl(\one\{X\in G_1\},\ldots,\one\{X\in G_K\}\bigr)^\top$, $\gH=\gF=\Bigl\{\sum_{Z=z} \beta_z\,\one\{Z=z\} : \beta_z\in\R,\ z\in\gZ\Bigr\}$ and consider the pretrained sublevel prediction sets in \eqref{eq:pretrain_level_set}, then $T(h(X),Y)=s(X,Y)-h(X)$.
If the distribution of $s(X,Y)\mid X$ is a Lipschitz function, then for any $h_1,h_2\in\gH$,
\begin{align}
    &\quad\LRabs{\sP\{T(h_1(X),Y)\mid X,Z\}- \sP\{T(h_2(X),Y)\mid X,Z\}}\nonumber\\
    &=\LRabs{\sP\{s(X,Y)\le h_1(X)\mid X\}-\sP\{s(X,Y)\le h_2(X)\mid X\}}\lesssim \LRabs{h_1(X)-h_2(X)},\nonumber
\end{align} 
Hence, Assumption \ref{assum:oracle_exist}(iii) holds. Moreover, by letting $q_z = \inf\bigl\{t:\sP\{s(X,Y)\le t \mid Z=z\}\ge 1-\alpha\bigr\}$ for $z\in\gZ$, we have
\[
h^0(X)= \E\LRm{H^0(Z)\mid X}= {H^0(Z)} = \sum_{z\in\gZ} q_z\,\one\{Z=z\},
\]
Therefore, $\sP\{s(X,Y)\le H^0(Z)\mid Z\}=1-\alpha$ by the distribution of $s(X,Y)\mid X$ is a Lipschitz continuous and then Assumption \ref{assum:oracle_exist}(i) holds. Moreover, $h^0\in\gH$ by the definition of $\gH$ and  Assumption \ref{assum:oracle_exist}(ii) holds.

For the Equalized coverage setting in Example~\ref{exam:equal_coverage}, 
we assume Assumption \ref{assum:oracle_exist} holds. Then $h^0(X)\in\gH$ if $\gH$ is sufficiently rich.





\subsection{Proof of Theorem \ref{thm:mse_finite} and \ref{thm:mse_rkhs}}
By Theorem \ref{thm:mse}, the convergence rate of the MSCE is primarily governed by the critical radius
$\delta_{n,\mathcal{F}_U}$ and $\delta_{n,\mathcal{G}}$ associated with the function classes
$\mathcal{F}$ and $\mathcal{G}$. 
Accordingly, the following theorem characterizes the orders of
$\delta_{n,\mathcal{F}_U}$ and $\delta_{n,\mathcal{G}}$ in the setting where $\fC$ is a VC class with VC dimension $ d_{\fC}$ and $\fC$ is chosen to be shape-constrained prediction set \eqref{eq:structure_set} and $C^{\rm ora}(X)=\{y\in\gY: T(h^{\rm ora}(X),Y)\le 0\}$.

\begin{theorem}\label{thm:critical_radius}
    Let $\fC$ be VC class with VC dimsnsion $ d_{\fC}$ and consider the shape-constrained prediction set \eqref{eq:structure_set}. Then we obtain following results:
    \begin{enumerate}
        \item[(1)] If $\gF$ is a VC class with VC dimsnsion $ d_{\gF}$, then
        \begin{align*}
            \delta_{n,\gF_U}\asymp\sqrt{\frac{ d_{\gF}}{n}\log\frac{n}{ d_{\gF}}},\quad \delta_{n,\gG} \asymp \sqrt{\frac{ d_{\gF}\vee d_{\fC}}{n}\log\LRs{\frac{n}{ d_{\gF}\vee d_{\fC}}}}.
        \end{align*}
        \item [(2)] If $\gF$ is a RKHS equipped with a Gaussian kernel and $\gZ=[0,1]^{d_{Z}}$, then
        \begin{align*}
            \delta_{n,\gF_U}\asymp\sqrt{\frac{d_{\gZ}\log^{d_{\gZ}+1}(n/d_{\gZ})}{n}},\quad \delta_{n,\gG} \asymp \sqrt{\frac{d_{\fC}\log\LRs{\frac{n}{d_{\fC}}}\vee d_{\gZ}\log^{d_{\gZ}+1}(n/d_{\gZ})}{n} }.
        \end{align*}
    \end{enumerate}
\end{theorem}

\begin{proof}[Proof of Theorem \ref{thm:critical_radius}] 
We bound $\delta_{n,\gF_U}$ and $\delta_{n,\gG}$ respectively.
    \paragraph*{Step 1: Upper bounded $\delta_{n,\gF_U}$.}
    \!\newline
    \textit{Case (a): $\gF$ is a VC class.}  Since $\gF_U = \{f \in \gF : \|f\|_{\gF}^2 \le U\}$ is a subset of $\gF$, it is straightforward to verify that for any pair of function classes $\gU \subset \gV$, the corresponding critical radius satisfy $\delta_{n,\gU} \le \delta_{n,\gV}$. Hence, it suffices to bound $\delta_{\gF}$. Let $ d_{\gF}$ be the VC-dimension of $\gF$.
    Using Theorem 2.6.4 in \citet{van2023weak}, we can bound the covering number of $\gF$ by its VC dimension $ d_{\gF}$:
    \begin{align}
        \label{eq:covering_number_upper_vc}
         N(\rho, \gF, \|\cdot\|_{L_2}) \leq K  d_{\gF} (16e)^{ d_{\gF}} \LRs{1/\rho}^{2 d_{\gF}},
    \end{align}
    where $C$ is a universal constant.
    Now use Dudley inequality, e.g., Theorem 5.22 in \citet{wainwright2019high},
    we can bound localized Rademacher complexity $\overline{\gR}_{n}(\delta;\gF)$ by
    \begin{align}
    \label{eq:local_cover_upper}
        \overline{\gR}_{n}(\delta;\gF) = \E\left[\sup_{\substack{f\in \gF\\ \|f\|_{L_2} \leq \delta}} \left|\frac{1}{n}\sum_{i=1}^n \varepsilon_i f(Z_i)\right|\right] \leq \frac{C_0}{\sqrt{n}} \int_0^\delta \sqrt{\log N(\rho, \gF, \|\cdot\|_{L_2})} d\rho,
    \end{align}
    where $C_0$ is a universal constant. Hence, plugging \eqref{eq:covering_number_upper_vc} into \eqref{eq:local_cover_upper}, we have
    \begin{align}
        \label{eq:rade_upper_F_vc}
        \overline{\gR}_{n}(\delta;\gF)&\leq \frac{C_0}{\sqrt{n}}\int_0^\delta\sqrt{\log (C d_{\gF}) + d_{\gF}\log16e +2 d_{\gF}\log(1/\rho)}d\rho\nonumber\\
        &\leq \frac{C_0}{\sqrt{n}}\LRs{\int_0^\delta\sqrt{2 d_{\gF}\log(1/\rho)}d\rho+\delta\sqrt{\log (C d_{\gF}) + d_{\gF}\log16e}}\nonumber\\
        &\leq \delta\sqrt{\frac{C_1 d_{\gF}\log(1/\delta^2)}{n}}+\delta \sqrt{\frac{C_2(\log  d_{\gF} +  d_{\gF})}{n}}\nonumber\\
        &\leq \delta\sqrt{\frac{C_3 d_{\gF}\log(1/\delta^2)}{n}},
    \end{align}
    where the last equality hold for sufficiently small $\delta$.
    Since $\delta_{\gF}$ is an arbitrary solution of $\overline{\mathcal{R}}_n(\delta ; \gF)\lesssim\delta^2$, then we just require 
    \begin{align*}
        \delta\sqrt{\frac{ d_{\gF}\log(1/\delta^2)}{n}}\lesssim\delta^2\Longleftrightarrow \frac{ d_{\gF} \log(1/\delta^2)}{n} \lesssim\delta^2.
    \end{align*}
    Considering the fixed point condition, let $t=\frac{1}{\delta^2}$ and check the root of the equation $t\log t=n/ d_{\gF}$. Taking the logarithm on both sides, we have:
    \begin{align*}
        \log t+\log\log t = \log (n/ d_{\gF}).
    \end{align*}
    By Lemma \ref{lemma:critical_equation} with $L_n = \log (n/ d_{\gF})$, we have
    \begin{align*}
        \frac{1}{\delta^2_\gF}\asymp\exp(L_n -\log L_n )=\frac{n/ d_{\gF}}{\log(n/ d_{\gF})}.
    \end{align*}
    Hence, it follows that $\delta_{n,\gF_U}\lesssim\delta_{\gF}\asymp\sqrt{\frac{ d_{\gF}}{n}\log\frac{n}{ d_{\gF}}}$.\newline
    \textit{Case (b): $\gF$ is a RKHS equipped with a Gaussian kernel.} Since  $\gF$ is a RKHS equipped with a Gaussian kernel and $\gF_U=\{f\in\gF:\|f\|_{\gF}^2\le U\}$.  When $\gF_U$ is equipped with a Gaussian kernel with $\|f\|_\gF \le U$, $Z\in[0,1]^{d_{\gZ}}$, by Proposition 1 of \cite{zhou2002covering},
    \begin{align}
        \label{eq:covering_number_rkhs}
        \log N(\rho, \gF_U, \|\cdot\|_{L_2}) \leq 4^{d_{\gZ}}(6d_Z+2)\LRs{\log(U/\rho)}^{d_{\gZ}+1}.
    \end{align}
    By substituting \eqref{eq:covering_number_rkhs} into \eqref{eq:local_cover_upper},
    \begin{align}\label{eq:rade_upper_F_rkhs}
        \overline{\gR}_{n}(\delta;\gF_U)&\leq \frac{C_0}{\sqrt{n}}\int_0^\delta\sqrt{4^{d_{\gZ}}(6d_Z+2)\LRs{\log(U/\rho)}^{d_{\gZ}+1}}d\rho\nonumber\\
        &\leq C_0\delta\sqrt{\frac{(2d_Z+1)(4\log(U^2/\delta^2))^{d_{\gZ}+1}}{n}}.
    \end{align}
    Since $\delta_{\gF_U}$ is an arbitrary solution of $\overline{\mathcal{R}}(\delta ; \gF_U)\lesssim\delta^2$, then we just require 
    \begin{align*}
        \delta\sqrt{\frac{d_{\gZ}}{n}\log^{d_{\gZ}+1}(U^2/\delta^2)}\lesssim\delta^2\Longleftrightarrow \frac{d_z}{n}\log^{{d_{\gZ}+1}}(U^2/\delta^2)\lesssim\delta^2.
    \end{align*}
    Considering the fixed point condition, let $t=U^2/\delta^2$ and check the root of the equation $t(\log t)^{d_{\gZ}+1}=nB^2/d_{\gZ}$. Taking the logarithm on both sides, we have:
    \begin{align*}
        \log t+{(d_{\gZ}+1)}\log\log t =\log (nU^2/d_{\gZ}).
    \end{align*}
    By Lemma \ref{lemma:critical_equation} with $L_n = \log (nU^2/d_{\gZ})$, we have
    \begin{align}
    \label{eq:delta_gF_RKHS}
        \frac{U^2}{\delta^2_\gF}\asymp\exp\LRl{\log (nU^2/d_{\gZ})-{(d_{\gZ}+1)}\log \log (nU^2/d_{\gZ})}=\frac{nU^2/d_{\gZ}}{\log^{d_{\gZ}+1}(nU^2/d_{\gZ})}.
    \end{align}
    Hence, $\delta_{n,\gF_U}\asymp\sqrt{\frac{d_{\gZ}\log^{d_{\gZ}+1}(n/d_{\gZ})}{n}}$.

    \paragraph*{Step 2: Upper bounded $\delta_{n,\gG}$.} Since we consider structure set \eqref{eq:structure_set},
     by the definition of $\gG$ in \eqref{eq:def_G}, $g_h\in\gG$ can be written as,
     \begin{align*}
         g_h(X,Z,Y)=f_h(Z)\LRs{\mathbbm{1}\{T(h(X), Y) \leq 0\}-\mathbbm{1}\{T(h^{\rm ora}(X), Y) \leq 0\}},
     \end{align*}
    where $f_h=\argmin_{f\in {\gF_U}}\E[(f-(\alpha(Z;h)-\alpha(Z;h^{\rm ora})))^2]$. We define the function class 
    \begin{align*}
        {\gG}^*:=\left\{f(Z)\LRs{\mathbbm{1}\{T(h(X), Y) \leq 0\}-\mathbbm{1}\{T(h^{\rm ora}(X), Y) \leq 0\}}:f\in\gF_U, h\in\gH\right\}.
    \end{align*}
    Note that $\gG\subseteq\gG^*$, then we have $\overline{\gR}_{n}(\delta;\gG)\le \overline{\gR}_{n}(\delta;\gG^*)$. 
    Let $\gamma_h(x,y)=\mathbbm{1}\{T(h(x), y) \leq 0\}-\mathbbm{1}\{T(h^{\rm ora}(x), y) \leq 0\}$ and $\Gamma_\gH=\{\gamma_h:h\in\gH\}$. 
    By Assumption \ref{assum:VC_class_C}, the VC-dimension of $\Gamma_\gH$ is also $ O(d_{\fC})$.
    By Lemma \ref{lemma:product_covering_number} and $\|f\|_\infty\le 1$ and $\|\gamma_h\|_\infty\leq 1$, we have 
    \begin{align}
        \log N(\rho, \gG^*, \|\cdot\|_{L_2}) \le \log N(\rho/2, \gF_U, \|\cdot\|_{L_2}) + \log N(\rho/2, \Gamma_\gH, \|\cdot\|_{L_2}).\nonumber
    \end{align}
    Together with \eqref{eq:covering_number_upper_vc} and \eqref{eq:local_cover_upper}, we have
    \begin{align}\label{eq:rade_upper_G_general}
        \overline{\gR}_{n}(\delta;\gG)&\le \overline{\gR}_{n}(\delta;\gG^*)\le \frac{C_0}{\sqrt{n}} \int_0^\delta \sqrt{\log N(\rho, \gG^*, \|\cdot\|_{L_2})} d\rho\nonumber\\
        &\leq 2\int_0^{\delta/2} \sqrt{\log N(\rho, \gF_U, \|\cdot\|_{L_2})} d\rho + 2\int_0^{\delta/2} \sqrt{\log N(\rho, \Gamma_\gH, \|\cdot\|_{L_2})} d\rho \nonumber\\
        &\leq 2\int_0^{\delta/2}\sqrt{\log N(\rho, \gF_U, \|\cdot\|_{L_2})} d\rho + \delta\sqrt{\frac{C_5 d_{\fC}\log(4/\delta^2)}{n}}.
    \end{align}
    \textit{Case (a): $\gF$ is a VC class.} Combining \eqref{eq:rade_upper_F_vc} and \eqref{eq:rade_upper_G_general}, we have
    \begin{align*}
        \overline{\gR}_{n}(\delta;\gG)\le \delta \sqrt{\frac{C_6\LRs{d_{\fC}\vee  d_{\gF}}\log(4/\delta^2)}{n}}
    \end{align*}
    Considering the fixed point condition, and let $t=\frac{2}{\delta}$, $L=\log\frac{n}{ d_{\gF}\vee d_{\fC}}$. By Lemma \ref{lemma:critical_equation},
    \begin{align*}
        &\frac{1}{\delta^2_\gG}\asymp\exp(L-\log L)=\frac{n/\LRs{ d_{\gF}\vee d_{\fC}}}{\log(n/( d_{\gF}\vee d_{\fC}))}\\ \Longrightarrow\quad  &\delta_{n,\gG}\lesssim\sqrt{\frac{ d_{\gF}\vee d_{\fC}}{n}\log\LRs{\frac{n}{ d_{\gF}\vee d_{\fC}}}}.
    \end{align*}
    \textit{Case (b): $\gF$ is a RKHS equipped with a Gaussian kernel.}
    Combining \eqref{eq:rade_upper_F_rkhs} and \eqref{eq:rade_upper_G_general}, we have
    \begin{align*}
        \overline{\gR}_{n}(\delta;\gG)\le \delta\sqrt{\frac{C_4\log^{d_{\gZ}+1}(4U^2/\delta^2)}{n}}+\delta \sqrt{\frac{C_5{d_{\fC}}\log(4/\delta^2)}{n}}
    \end{align*}
    Since we require $\overline{\mathcal{R}}(\delta ; \gG)\lesssim\delta^2$, then using Lemma \ref{lemma:critical_equation}, we can have
    \begin{align*}
        \begin{cases}
            \frac{d_{\gZ}}{n}\log^{{d_{\gZ}+1}}(U^2/\delta^2)\lesssim \frac{1}{2}\delta^2,\\[4pt]
            \frac{d_{\fC}}{n}{\log\frac{1}{\delta^2}}\lesssim \frac{1}{2}\delta^2.
        \end{cases} 
        \Longrightarrow \delta_{n,\gG}\asymp \sqrt{\frac{d_{\fC}\log\LRs{\frac{n}{d_{\fC}}}\vee d_{\gZ}\log^{d_{\gZ}+1}(n/d_{\gZ})}{n} }.
    \end{align*}
\end{proof}

\begin{proof}[Proof of Theorem \ref{thm:mse_finite}]
    Since we choose $\gF=\{\sum_{z\in\gZ} \beta_z \one\{Z=z\}:\beta_z\in\R\}$, for any $C\in\fC$, the function $\alpha(z;C)$ corresponds to $Z = z$ takes at most $|\gZ|$ distinct values which indicates that $\alpha(z;C)\in\gF$. Therefore, by definitions of $\gF$ and $f_C = \argmin_{f \in \gF_U}\E[(f(Z) - (\alpha(Z; C^{\rm ora}) - \alpha(Z; C)))^2]$, we have $\E[(f_C(Z) - (\alpha(Z; C^{\rm ora}) - \alpha(Z; C)))^2]=0$. Hence, $\eta=0$. 
   By Theorem \ref{thm:mse} and \ref{thm:critical_radius},  we choose $\gamma= 0$ and then there exist $C>0$ such that w.p. $1-4\zeta$:
     \begin{align*}
        &\E\LRm{(\alpha(Z; \widehat{C}) - \alpha)^2\mid \gD_n}\\
        &\qquad\le C\LRs{\epsilon^2 +\frac{ d_{\gF}\vee d_{\fC}}{n}\log\LRs{\frac{n}{ d_{\gF}\vee d_{\fC}}} + {\frac{\log((\log n) /\zeta)}{n}}} :=A_n.
    \end{align*}
    Therefore, Let $\gE=\LRl{\E\LRm{(\alpha(Z; \widehat{C}) - \alpha)^2\mid \gD_n}\le A_n}$
    \begin{align*}
        \E\LRm{(\alpha(Z; \widehat{C}) - \alpha)^2} &= \E\LRm{\E\LRm{(\alpha(Z; \widehat{C}) - \alpha)^2\mid \gD_n}}\\
        &=\E\LRm{\E\LRm{(\alpha(Z; \widehat{C}) - \alpha)^2\mid \gD_n}\one_{\gE}}\\
        &\qquad+\E\LRm{\E\LRm{(\alpha(Z; \widehat{C}) - \alpha)^2\mid \gD_n}\one_{\gE^c}}\\
        &\le A_n+4\zeta.
    \end{align*}
    By choosing $\zeta=\frac{1}{n}$, we have 
    \begin{align*}
        \E\LRm{(\alpha(Z; \widehat{C}) - \alpha)^2}\lesssim{\epsilon^2+\frac{ d_{\gF}\vee d_{\fC}}{n}\log\LRs{\frac{n}{ d_{\gF}\vee d_{\fC}}}}.
    \end{align*}
\end{proof}

\begin{proof}[Proof of Theorem \ref{thm:mse_rkhs}]
    By Theorem \ref{thm:mse} and \ref{thm:critical_radius},  we choose $\gamma\ge 0$ s.t. $\gamma\|f\|_\gF^2\le \|f\|_2^2$ and combining all results above, there exist $C>0$ such that w.p. $1-4\zeta$:
     \begin{align*}
        &\E\LRm{(\alpha(Z; \widehat{C}) - \alpha)^2\mid \gD_n}\\
        &\qquad\le C\LRs{\epsilon^2+\eta^2 +\frac{d_{\fC}\log\LRs{\frac{n}{d_{\fC}}}\vee d_{\gZ}\log^{d_{\gZ}+1}\LRs{\frac{n}{d_{\gZ}}}}{n} + {\frac{\log((\log n) /\zeta)}{n}}} :=A_n.
    \end{align*}
    Therefore, Let $\gE=\LRl{\E\LRm{(\alpha(Z; \widehat{C}) - \alpha)^2\mid \gD_n}\le A_n}$
    \begin{align*}
        \E\LRm{(\alpha(Z; \widehat{C}) - \alpha)^2} &= \E\LRm{\E\LRm{(\alpha(Z; \widehat{C}) - \alpha)^2\mid \gD_n}}\\
        &=\E\LRm{\E\LRm{(\alpha(Z; \widehat{C}) - \alpha)^2\mid \gD_n}\one_{\gE}}\\ &+\E\LRm{\E\LRm{(\alpha(Z; \widehat{C}) - \alpha)^2\mid \gD_n}\one_{\gE^c}}\\
        &\le A_n+4\zeta.
    \end{align*}
    By choosing $\zeta=\frac{1}{n}$, we have 
    \begin{align*}
        \E\LRm{(\alpha(Z; \widehat{C}) - \alpha)^2}\lesssim{\epsilon^2+\eta^2 +\frac{d_{\fC}\log\LRs{\frac{n}{d_{\fC}}}\vee d_{\gZ}\log^{d_{\gZ}+1}\LRs{\frac{n}{d_{\gZ}}}}{n}}.
    \end{align*}
\end{proof}

\subsection{Proof of Corollary \ref{cor:mse_finite_group_X}}
\begin{proof}
    Since $Z=(\one\{X\in G_1\},\ldots,\one\{X\in G_K\})^\top$, then $h^0(X)= \E[H^0(Z)\mid X]=H^0(Z)$ is a function of $Z$ and takes at most $|\gZ|$ distinct values. Hence by choosing $\gH=\gF=\Bigl\{\sum_{Z=z} \beta_z\,\one\{Z=z\} : \beta_z\in\R,\ z\in\gZ\Bigr\}$, we have $h^0\in\gH$.
    Therefore Assumption \ref{assum:oracle_exist}(i,ii) holds and $H^0(Z)$ is measurable w.r.t. $\sigma(X)$. By Theorem \ref{thm:LS2_projX}, we have $\epsilon=0$.
    Let $\gZ_k=\{z\in\gZ:z_k=1\}$ and $\gZ_{k}\subseteq\gZ$, then $\one\{X\in G_k\}=\one\{Z\in\gZ_k\}$.
    Hence for any $k\in[K]$,
    \begin{align*}
        \sP\{Y \notin \widehat{C}(X;h) \mid X \in G_k\} &= \sP\{Y \notin \widehat{C}(X;h) \mid Z \in \gZ_k\}\\
        &= \sum_{z\in\gZ_k}\sP\{Y \notin \widehat{C}(X;h) \mid Z=z\}\sP\{Z=z\mid Z \in \gZ_k\}\\
        &= \sum_{z\in\gZ_k}\sP\{Y \notin \widehat{C}(X;h) \mid Z=z\}\frac{\sP\{Z=z\}}{\sP\{Z \in \gZ_k\}}
    \end{align*}
    Therefore,
    \begin{align*}
        (\sP\{Y \notin \widehat{C} \mid X \in G_k\}-\alpha)^2\sP\{X\in G_k\}
        &= (\sP\{Y \notin \widehat{C} \mid Z \in \gZ_k\}-\alpha)^2\sP\{Z\in \gZ_k\}\\
        &= \LRs{\sum_{z\in\gZ_k}\frac{\sP\{Z=z\}}{\sP\{Z \in \gZ_k\}}(\alpha(z,\widehat{C})-\alpha)}^2\sP\{Z\in\gZ_k\}\\
        &\overset{(i)}{\le}\sum_{z\in\gZ_k}(\alpha(z,\widehat{C})-\alpha)^2\sP\{Z=z\}\\
        &{\le} \sum_{z\in\gZ}(\alpha(z,\widehat{C})-\alpha)^2\sP\{Z=z\}=
        \E\LRm{(\alpha(Z; \widehat{C}) - \alpha)^2}.
    \end{align*}
    where (i) we use Jensen's inequality. Therefore,
    \begin{align*}
        \LRabs{\sP\{Y \notin \widehat{C} \mid X \in G_k\}-\alpha}\le\sqrt{\frac{\E{[(\alpha(Z; \widehat{C}) - \alpha)^2]}}{\sP\{X\in G_k\}}}
    \end{align*}
    Hence by Theorem \ref{thm:mse_finite}, we have 
    \begin{align*}
        \left|\sP\LRl{Y\in\widehat{C}(X)\mid X\in G_k}-(1-\alpha)\right|\lesssim\sqrt{\frac{d_\fC+|\gZ|}{n\,\sP(X\in G_k)}\log \frac{n}{d_\fC+|\gZ|}}.
    \end{align*}
    Moreover, for pretrained sublevel set \eqref{eq:pretrain_level_set}, $T(h(X,Y))=s(X,Y)-h(X)$ and then $\{Y\notin C(X;h)\}=\{s(X,Y)>h(X)\}$ is a subgraph for all $h\in\gH$. Since $\gH=\Bigl\{\sum_{Z=z} \beta_z\,\one\{Z=z\} : \beta_z\in\R,\ z\in\gZ\Bigr\}$ and it has VC dimension $d_{\gH}\le |\gZ|$, then $\{(x,y):s(x,y)>h(x), h\in\gH\}$ is a VC-subgraph with VC dimension $O(|\gZ|)$.
\end{proof}

\subsection{Proof of Corollary \ref{cor:mse_finite_equal}}
\begin{proof}
    Since Assumption \ref{assum:oracle_exist} holds and $|\gZ|< \infty$, then let $\pi_z(X)=\sP\{Z=z\mid X\}$ and $h^0(X)=\E[H^0(Z)\mid X]=\sum_{z\in\gZ}\pi_z(X)H^0(z)$
    by Eq. \eqref{eq:mse_approx_upper_bound}, we have 
    \begin{align*}
        \mathsf{MSCE}(C^{\rm ora})&\le \kappa^2\E[\|H^0(Z)-h^0(X)\|^2]\\
        &= \kappa^2\E\!\left[\sum_{z\in\gZ}\pi_z(X)\|H^0(z)-h^0(X)\|^2\right]\\
        &= \frac{\kappa^2}{2}\E\!\left[\sum_{z,z'\in\gZ}\pi_z(X)\pi_{z'}(X)\|H^0(z)-H^0(z')\|^2\right]\\
        &\le \frac{\kappa^2\Delta_{H^0}^2}{2}\E\!\left[\left\{1-\sum_{z\in\gZ}\pi_z(X)^2\right\}\right],
    \end{align*}
    where $\Delta_{H^0}:=\max_{z,z'\in\gZ}\|H^0(z)-H^0(z')\|$.
    For each $z \in \gZ$, note that 
    \begin{align*}
        \E\LRm{(\alpha(Z; \widehat{C}) - \alpha)^2}=\sum_{z\in\gZ}(\alpha(z,\widehat{C})-\alpha)^2\sP\{Z=z\}.
    \end{align*}
    Hence,
    \begin{align*}
        \LRabs{\sP\{Y \notin \widehat{C} \mid Z=z\}-\alpha}\le\sqrt{\frac{\E{[(\alpha(Z; \widehat{C}) - \alpha)^2]}}{\sP\{Z=z\}}}.
    \end{align*}
    Hence by Theorem \ref{thm:mse_finite}, we have 
    \begin{align*}
        \left|\sP\LRl{Y\in\widehat{C}^{\rm{opt}}(X)\mid Z=z}-(1-\alpha)\right| &\lesssim \sqrt{\frac{\kappa^2\Delta_{H^0}^2\E\!\left[\left\{1-\sum_{z\in\gZ}\pi_z(X)^2\right\}\right]}{2\cdot\sP\{Z=z\}}}\\
        &\quad+\sqrt{\frac{d_{\fC}+|\gZ| }{n\cdot \sP(Z=z)}\log\LRs{\frac{n}{d_{\fC}+|\gZ|}}}.
    \end{align*}
\end{proof}

\subsection{Corollary for test-conditional coverage}
\begin{corollary}\label{cor:mse_rkhs_cond_X}
    Consider Example \ref{exam:test_conditional} with $Z=X$. Suppose Assumptions \ref{assum:appro_error}-\ref{assum:VC_class_C} hold. If $\alpha(Z;C) - \alpha \in \gF_U$ with $U=1$ for any $C\in \fC$, then
    \begin{align}
        \mathsf{MSCE}(\widehat{C}) \lesssim \frac{d_{\fC}\log(n/d_{\fC})}{n}+ \frac{d_{\gZ}\log^{d_{\gZ}+1}(n/d_{\gZ})}{n}.\nonumber
    \end{align}
\end{corollary}
\begin{proof}
    Since Assumption \ref{assum:oracle_exist} holds and $Z=X$, then by Theorem \ref{thm:LS2_projX}, we have $\epsilon=0$.
    Moreover, since we assume $\alpha(Z;C)-\alpha(Z;C^{\rm ora})\in\gF_U$, by the definition of $f_C$,
    $$f_C=\argmin_{f \in\gF} \E\LRm{(f(Z)-(\alpha(Z;C)-\alpha(Z;C^{\rm ora})))^2},$$
    then we have $\eta=0$. 
    Hence, by Theorem \ref{thm:mse_rkhs}, we have 
    \begin{align}
        \E\LRm{\LRs{\alpha(X_{n+1}; \widehat{C}) - \alpha}^2} \lesssim \frac{d_{\fC}\log(n/d_{\fC})}{n}+ \frac{d_{\gZ}\log^{d_{\gZ}+1}(n/d_{\gZ})}{n}.\nonumber
    \end{align}
\end{proof}

\subsection{Multi-group coverage guarantee}
In addition to the MSE results, we also have the multi-group coverage guarantee for MOPI.

\begin{theorem}[Multi-group coverage]\label{thm:multi_group}
    Under the same settings of Theorem \ref{thm:mse} and assume $f$ is bounded by some constant $B>0$ for all $f\in\gF$, with probability at least $1-4\zeta$:
    $$\E\LRm{f(Z) \LRs{\mathbbm{1}\{Y\notin\widehat{C}(X)\} - \alpha} \mid \gD_n}\lesssim \LRs{\E\LRm{(\alpha(Z; \widehat{C}) - \alpha)^2 \mid \gD_n}}^{1/2}.$$   
\end{theorem}

\begin{proof}
    Since $\|f\|_{\infty}\le B$, then $\sup_{f\in\gF}\|f\|_2 \le B$.
    Note that $\E\LRm{f(Z)\LRs{\one\{Y\notin C(X)\}-\alpha}}=\E[f(Z)(\alpha(Z;C)-\alpha)]$, by Cauchy-Schwarz inequality, 
    \begin{align*}
        \sup_{f\in\gF}\E\LRm{f(Z)\LRs{\one\{Y\notin C(X)\}-\alpha}}\le \sup_{f\in\gF}\|f\|_2 \LRs{\E[(\alpha(Z;C)-\alpha)^2]}^{1/2}.
    \end{align*}
    for all $C\in\mathfrak{C}$.
\end{proof}

\subsection{Proof of Theorem \ref{thm:mse_smooth_finite}}\label{appen:proof_mse_smooth}
\begin{lemma}\label{lemma:smooth_indicator_approx}
    Let $\tilde{\mathbbm{1}}\{a<b\}$ be the smoothed indicator function defined by
    $$
    \tilde{\mathbbm{1}}\{a<b\}=\frac{1}{2}\left(1+\operatorname{erf}\left(\frac{b-a}{\sqrt{2} r}\right)\right)
    $$
    where $\operatorname{erf}(x)$ is the error function, defined as $\operatorname{erf}(x)=\frac{2}{\sqrt{\pi}} \int_0^x e^{-t^2} d t$, and $r>0$ is the smoothing parameter. 
    The error between the smoothed indicator function $\tilde{\mathbbm{1}}\{a<b\}$ and the actual indicator function $\mathbbm{1}\{a<b\}$, integrated over all $a$, is given by
    $$
    E=\int_{-\infty}^{\infty}|\tilde{\mathbbm{1}}\{a<b\}-\mathbbm{1}\{a<b\}| d a = \sqrt{\frac{\pi}{2}} r.
    $$
\end{lemma}
\begin{proof}
    Let $x=a-b$. Note
\begin{align*}
    \tilde{\mathbbm{1}}\{a<b\}
    =\tfrac12\Big(1+\operatorname{erf}\!\big(\tfrac{b-a}{\sqrt{2}r}\big)\Big)
    =\Phi\!\big(\tfrac{b-a}{r}\big)=\Phi\!\big(-\tfrac{x}{r}\big),
\end{align*}
where $\Phi$ is the standard normal CDF and $\varphi$ its density.
\begin{align*}
    E=\int_{-\infty}^{\infty}\big|\Phi(-x/r)-\mathbbm{1}\{x<0\}\big|\,dx
=\int_{-\infty}^{0}\big(1-\Phi(-x/r)\big)\,dx+\int_{0}^{\infty}\Phi(-x/r)\,dx.
\end{align*}
By symmetry,
\begin{align*}
    E=2\int_{0}^{\infty}\Phi(-x/r)\,dx=2r\int_{0}^{\infty}\Phi(-u)\,du
=2r\int_{0}^{\infty}\big(1-\Phi(u)\big)\,du.
\end{align*}
Since $1-\Phi(u)=\int_u^{\infty}\varphi(t)\,dt$, Fubini gives
\begin{align*}
    \int_{0}^{\infty}\big(1-\Phi(u)\big)\,du=\int_{0}^{\infty}\int_{u}^{\infty}\varphi(t)\,dt\,du
=\int_{0}^{\infty}t\varphi(t)\,dt=\frac{1}{\sqrt{2\pi}}.
\end{align*}
Therefore $E=2r\cdot\frac{1}{\sqrt{2\pi}}=r\sqrt{\frac{2}{\pi}}.$
\end{proof}

\begin{lemma}
    Let $f(x,b)=\tilde{\mathbbm{1}}\{x<b\}=\frac{1}{2}\left(1+\operatorname{erf}\left(\frac{b-x}{\sqrt{2} r}\right)\right)$, where $\operatorname{erf}(x)=\frac{2}{\sqrt{\pi}} \int_0^x e^{-t^2} d t$ and $r>0$ is the smoothing parameter. Then $f (x, b)$ is Lipschitz continuous with respect to $x$ with a Lipschitz constant $\frac{1}{\sqrt{2\pi}r}$.
\end{lemma}
\begin{proof}
    Since 
    $$\left|\frac{\partial}{\partial x}f(x,b)\right|=\left|\frac{1}{2}\frac{\partial}{\partial x}\left(1+\frac{2}{\sqrt{\pi}} \int_0^{(b-x)/\sqrt{2}r} e^{-t^2} d t\right)\right|=\frac{1}{\sqrt{2\pi}r}e^{\frac{-(b-x)^2}{2r^2}}\le \frac{1}{\sqrt{2\pi}r},$$
    $f(x,b)$ is Lipschitz continuous with respect to $x$ with Lipschitz constant $\frac{1}{\sqrt{2\pi}r}$.
\end{proof}

\begin{proof}[Proof of Theorem \ref{thm:mse_smooth_finite}]
    For $f\in \gF, h\in \gH$, we denote the function
    \begin{align}
        \tilde{\psi}(h, f)(X,Z,Y) &= f(Z) \LRs{\tilde{\mathbbm{1}}\{T(h(X),Y)>0\} - \alpha},\nonumber\\
       {\psi}(h, f)(X,Z,Y) &= f(Z) \LRs{{\mathbbm{1}}\{T(h(X),Y)>0\} - \alpha},\nonumber
    \end{align}
    and assume $\sup_{f\in\gF} \|f\|_{\infty}\le \overline{C}_{\gF}$ and the density of $T(h(X),Y)\mid Z$, $p_{T\mid Z}$ is upper bounded by $\overline{C}_T$. Consider
    \begin{align}\label{eq:smooth_gap}
        \sup_{f\in\gF} &\LRl{P \tilde{\psi}(h, f)-P \psi(h, f)}
        =\sup_{f\in\gF} \E\LRm{f(Z)\LRs{\tilde{\mathbbm{1}}\{T(h(X),Y)>0\}-\mathbbm{1}\{T(h(X),Y)>0\}}}\nonumber\\
        &\le \sup_{f\in\gF} \|f\|_{\infty}\E{\LRm{\LRabs{\tilde{\mathbbm{1}}\{T(h(X),Y)>0\}-{\mathbbm{1}}\{T(h(X),Y)>0\}}}}\nonumber\\
        &\le \overline{C}_{\gF} \E_{Z}\LRm{\E\LRm{\LRabs{\mathbbm{1}\{T(h(X),Y)>0\}-\tilde{\mathbbm{1}}\{T(h(X),Y)>0\}}\mid Z}}\nonumber\\
        &\le \sqrt{\frac{2}{\pi}} \overline{C}_{\gF} \overline{C}_T r,
    \end{align}
    where we use Lemma \ref{lemma:smooth_indicator_approx}. 
    Recall that $\tilde{\delta}_{n,\gF} = \delta_{n,\gF_{U}} + \sqrt{\frac{\log(c_1 / \zeta)}{c_2 n}}$ for some positive constants $c_1$ and $c_2$.
    Since $\tilde{\psi}(h, f)$ is 1-Lipschitz w.r.t. $f$, using Lemma 14 in \citet{foster2023orthogonal}, we also have
    \begin{align}\label{eq:multiaccu_empi_process_bound_smooth}
         \sP\LRs{\forall f\in \gF,  \left|(P_n-P) \tilde{\psi}(h^{\rm ora}, f)\right| \leq 18 \LRs{\tilde{\delta}_{n,\gF} \|f\|_{L_2} + \LRs{1\vee \frac{\|f\|_{\gF}}{\sqrt{U}}}\tilde{\delta}_{n,\gF}^2}} \geq 1 - \zeta.
    \end{align}
    
    Define $\widetilde{\Psi}_{\gamma}(h,f) = P_n \tilde{\psi}(h,f) - P_n f^2-\gamma\|f\|_\gF^2$ for any $h\in \gH, f\in \gF$ and recall that ${\Psi}_{\lambda}(h,f) = P {\psi}(h,f) - \lambda P f^2$.
    Since $\tilde{h}=\argmin_{h\in\gH}\max_{f\in_\gF} \widetilde{\Psi}_{\gamma}(h, f)$, by the optimality of $\hat{h}$,
    \begin{align}\label{eq:upper_hhat_horacle_smooth}
        \sup_{f\in\gF}\widetilde{\Psi}_{\gamma}(\tilde{h},f) &\leq \sup_{f\in \gF}\widetilde{\Psi}_{\gamma}(h^{\rm ora},f). 
    \end{align}
    Combining \eqref{eq:variance_empi_process_lower} and \eqref{eq:multiaccu_empi_process_bound_smooth}, we can get w.p. $1-2\zeta$:
    \begin{align}\label{eq:empirical_sup_upper_smooth}
        \sup_{f\in \gF} &\widetilde{\Psi}_{\gamma}(h^{\rm ora},f) =\sup_{f\in \gF}\LRl{P_n \tilde{\psi}(h^{\rm ora},f) -  P_n f^2-\gamma\|f\|_\gF^2}\nonumber\\
        &\leq \sup_{f\in \gF}\LRl{P_n \tilde{\psi}(h^{\rm ora},f) - \LRs{\gamma\|f\|_{\gF}^2+\frac{1}{4}\|f\|_{L_2}^2-\frac{\tilde{\delta}_{n,\gF}^2}{2}}}\nonumber\\
        &\leq \sup_{f\in \gF}\LRl{P \tilde{\psi}(h^{\rm ora}, f) + 18 \LRs{\tilde{\delta}_{n,\gF} \|f\|_{L_2} + \LRs{1\vee \frac{\|f\|_{\gF}}{\sqrt{U}}}\tilde{\delta}_{n,\gF}^2} - \LRs{\gamma\|f\|_{\gF}^2+\frac{1}{4}\|f\|_{L_2}^2-\frac{\tilde{\delta}_{n,\gF}^2}{2}}} \nonumber\\
        &\overset{\rm (i)}{\leq} \sup_{f\in \gF}\LRl{P \psi(h^{\rm ora}, f) -\frac{1}{8}P f^2} + \sup_{f\in \gF}\LRl{36\tilde{\delta}_{n,\gF} \|f\|_{L_2} - \frac{1}{8}\|f\|_{L_2}^2} +18\tilde{\delta}_{n,\gF}^2+\frac{\tilde{\delta}_{n,\gF}^2}{2} \nonumber\\
        & + \sup_{f\in\gF} \LRl{P \tilde{\psi}(h^{\rm ora}, f)-P \psi(h^{\rm ora}, f)}\nonumber\\
        &{\leq} \sup_{f\in \gF}\Psi_{\frac{1}{8}}(h^{\rm ora}, f) + {2\cdot 36^2\tilde{\delta}_{n,\gF}^2} + 18\tilde{\delta}_{n,\gF}^2+\frac{\tilde{\delta}_{n,\gF}^2}{2} + \sqrt{\frac{2}{\pi}} \overline{C}_{\gF} \overline{C}_T r,
    \end{align}
    where (i) we used the assumption $\tilde{\delta}_{n,\gF} \|f\|_{\gF} \leq \sqrt{\frac{U}{2}}\|f\|_{L_2}$ and $\gamma\geq 0$.
    In addition, we also have the lower bound
    \begin{align}
        \sup_{f\in \gF} \widetilde{\Psi}_{\gamma}(\tilde{h},f) &= \sup_{f\in \gF} \LRl{P_n \tilde{\psi}(\tilde{h}, f) - P_n \tilde{\psi}(h^{\rm ora}, f) + P_n \tilde{\psi}(h^{\rm ora}, f) -  P_n f^2 - \gamma\|f\|_{\gF}^2}\nonumber\\
        &\geq \sup_{f\in \gF}\LRl{P_n \tilde{\psi}(\tilde{h}, f) - P_n \tilde{\psi}(h^{\rm ora}, f) - 2 \LRs{ P_n f^2 + \gamma\|f\|_{\gF}^2}}\nonumber\\
        &\qquad +\inf_{f\in \gF} \LRl{P_n \tilde{\psi}(h^{\rm ora}, f) +  P_n f^2 + \gamma\|f\|_{\gF}^2}\nonumber\\
        &= \sup_{f\in \gF}\LRl{P_n \tilde{\psi}(\tilde{h}, f) - P_n \tilde{\psi}(h^{\rm ora}, f) - 2\LRs{ P_n f^2 + \gamma\|f\|_{\gF}^2}} - \sup_{f\in \gF}\widetilde{\Psi}_{\gamma}(h^{\rm ora}, f),\nonumber
    \end{align}
    where we used the fact that $\gF$ is symmetric. It follows that w.p. $1-2\zeta$:
    
    \begin{align}\label{eq:penalized_EP_upper_smooth}
        &\sup_{f\in \gF}\LRl{P_n \tilde{\psi}(\tilde{h}, f) - P_n \tilde{\psi}(h^{\rm ora}, f) - 2\LRs{ P_n f^2+ \gamma\|f\|_\gF^2}}\nonumber\\
        &\qquad \leq \LRl{\sup_{f\in \gF} \widetilde{\Psi}_{\gamma}(\tilde{h},f) + \sup_{f\in \gF} \widetilde{\Psi}_{\gamma}(h^{\rm ora},f)}\nonumber\\
        &\qquad\overset{\rm (i)}{\leq} 2\sup_{f\in \gF} \widetilde{\Psi}_{\gamma}(h^{\rm ora},f) \nonumber\\
        &\qquad\overset{\rm (ii)}{\leq} 2\sup_{f\in \gF}\Psi_{\frac{1}{8}}(h^{\rm ora}, f) + {4\cdot 36^2\tilde{\delta}_{n,\gF}^2} + 36\tilde{\delta}_{n,\gF}^2+{\tilde{\delta}_{n,\gF}^2} + 2\sqrt{\frac{2}{\pi}} \overline{C}_{\gF} \overline{C}_T r \nonumber\\
        &\qquad\overset{\rm (iii)}{\leq} {4\epsilon^2} + {4\cdot 36^2\tilde{\delta}_{n,\gF}^2} + 36\tilde{\delta}_{n,\gF}^2+{\tilde{\delta}_{n,\gF}^2} + 2\sqrt{\frac{2}{\pi}} \overline{C}_{\gF} \overline{C}_T r,
    \end{align}
    where (i) holds due to \eqref{eq:upper_hhat_horacle_smooth}, (ii) holds due to \eqref{eq:empirical_sup_upper_smooth} and (iii) holds due to Lemma \ref{lemma:CP_LS2minimax_general} and the assumption $\E[(\alpha(Z;h^{\rm ora}) - \alpha)^2] \leq \epsilon^2$.

    Define $\tilde{\alpha}(Z;h)=\E[\tilde{\mathbbm{1}}\{T(h(X),Y)>0\}|Z]$ and note that
    \begin{align*}
        &\|\LRs{{\alpha}(Z;h)-\tilde{\alpha}(Z;h)} - \LRs{{\alpha}(Z;h^{\rm ora})-\tilde{\alpha}(Z;h^{\rm ora})}\|_{L_2}\nonumber\\
        &\quad\le \|{\alpha}(Z;h)-\tilde{\alpha}(Z;h)\|_{L_2} + \|{\alpha}(Z;h^{\rm ora})-\tilde{\alpha}(Z;h^{\rm ora})\|_{L_2}\nonumber\\
        &\quad\le \LRs{\E_Z\LRm{\LRs{\E\LRm{\tilde{\mathbbm{1}}\{T(h(X),Y)>0\}-{\mathbbm{1}}\{T(h(X),Y)>0\}\mid Z}}^2}}^{1/2}\nonumber\\
        &\quad+\LRs{\E_Z\LRm{\LRs{\E\LRm{\tilde{\mathbbm{1}}\{T(h^{\rm ora}(X),Y)>0\}-{\mathbbm{1}}\{T(h^{\rm ora}(X),Y)>0\}\mid Z}}^2}}^{1/2}\nonumber\\
        &\quad\le 2\sqrt{\frac{2}{\pi}} \overline{C}_T r.\nonumber
    \end{align*}
    Recall that
    \begin{align}
        {f}_h = \argmin_{f\in \gF_U} \E \LRm{(f(Z) - ({\alpha}(Z;h^{\rm ora}) - {\alpha}(Z;h)))^2},\nonumber
    \end{align}
     and we assume that $\E \LRm{({f}_h(X) - ({\alpha}(X;h^{\rm ora}) - {\alpha}(X;h)))^2} \leq \eta^2$ holds for any $h\in \gH$. Notice that
    \begin{align}
        P \tilde{\psi}(h, {f}_h) - P \tilde{\psi}(h^{\rm ora}, {f}_h) &= 
        \E\LRm{{f}_h(Z) (\tilde{\alpha}(Z;h) - \tilde{\alpha}(Z;h^{\rm ora}))}\nonumber\\
        &= \E\LRm{{f}_h^2(Z)} + \E\LRm{{f}_h(Z) ({\alpha}(Z;h) - {\alpha}(Z;h^{\rm ora}) - {f}_h(Z))}\nonumber\\
        &+ \E\LRm{{f}_h(Z) (\LRl{\tilde{\alpha}(Z;h) - \tilde{\alpha}(Z;h^{\rm ora})}-\LRl{{\alpha}(Z;h) - {\alpha}(Z;h^{\rm ora})})}\nonumber\\
        &\geq \|{f}_h\|_{L_2}^2 - \|{f}_h\|_{L_2} \|{\alpha}(Z;h) - {\alpha}(Z;h^{\rm ora}) - {f}_h(Z)\|_{L_2}\nonumber\\
        &- \|{f}_h\|_{L_2} \|\LRl{\tilde{\alpha}(Z;h) - \tilde{\alpha}(Z;h^{\rm ora})}-\LRl{{\alpha}(Z;h) - {\alpha}(Z;h^{\rm ora})}\|_{L_2}\nonumber\\
        &\geq \|{f}_h\|_{L_2}^2 - \|{f}_h\|_{L_2} (\eta+2\sqrt{\frac{2}{\pi}} \overline{C}_T r).\nonumber
    \end{align}
    Note that 
    \begin{align}
        \|{\alpha}(Z;h)\!-\!{\alpha}(Z;h^{\rm ora})\|_{L_2}&\!\le\!\|\tilde{\alpha}(Z;h)\!-\!\tilde{\alpha}(Z;h^{\rm ora})\|_{L_2}\nonumber\\
        &\quad +\|\LRs{{\alpha}(Z;h)\!-\!\tilde{\alpha}(Z;h)}\! - \!\LRs{{\alpha}(Z;h^{\rm ora})\!-\!\tilde{\alpha}(Z;h^{\rm ora})}\|_{L_2}\nonumber\\
        &\quad\le \|\tilde{\alpha}(Z;h)\!-\!\tilde{\alpha}(Z;h^{\rm ora})\|_{L_2} + 2\sqrt{\frac{2}{\pi}} \overline{C}_T r.\nonumber
        \end{align}
    Therefore, by symmetry we have
    \begin{align}
        \label{eq:alpha_smooth_approx_gap}
            \big|\|{\alpha}(Z;h)\!-\!{\alpha}(Z;h^{\rm ora})\|_{L_2}- \|\tilde{\alpha}(Z;h)\!-\!\tilde{\alpha}(Z;h^{\rm ora})\|_{L_2}\big|\le 2\sqrt{\frac{2}{\pi}} M r.
    \end{align}
    Then we have
    \begin{align}\label{eq:pop_weightcover_diff_smooth}
        \frac{P \tilde{\psi}(h, {f}_h) - P \tilde{\psi}(h^{\rm ora}, {f}_h)}{\|{f}_h\|_{L_2}} 
        &\geq \|{f}_h\|_{L_2} - \eta-2\sqrt{\frac{2}{\pi}} \overline{C}_T r\nonumber\\
        &\geq {\|{\alpha}(Z;h) - {\alpha}(Z;h^{\rm ora})\|_{L_2}-2\sqrt{\frac{2}{\pi}} \overline{C}_T r}-2\eta,
    \end{align}
    where the last inequity we use \eqref{eq:alpha_smooth_approx_gap}.
    Let $r \in [0,1]$, then $r {f}_{\tilde{h}} \in \gF_U$ since $\gF_U$ is star-shaped.
    Recall that $v^2(h, h^{\rm ora}) = \E[|\alpha(Z;h) - \alpha(Z;h^{\rm ora})|^2]$ and define
    \begin{align}
        \widetilde{\gG} = \LRl{\tilde{\psi}(h, {f}_h) - \tilde{\psi}(h^{\rm ora}, {f}_h): h\in \gH}\nonumber.
    \end{align}
    By the optimality of $f_h$, $\|f_h-({\alpha}(X;h^{\rm ora}) - {\alpha}(X;h))\|_{L_2}\le \|0-({\alpha}(X;h^{\rm ora}) - {\alpha}(X;h))\|_{L_2}=v({h},h^{\rm ora})$, for all $h\in\gH$. Thus
    \begin{align}\label{eq:upper_bound_f_h}
        \|\tilde{\psi}(h, {f}_h) - \tilde{\psi}(h^{\rm ora}, {f}_h)\|_{L_2}\le \|f_h\|_{L_2}\le 2v({h},h^{\rm ora}).
    \end{align}
    Then we have w.p. $1-2\zeta$:
    \begin{align}
    \label{eq:upper_bound_of_mse_smooth}
        &\sup_{f\in \gF}\LRl{P_n {\psi}(\tilde{h}, f) - P_n {\psi}(h^{\rm ora}, f) -  2\LRs{ P_n f^2+ \gamma\|f\|_\gF^2}}\nonumber\\
        & \geq r P_n \tilde{\psi}(\tilde{h}, {f}_{\tilde{h}}) - r P_n \tilde{\psi}(h^{\rm ora}, {f}_{\tilde{h}}) - 2r^2 \LRs{ P_n {f}_{\tilde{h}}^2+ \gamma\|{f}_{\tilde{h}}\|_\gF^2}\nonumber\\
        &\overset{\rm (i)}{\geq}  r P_n \tilde{\psi}(\tilde{h}, {f}_{\tilde{h}}) - r P_n \tilde{\psi}(h^{\rm ora}, {f}_{\tilde{h}}) - 2r^2 \LRs{\gamma\|{f}_{\tilde{h}}\|_{\gF}^2+\frac{7}{4}\|{f}_{\tilde{h}}\|_{L_2}^2+\frac{\tilde{\delta}_{n,\gF}^2}{2}}\nonumber\\
        &\overset{\rm (ii)}{\geq} r \LRs{P \tilde{\psi}(\tilde{h}, {f}_{\tilde{h}}) - P \tilde{\psi}(h^{\rm ora}, {f}_{\tilde{h}}) - { v(\tilde{h},h^{\rm ora})}\LRs{4 \delta_{n,\widetilde{\gG}} + 2\sqrt{\frac{2 \log(R_n/\zeta)}{n}}}}\nonumber\\
        &\qquad- \frac{r\log(R_n/\zeta)}{3n} - 2r^2 \LRs{\gamma\|{f}_{\tilde{h}}\|_{\gF}^2+\frac{7}{4}\|{f}_{\tilde{h}}\|_{L_2}^2+\frac{\tilde{\delta}_{n,\gF}^2}{2}} \nonumber\\
        &\overset{\rm (iii)}{\geq} r \LRs{(v(\tilde{h},h^{\rm ora})\!-\!2\sqrt{\frac{2}{\pi}} \overline{C}_T r\!-\!2\eta)\|{f}_{\tilde{h}}\|_{L_2} \!-\!{ v(\tilde{h},h^{\rm ora})}\LRs{4 \delta_{n,\widetilde{\gG}} + 2\sqrt{\frac{2 \log(R_n/\zeta)}{n}}}}\nonumber\\
        &\qquad  - \frac{r\log(R_n/\zeta)}{3n} -2r^2 \LRs{\gamma\|{f}_{\tilde{h}}\|_{\gF}^2+\frac{7}{4}\|{f}_{\tilde{h}}\|_{L_2}^2+\frac{\tilde{\delta}_{n,\gF}^2}{2}},
    \end{align}
    where (i) use \eqref{eq:variance_empi_process_upper}; (ii) holds due to Lemma \ref{lemma:concentration_with_variance} and \eqref{eq:upper_bound_f_h}; and (iii) holds due to \eqref{eq:pop_weightcover_diff_smooth}. $R_n = \lceil\log_2(1/\delta_{n,\widetilde{\gG}})\rceil$.
    If $4 \delta_{n,\widetilde{\gG}} + 2\sqrt{\frac{2 \log(R_n/\zeta)}{n}} \geq \|{f}_{\tilde{h}}\|_{L_2}/2$, we have
    \begin{align}
        v(\tilde{h}, h^{\rm ora}) \leq \|{f}_{\tilde{h}}\|_{L_2} + \eta + 2\sqrt{\frac{2}{\pi}} \overline{C}_T r \lesssim \delta_{n,\widetilde{\gG}} + \sqrt{\frac{\log(R_n/\zeta)}{n}} + \eta +  r.\nonumber
    \end{align}
    Otherwise, i.e. $4 \delta_{n,\widetilde{\gG}} + 2\sqrt{\frac{2 \log(R_n/\zeta)}{n}} \leq \|{f}_{\tilde{h}}\|_{L_2}/2$, using the upper bound in \eqref{eq:penalized_EP_upper_smooth}, by \eqref{eq:upper_bound_of_mse_smooth} w.p. $1-4\zeta$ we get
    \begin{align}
        v(\tilde{h},h^{\rm ora}) &\leq  \frac{2}{r \|{f}_{\tilde{h}}\|_{L_2}}\LRs{{4\epsilon^2} + {4\cdot 36^2\tilde{\delta}_{n,\gF}^2} + 36\tilde{\delta}_{n,\gF}^2+{\tilde{\delta}_{n,\gF}^2} + 2\sqrt{\frac{2}{\pi}} \overline{C}_{\gF} \overline{C}_T r} \nonumber\\
        &\qquad + \frac{2\log(R_n/\zeta)}{3n \|{f}_{\tilde{h}}\|_{L_2}} + \frac{4r}{\|{f}_{\tilde{h}}\|_{L_2}}\LRs{\gamma\|{f}_{\tilde{h}}\|_{\gF}^2+\frac{7}{4}\|{f}_{\tilde{h}}\|_{L_2}^2+\frac{\tilde{\delta}_{n,\gF}^2}{2}} + 2\sqrt{\frac{2}{\pi}} \overline{C}_T r + 2\eta.\nonumber
    \end{align}
    If $\|{f}_{\tilde{h}}\|_{L_2} \geq \tilde{\delta}_{n,\gF} \geq \sqrt{\frac{\log(c_1/\zeta)}{c_2 n}}$ , choosing $r = \max\{\epsilon, \tilde{\delta}_{n,\gF}\}/\|{f}_{\tilde{h}}\|_{L_2}$,  by $\gamma\|f\|_\gF^2\leq \|f\|_{L_2}^2$ for all $f\in\gF$, we have
     \begin{align}
         v(\tilde{h},h^{\rm ora}) \lesssim\epsilon+\tilde{\delta}_{n,\gF}+\eta+\sqrt{\frac{\log(R_n/\zeta)}{n}}+\frac{\nu\|h^{\rm ora}\|_{\gH}^2}{\tilde{\delta}_{n,\gF}} +(1+\tilde{\delta}_{n,\gF})\frac{ r}{\tilde{\delta}_{n,\gF}}.\nonumber
     \end{align}
     Otherwise, i.e. $\|{f}_{\tilde{h}}\|_{L_2} \leq \tilde{\delta}_{n,\gF}$,  we have
     \begin{align}
         v(\tilde{h}, h^{\rm ora}) \leq \|{f}_{\tilde{h}}\|_{L_2} + \eta +2\sqrt{\frac{2}{\pi}} \overline{C}_T r\lesssim \tilde{\delta}_{n,\gF} + \eta +  r.\nonumber
     \end{align}
    Combining the relations above, we can conclude that
    \begin{align}
        v(\tilde{h},h^{\rm ora}) \lesssim \epsilon+\delta_{n,\gF_U}+\delta_{n,\widetilde{\gG}} + \eta+\frac{\nu\|h^{\rm ora}\|_{\gH}^2}{\tilde{\delta}_{n,\gF}}+\frac{ r}{\tilde{\delta}_{n,\gF}}+ \sqrt{\frac{\log(R_n/\zeta)}{n}}:=A_n,\nonumber
    \end{align}
    where we used the definition of $\tilde{\delta}_{n,\gF}$.
    Therefore, Let $\gE=\LRl{\E\LRm{(\alpha(Z; \widetilde{C}) - \alpha)^2\mid \gD_n}\le A_n}$,
    \begin{align*}
        \E\LRm{(\alpha(Z; \widetilde{C}) - \alpha)^2} &= \E\LRm{\E\LRm{(\alpha(Z; \widetilde{C}) - \alpha)^2\mid \gD_n}}\\
        &=\E\LRm{\E\LRm{(\alpha(Z; \widetilde{C}) - \alpha)^2\mid \gD_n}\one_{\gE}}\\
        &\qquad+\E\LRm{\E\LRm{(\alpha(Z; \widetilde{C}) - \alpha)^2\mid \gD_n}\one_{\gE^c}}\\
        &\le A_n+4\zeta.
    \end{align*}
    By choosing $\zeta=\frac{1}{n}$, we have 
    \begin{align*}
        \E\LRm{(\alpha(Z; \widetilde{C}) - \alpha)^2}
        &\lesssim \epsilon^2+\delta_{n,\gF_U}^2+\delta_{n,\widetilde{\gG}}^2 + \eta^2+\frac{\nu^2\|h^{\rm ora}\|_{\gH}^4}{(\delta_{n,\gF_U}+\sqrt{\log(n)/n})^2}\\
        &+\frac{r^2}{(\delta_{n,\gF_U}+\sqrt{\log(n)/n})^2}+ \sqrt{\frac{\log(n \cdot R_n)}{n}}.
    \end{align*}
    
\end{proof}

\section{Infinite-dimensional set-valued function class}
\label{sec:infinite_set_class}
Now we consider the case where $\fC$ is an infinite-dimensional set-valued function space, e.g., $\fC = \{C(x;h):h\in \gH\}$ and $\gH$ is an RKHS. In this case, let $\|\cdot\|_{\gH}$ be the equipped norm in the space $\gH$. Then, we construct the empirical MOPI prediction set via
\begin{align}\label{eq:sample_minimax_CP_infinite}
    \widehat{C} = \argmin_{C\in \fC}\max_{f\in \gF}\LRl{\widehat{\Psi}(C, f) - \gamma\|f\|_{\gF}^2 - \nu \|C\|_{\gH}^2},
\end{align}
where $\nu \geq 0$ is a regularization parameter for $C \in \fC$, {and $\|C\|_{\gH}^2 = \|h\|_{\gH}^2$ for $C(x) = C(x;h)$ for some $h\in \gH$}. Similar to Lemmas \ref{lemma:inner_close_form_finite_lambda} and \ref{lemma:inner_close_form_rkhs_lambda}, we can obtain the closed forms for the inner maximization for \eqref{eq:sample_minimax_CP_infinite} under the corresponding choice of $\gF$.


Next, we present the upper bound of MSCE for the set-valued function class defined in \eqref{eq:structure_set} when $\gF$ and $\gH$ are both RKHS with a Gaussian kernel.  For simplicity, we assume $\gZ = [0,1]^{d_{\gZ}}$ and $\gX = [0,1]^{d_{\gX}}$, where $d_{\gZ}$ and $d_{\gX}$ are the dimensions.

\begin{theorem}\label{thm:mse_infinite_set_class}
    Suppose Assumption \ref{assum:appro_error} holds for the structured set-valued class $\fC$ in \eqref{eq:structure_set}. Assume the marginal density of 
    $T(h(X),Y)$ has an upper bound for any $h\in \gH$, and $T(\cdot,\cdot)$ is Lipschitz continuous with respect to its first variable.
    If the hyperparameters $\gamma$ and $\nu$ in optimization problem \eqref{eq:sample_minimax_CP_infinite} are chosen such that  $$\gamma\|f\|_\gF^2\leq \|f\|_{L_2}^2$$ for all $f\in\gF$ and $\nu\lesssim\frac{d_{\gZ}\log^{d_{\gZ}+1}(n/d_{\gZ})}{n}$, then for sufficiently large $n$, we have
    \begin{align*}
        \E\LRm{\LRs{\alpha(Z; \widehat{C}) - \alpha}^2} \lesssim \mathsf{MSCE}(C^{\rm ora}) +\eta^2
        +\frac{md_{\gX}\log^{d_{\gX}+1}(mn/d_{\gX})}{n}+ \frac{d_{\gZ}\log^{d_{\gZ}+1}(n/d_{\gZ})}{n}.
    \end{align*}
\end{theorem}
The establishment of MSCE bound differs from that in Section~\ref{sec:theory_finite_C} when $\fC$ is finite-dimensional.
When $\fC$ is infinite-dimensional, its VC dimension is also infinite, so the complexity of the induced class $\gG$ can no longer be controlled via VC theory as we did in Section \ref{sec:rate_VC_class}. In addition, we cannot directly apply standard contraction arguments {(e.g., Proposition 5.28 of \citet{wainwright2019high})} to $\gG$ because the functions in $\gG$ are not Lipschitz due to the presence of indicator functions. To address this challenge, we use the \emph{margin cost functions} in the generalization theory of classification problems \citep{mason2000improved,koltchinskii2002empirical}, which are Lipschitz continuous and can approximate indicator functions. The final bounds are determined by optimizing the margin parameter to achieve the fastest convergence rate.

\subsection{Proof of Theorem \ref{thm:mse_infinite_set_class}}
In this section, we consider $\gF$ and $\gH$ are both RKHS with a Gaussian kernel and $C(x;h) = \{y\in \gY: T(h(x), y) \leq 0\}$. 
Formally, let $\gK_\gZ(\cdot,\cdot):\gZ\times\gZ\to \R$, $\gK_\gX(\cdot,\cdot):\gX\times\gX\to \R$ be Gaussian kernel
and $\sK_{\gK_\gZ}$, $\sK_{\gK_\gX}$ denote RKHSs equipped with kernel $\gK_\gZ$, $\gK_\gX$ respectively.
Let $\gF=\sK_{\gK_\gZ}$ and 
\begin{align}\label{eq:gH_rkhs}
    \gH=\{h=(h_1,\ldots, h_m)^\top: h_j\in\sK_{\gK_\gX}, \|h_j\|_{\sK_{\gK_\gX}}^2\le U_\gX, j\in[m]\}.
\end{align}
Consider the shape-constrained prediction set \eqref{eq:structure_set} and let $C^{\rm ora}(X)=\{y\in\gY: T(h^{\rm ora}(X),Y)\le 0\}$.
Align with the proof of Theorem \ref{thm:mse_complete}, we define the following notations. Let $\psi(h,f)=f(z)(\one\{y\notin C(x;h)\}-\alpha)=f(z)((1-\one\{T(h(x),y)\leq 0\})-\alpha)$, $v(h,h^{\rm ora})=v(C,C^{\rm ora})=\|\alpha(Z;C(X;h))-\alpha(Z;C(X;h^{\rm ora}))\|_{L_2}$ and 
$$f_h=\argmin_{f\in\gF_U}\E\LRm{(f(Z)-\{\alpha(Z;C(X;h))-\alpha(Z;C(X;h^{\rm ora}))\})^2}.$$
For any $C\in\fC$, define
$$
\|C\|^2_{\gH}:=\sum_{j=1}^m\|h_j\|_{\sK_{\gK_\gX}}^2.
$$
\begin{proof}
    Note that
    \begin{align*}
         \mathbbm{1}\{y \not\in C(x; h)\} - \mathbbm{1}\{y \not\in C(x; h^{\rm ora})\} 
         &= \mathbbm{1}\{T(h^{\rm ora}(x), y) \leq 0\}-\mathbbm{1}\{T(h(x), y) \leq 0\}\\
         &=: D(h(x),y).
    \end{align*}
    Denote $\delta_{n}(\sigma_{\ell}) = \max\{\delta_{n,\gG_{\overline{\varphi}_{\sigma_{\ell}}}^+}, \delta_{n,\gG_{\varphi_{\sigma_{\ell}}}^+}, \delta_{n,\gG_{\overline{\varphi}_{\sigma_{\ell}}}^-}, \delta_{n,\gG_{\varphi_{\sigma_{\ell}}}^-}\}$ and $R_n(\sigma_{\ell}) = \lceil\log_2(1/ \delta_{n}(\sigma_{\ell})) \rceil+1$. Define sequence $\sigma_{\ell} = 2^{-\ell+1}$ for $\ell \geq 1$. Taking $x = \log(4 R_n(\sigma_{\ell})/\zeta)$ for Lemma \ref{lemma:margin_concentrarion}, we can guarantee that w.p. $1-\zeta$,
    \begin{align}
        \forall h \in \gH,\quad \left|(P_n-P) \LRl{D(h)f_h}\right|
        &\leq 2\inf_{\ell \geq 1}\Bigg\{\Delta(h, \sigma_{\ell}) + \overline{\Delta}(h, \sigma_{\ell}) + \frac{2\log(4R_n(\sigma_{\ell})/\zeta)}{3n}\nonumber\\ 
        &\qquad + 4v(h, h^{\rm ora}) \LRs{2\delta_{n}(\sigma_{\ell}) + \sqrt{\frac{2\log(4R_n(\sigma_{\ell})/\zeta)}{n}}} \Bigg\},\nonumber
    \end{align}
    where we also used the fact $\|f_h\|_{L_2} \leq 2v(h, h^{\rm ora})$. By Lemma \ref{lemma:critical_radius_infinite_class}, we have $\Delta(h, \sigma_{\ell}),\overline{\Delta}(h, \sigma_{\ell})\lesssim\sigma_\ell$ and 
    $$
    \delta_{n}(\sigma_{\ell})\lesssim\sqrt{\frac{{d_{\gZ}\log^{d_{\gZ}+1}(n/d_{\gZ})}}{n}}+ \sqrt{\frac{m d_{\gX}\log^{d_{\gX}+1}(\frac{mn}{\sigma^2_l d_{\gX}})}{n}}.
    $$
    Hence, let $\ell=\lceil\log_2 n\rceil$, then $\Delta(h, \sigma_{\ell}),\overline{\Delta}(h, \sigma_{\ell})\lesssim\frac{1}{n}$,
    $$
    \delta_{n}(\sigma_{\ell})\lesssim\sqrt{\frac{{d_{\gZ}\log^{d_{\gZ}+1}(n/d_{\gZ})}}{n}}+ \sqrt{\frac{m d_{\gX}\log^{d_{\gX}+1}(\frac{mn}{d_{\gX}})}{n}}:=\delta_{n,\gG},
    $$
    and $R_n(\sigma_{\ell})\lesssim \log_2n$. Thus, $\forall h \in \gH$,  w.p. $1-\zeta$,
    \begin{align}\label{eq:concentration_infinite_class}
        \left|(P_n-P) \LRl{D(h)f_h}\right|
        &\leq O(1)\Bigg\{\frac{1}{n} + \frac{\log(\log_2n/\zeta)}{n}+ v(h, h^{\rm ora}) \LRs{\delta_{n,\gG} + \sqrt{\frac{\log(\log_2n/\zeta)}{n}}} \Bigg\}.
    \end{align}
    Similar to \eqref{eq:excess_risk_lower}, using \eqref{eq:concentration_infinite_class} we have
    \begin{align}\label{eq:excess_risk_lower_infH}
            &\sup_{f\in \gF}\LRl{P_n {\psi}(\widehat{h}, f) - P_n {\psi}(h^{\rm ora}, f) - 2\LRs{ P_n f^2 + \gamma\|f\|_{\gF}^2}}\nonumber\\
            & \geq r P_n {\psi}(\widehat{h}, f_{\widehat{h}}) - r P_n {\psi}(h^{\rm ora}, f_{\widehat{h}}) - 2r^2 \LRs{ P_n f_{\widehat{h}}^2 + \gamma\|f_{\widehat{h}}\|_{\gF}^2}\nonumber\\
            &\geq r (v(\widehat{h},h^{\rm ora}) - 2\eta)\|f_{\widehat{h}}\|_{L_2} - 2r^2  \LRs{{\gamma}\|f_{\widehat{h}}\|_{\gF}^2+\frac{7}{4}\|f_{\widehat{h}}\|_{L_2}^2+\frac{\tilde{\delta}_{n,\gF}^2}{2}}\nonumber\\
            & - rO(1)\Bigg\{\frac{1}{n} + \frac{\log(\log_2n/\zeta)}{n}+ v(h, h^{\rm ora}) \LRs{\delta_{n,\gG} + \sqrt{\frac{\log(\log_2n/\zeta)}{n}}} \Bigg\}.\nonumber
    \end{align}
    If $2 \delta_{n,\gG} + \sqrt{\frac{2 \log(\log_2n/\zeta)}{n}} \geq \|f_{\widehat{h}}\|_{L_2}/2$, we have
    \begin{align}
        v(\widehat{h}, h^{\rm ora}) \leq \|f_{\widehat{h}}\|_{L_2} + \eta \lesssim \delta_{n,\gG} + \sqrt{\frac{\log(\log_2n/\zeta)}{n}} + \eta.\nonumber
    \end{align}
    Otherwise, i.e., $2 \delta_{n,\gG} + \sqrt{\frac{2 \log(\log_2 n/\zeta)}{n}} \leq \|f_{\widehat{h}}\|_{L_2}/2$, using \eqref{eq:penalized_EP_upper}, w.p. $1-4\zeta$, we get
    \begin{align}
        v(\widehat{h},h^{\rm ora}) &\lesssim  \frac{1}{r \|f_{\widehat{h}}\|_{L_2}}\LRs{{\epsilon^2} + {\tilde{\delta}_{n,\gF}^2} +{\nu}\|C^{\rm ora}\|_{\gH}^2} + \eta\nonumber\\
        &\qquad  + \frac{1}{n \|f_{\widehat{h}}\|_{L_2}} + \frac{\log(\log_2 n/\zeta)}{n \|f_{\widehat{h}}\|_{L_2}} + \frac{r}{\|f_{\widehat{h}}\|_{L_2}}\LRs{\gamma\|f_{\widehat{h}}\|_{\gF}^2+\|f_{\widehat{h}}\|_{L_2}^2+{\tilde{\delta}_{n,\gF}^2}},\nonumber
    \end{align}
    If $\|f_{\widehat{h}}\|_{L_2} \geq \tilde{\delta}_{n,\gF} \geq \sqrt{\frac{\log(c_1/\zeta)}{c_2 n}}$ , choosing $r = \max\{\epsilon, \tilde{\delta}_{n,\gF}\}/\|f_{\widehat{C}}\|_{L_2}$,  by $\gamma\|f\|_\gF^2\leq \|f\|_{L_2}^2$, we have
     \begin{align}
         v(\widehat{h},h^{\rm ora}) &\lesssim {\epsilon} +{\tilde{\delta}_{n,\gF}} +\frac{\nu\|C^{\rm ora}\|_{\gH}^2}{\tilde{\delta}_{n,\gF}}+ \eta + \sqrt{\frac{\log(\log_2 n/\zeta)}{n}} + \tilde{\delta}_{n,\gF}+\tilde{\delta}_{n,\gF} + \tilde{\delta}_{n,\gF}\nonumber\\
         &\lesssim {\epsilon} +{\tilde{\delta}_{n,\gF}} + \frac{\nu\|C^{\rm ora}\|_{\gH}^2}{\tilde{\delta}_{n,\gF}}+ \eta+ \sqrt{\frac{\log(\log_2 n/\zeta)}{n}}.\nonumber
     \end{align}
     Therefore,
     \begin{align}
         v(\widehat{h},h^{\rm ora}) \lesssim\epsilon+\tilde{\delta}_{n,\gF}+\eta+\sqrt{\frac{\log(\log_2 n/\zeta)}{n}}+\frac{\nu\|C^{\rm ora}\|_{\gH}^2}{\tilde{\delta}_{n,\gF}}.
     \end{align}
     Otherwise, i.e. $\|f_{\widehat{h}}\|_{L_2} \leq \tilde{\delta}_{n,\gF}$,  we have
     \begin{align}
         v(\widehat{h}, h^{\rm ora}) \leq \|f_{\widehat{h}}\|_{L_2} + \eta \lesssim \tilde{\delta}_{n,\gF} + \eta.\nonumber
     \end{align}
     Combining the relations above, we can conclude that
    \begin{align}
        v(\widehat{h},h^{\rm ora}) \lesssim \epsilon+\delta_{n,\gF_{U}}+\delta_{n,\gG} + \eta+\frac{\nu\|C^{\rm ora}\|_{\gH}^2}{\tilde{\delta}_{n,\gF}}+ \sqrt{\frac{\log(\log_2 n/\zeta)}{n}},\nonumber
    \end{align}
    where we used the definition of $\tilde{\delta}_{n,\gF}$. Hence there exist $C>0$ such that w.p. $1-4\zeta$:
     \begin{align*}
        &\E\LRm{(\alpha(Z; \widehat{C}) - \alpha)^2\mid \gD_n}\\
        &\qquad\le C\LRs{\epsilon^2+\eta^2 +\delta^2_{n,\gF_{U}}+\delta^2_{n,\gG} + \frac{\nu^2\|C^{\rm ora}\|_{\gH}^4}{\delta^2_{n,\gF_{U}}} + {\frac{\log((\log_2 n) /\zeta)}{n}}} :=A_n.
    \end{align*}
    Therefore, Let $\gE=\LRl{\E\LRm{(\alpha(Z; \widehat{C}) - \alpha)^2\mid \gD_n}\le A_n}$
    \begin{align*}
        \E\LRm{(\alpha(Z; \widehat{C}) - \alpha)^2} &= \E\LRm{\E\LRm{(\alpha(Z; \widehat{C}) - \alpha)^2\mid \gD_n}}\\
        &=\E\LRm{\E\LRm{(\alpha(Z; \widehat{C}) - \alpha)^2\mid \gD_n}\one_{\gE}}\\ &+\E\LRm{\E\LRm{(\alpha(Z; \widehat{C}) - \alpha)^2\mid \gD_n}\one_{\gE^c}}\\
        &\le A_n+4\zeta.
    \end{align*}
    By choosing $\zeta=\frac{1}{n}$, we have 
    \begin{align*}
        \E\LRm{(\alpha(Z; \widehat{C}) - \alpha)^2}\lesssim{\epsilon^2+\eta^2 +\delta^2_{n,\gF_{U}}+\delta^2_{n,\gG} + \frac{\nu^2\|C^{\rm ora}\|_{\gH}^4}{\delta^2_{n,\gF_{U}}}+{\frac{\log(n\log_2 n)}{n}}}.
    \end{align*}
    Since $\gF=\sK_{\gK_\gZ}$, by Theorem \ref{thm:critical_radius}, we have $\delta_{n,\gF_U}\lesssim\sqrt{\frac{{d_{\gZ}\log^{d_{\gZ}+1}(n/d_{\gZ})}}{n}}$. By the definition of $\gH$ in \eqref{eq:gH_rkhs}, $\|C^{\rm ora}\|_{\gH}^2\le mU_\gX$. Then by choosing $\nu\lesssim\frac{{d_{\gZ}\log^{d_{\gZ}+1}(n/d_{\gZ})}}{n}$, we have
        \begin{align*}
        \E\LRm{(\alpha(Z; \widehat{C}) - \alpha)^2}\lesssim{\epsilon^2+\eta^2 +\frac{md_{\gX}\log^{d_{\gX}+1}(mn/d_{\gX})}{n}+ \frac{d_{\gZ}\log^{d_{\gZ}+1}(n/d_{\gZ})}{n}}.
    \end{align*}
\end{proof}

\subsubsection{Margin concentration lemma}
\begin{lemma}\label{lemma:margin_concentrarion}
    Define sequence $\sigma_{\ell} = 2^{-\ell+1}$ for $\ell \geq 1$. For any $x > 0$, w.p. $1- (\lceil \log_2(1/\delta_{n,\widetilde{\gG}_{\varphi}^-}) \rceil+\lceil \log_2(1/\delta_{n,{\gG}_{\varphi}^-}) \rceil+2) e^{-x} - (\lceil \log_2(1/\delta_{n,\widetilde{\gG}_{\varphi}^+}) \rceil+\lceil \log_2(1/\delta_{n,{\gG}_{\varphi}^+}) \rceil+2) e^{-x}$,
    \begin{align}
        \forall h \in \gH,\quad \left|(P_n-P) \LRl{D(h)f_h}\right|
        &\leq \inf_{\ell \geq 1}\LRl{ \overline{\Delta}(h, \sigma_{\ell}) +\|f_h\|_{L_2} \LRs{2\delta_{n,\gG_{\overline{\varphi}_{\sigma_{\ell}}}^+} + \sqrt{\frac{2\log x}{n}}} + \frac{\log x}{3n}}\nonumber\\
        &+ \inf_{\ell \geq 1}\LRl{\Delta(h, \sigma_{\ell}) + \|f_h\|_{L_2} \LRs{2\delta_{n,\gG_{\varphi_{\sigma_{\ell}}}^+} + \sqrt{\frac{2\log x}{n}}} + \frac{\log x}{3n}}\nonumber\\
        &+\inf_{\ell \geq 1}\LRl{ \overline{\Delta}(h, \sigma_{\ell}) +\|f_h\|_{L_2} \LRs{2\delta_{n,\gG_{\overline{\varphi}_{\sigma_{\ell}}}^-} + \sqrt{\frac{2\log x}{n}}} + \frac{\log x}{3n}}\nonumber\\
        &+ \inf_{\ell \geq 1}\LRl{{\Delta}(h, \sigma_{\ell}) + \|f_h\|_{L_2} \LRs{2\delta_{n,\gG_{\varphi_{\sigma_{\ell}}}^-} + \sqrt{\frac{2\log x}{n}}} + \frac{\log x}{3n}},\nonumber
    \end{align}
    where $g_{\overline{\varphi}_{\sigma_{\ell}}}^+(h) \in \gG_{\overline{\varphi}_{\sigma_{\ell}}}^+$, ${g}_{\varphi_{\sigma_{\ell}}}^+(h) \in \gG_{\varphi_{\sigma_{\ell}}}^+$ are defined in \eqref{eq:function_class_positive_up} and \eqref{eq:function_class_positive_low}; $g_{\overline{\varphi}_{\sigma_{\ell}}}^-(h) \in \gG_{\overline{\varphi}_{\sigma_{\ell}}}^-$, ${g}_{\varphi_{\sigma_{\ell}}}^-(h) \in \gG_{\varphi_{\sigma_{\ell}}}^-$ are defined in \eqref{eq:function_class_negative_up} and \eqref{eq:function_class_negative_low}; $\overline{\Delta}(h, \sigma_{\ell})$ and $\Delta(h, \sigma_{\ell})$ are defined in \eqref{eq:smooth_gap_up} and \eqref{eq:smooth_gap_low}.
\end{lemma}
\begin{proof}
    For any function $f$, we write $[f]^+ = f\mathbbm{1}\{f \geq 0\}$ and $[f]^- = -f\mathbbm{1}\{f < 0\}$.
    We first notice that
    \begin{align}
        (P_n - P)D(h)f_h
        &= \underbrace{(P_n - P)\LRl{D_+(h)[f_h]^+}}_{(\mathrm{I})} -  \underbrace{(P_n - P)\LRl{D(h) [f_h]^-}}_{(\mathrm{II})}.\nonumber
    \end{align}
    \begin{assumption}\label{assum:margin_function}
        The Lipschitz functions $\overline{\varphi}$ and $\varphi$ satisfy $0 \leq \varphi(u) \leq \mathbbm{1}\{u \leq 0\} \leq \overline{\varphi}(u) \leq 1.$
    \end{assumption}
    \paragraph*{Bounding the term $(\mathrm{I})$.}
    Given $\overline{\varphi}$ and $\varphi$, we define two function classes
    \begin{align}
        \gG_{\overline{\varphi}}^+ &= \{g_{\overline{\varphi}}^+(h)=[f_h]^+ \LRs{\overline{\varphi} \circ T(h) -(1-\mathbbm{1}\{T(h^{\rm ora}) \leq 0\})}:h\in \gH\},\label{eq:function_class_positive_up}\\
        \gG_{\varphi}^+ &= \{g_{\varphi}^+(h)=[f_h]^+ \LRs{\varphi \circ T(h) - (1-\mathbbm{1}\{T(h^{\rm ora}) \leq 0\})}:h\in \gH\}.\label{eq:function_class_positive_low}
    \end{align}
    Using Lemma \ref{lemma:concentration_with_variance}, w.p. $1- (\lceil \log_2(1/\delta_{n,\widetilde{\gG}_{\varphi}^+}) \rceil+1) e^{-x}$, we have
    \begin{align}
        \forall h \in \gH \in \gG_{\overline{\varphi}}^+,\quad \left|(P_n - P) g_{\overline{\varphi}}^+(h)\right| &\leq \|g_{\overline{\varphi}}^+(h)\|_{L_2} \LRs{2\delta_{n,\gG_{\overline{\varphi}}^+} + \sqrt{\frac{2\log x}{n}}} + \frac{\log x}{3n}\nonumber\\
        &\leq \|f_h\|_{L_2} \LRs{2\delta_{n,\gG_{\overline{\varphi}}^+} + \sqrt{\frac{2\log x}{n}}} + \frac{\log x}{3n}.\nonumber
    \end{align}
    And, w.p. $1- (\lceil \log_2(1/\delta_{n,\gG_{\varphi}^+}) \rceil+1) e^{-x}$, we have
    \begin{align}
        \forall h \in \gH,\quad \left|(P_n - P) g_{\varphi}^+(h)\right| &\leq \|g_{\varphi}^+(h)\|_{L_2} \LRs{2\delta_{n,\gG_{\varphi}^+} + \sqrt{\frac{2\log x}{n}}} + \frac{\log x}{3n}\nonumber\\
        &\leq \|f_h\|_{L_2} \LRs{2\delta_{n,\gG_{\varphi}^+} + \sqrt{\frac{2\log x}{n}}} + \frac{\log x}{3n}.\nonumber
    \end{align}
    For any $\sigma \in [0,1]$, we take the Lipschitz continuous function as
        \begin{align*}
        \overline{\varphi}_{\sigma}(u) =
        \begin{cases}
            1, & u < -\sigma,\\
            -u/\sigma , & - \sigma \leq u \leq 0\\
            0, & a \leq 0.
        \end{cases},\quad
        \varphi_{\sigma}(u) =
        \begin{cases}
            1, & u < 0,\\
            -u/\sigma , & 0 \leq u < \sigma\\
            0, & u \geq \sigma.
        \end{cases},
    \end{align*}
    which satisfies Assumption \ref{assum:margin_function}. 
    In addition, we know $L_{\overline{\varphi}_{\sigma}},L_{\varphi_{\sigma}} \leq 1/\sigma_{\ell}$. Define Lipschitz constants $\sigma_{\ell} = 2^{-\ell+1}$ for $\ell \geq 1$. Denote the quantities
    \begin{align}
        \overline{\Delta}(h, \sigma_{\ell}) &:= P\{\overline{\varphi}_{\sigma_{\ell}} \circ T(h) \neq \mathbbm{1}\{T(h) \leq
        0\}\},\label{eq:smooth_gap_up}\\
        {\Delta}(h, \sigma_{\ell}) &:= P\{\varphi_{\sigma_{\ell}} \circ T(h) \neq \mathbbm{1}\{T(h) \leq
        0\}\}.\label{eq:smooth_gap_low}
    \end{align}
    Then, w.p. $1 - R_n \sum_{\ell=1}^{\infty} e^{- (x + 2\log \ell)} \geq 1 - 2R_n e^{-x}$, we have
    \begin{align}
        \forall h \in \gH,\quad &P_n \LRl{D(h)[f_h]^+} 
        \leq \inf_{\ell \geq 1}\LRl{P_n g_{\overline{\varphi}_{\sigma_{\ell}}}^+(h) } \nonumber\\
        &\leq \inf_{\ell \geq 1}\LRl{Pg_{\overline{\varphi}_{\sigma_{\ell}}}^+(h) + \left|(P_n - P) g_{\overline{\varphi}_{\sigma_{\ell}}}^+(h)\right|} \nonumber\\
         &\leq P\LRl{D(h)[f_h]^+} + \inf_{\ell \geq 1}\LRl{ \overline{\Delta}(h, \sigma_{\ell}) +\|f_h\|_{L_2} \LRs{2\delta_{n,\gG_{\overline{\varphi}_{\sigma_{\ell}}}^+} + \sqrt{\frac{2\log x}{n}}} + \frac{\log x}{3n}},\nonumber
    \end{align}
    and
    \begin{align}
        \forall h \in \gH,\quad &P_n \LRl{D(h)[f_h]^+} 
        \geq \sup_{\ell \geq 1}\LRl{P_n g_{\varphi_{\sigma_{\ell}}}^+(h)} \nonumber\\
        &\geq \sup_{\ell \geq 1}\LRl{Pg_{\varphi_{\sigma_{\ell}}}^+(h) - \left|(P_n - P) g_{\varphi_{\sigma_{\ell}}}^+(h)\right|}\nonumber\\
        &\geq P\LRl{D(h)[f_h]^+} + \sup_{\ell \geq 1}\LRl{-\Delta(h, \sigma_{\ell}) - \|f_h\|_{L_2} \LRs{2\delta_{n,\gG_{\varphi_{\sigma_{\ell}}}^+} + \sqrt{\frac{2\log x}{n}}} - \frac{\log x}{3n}}.\nonumber
    \end{align}
    It follows that w.p. $1- (\lceil \log_2(1/\delta_{n,\widetilde{\gG}_{\varphi}^+}) \rceil+\lceil \log_2(1/\delta_{n,{\gG}_{\varphi}^+}) \rceil+2) e^{-x}$
    \begin{align}\label{eq:diff_positive}
        \forall h \in \gH,\quad &\left|(P_n-P) \LRl{D(h)[f_h]^+}\right|\nonumber\\
        &\leq \inf_{\ell \geq 1}\LRl{ \overline{\Delta}(h, \sigma_{\ell}) +\|f_h\|_{L_2} \LRs{2\delta_{n,\gG_{\overline{\varphi}_{\sigma_{\ell}}}^+} + \sqrt{\frac{2\log x}{n}}} + \frac{\log x}{3n}}\nonumber\\
        &+ \inf_{\ell \geq 1}\LRl{\Delta(h, \sigma_{\ell}) + \|f_h\|_{L_2} \LRs{2\delta_{n,\gG_{\varphi_{\sigma_{\ell}}}^+} + \sqrt{\frac{2\log x}{n}}} + \frac{\log x}{3n}}.
    \end{align}
    \paragraph*{Bounding the term $(\mathrm{II})$.}
    Similarly, given $\overline{\varphi}$ and $\varphi$, we define another two function classes
    \begin{align}
        \gG_{\overline{\varphi}}^- &= \{[f_h]^- \LRs{\overline{\varphi} \circ T(h) -(1-\mathbbm{1}\{T(h^{\rm ora}) \leq 0\})}:h\in \gH\},\label{eq:function_class_negative_up}\\
        \gG_{\varphi}^- &= \{[f_h]^- \LRs{\overline{\varphi} \circ T(h) -(1-\mathbbm{1}\{T(h^{\rm ora}) \leq 0\})}:h\in \gH\}.\label{eq:function_class_negative_low}
    \end{align}
    Recalling the definition of $D(h)$, we have
    \begin{align}
        P\{g_{\overline{\varphi}_{\sigma_{\ell}}}^- \neq D(h)[f_h]^-\} = P\{\overline{\varphi}_{\sigma_{\ell}} \circ T(h) \neq \mathbbm{1}\{T(h) \leq 0\}\} = \overline{\Delta}(h, \sigma_{\ell}),\nonumber\\
        P\{{g}_{\varphi_{\sigma_{\ell}}}^+ \neq D(h)[f_h]^-\} = P\{\varphi_{\sigma_{\ell}} \circ T(h) \neq \mathbbm{1}\{T(h) \leq 0\}\} = {\Delta}(h, \sigma_{\ell}).\nonumber
    \end{align}
    We can show that w.p. $1- (\lceil \log_2(1/\delta_{n,\widetilde{\gG}_{\varphi}^-}) \rceil+\lceil \log_2(1/\delta_{n,{\gG}_{\varphi}^-}) \rceil+2) e^{-x}$,
    \begin{align}\label{eq:diff_negative}
        \forall h \in \gH,\quad &\left|(P_n-P) \LRl{D(h)[f_h]^-}\right|\nonumber\\
        &\leq \inf_{\ell \geq 1}\LRl{ \overline{\Delta}(h, \sigma_{\ell}) +\|f_h\|_{L_2} \LRs{2\delta_{n,\gG_{\overline{\varphi}_{\sigma_{\ell}}}^-} + \sqrt{\frac{2\log x}{n}}} + \frac{\log x}{3n}}\nonumber\\
        &+ \inf_{\ell \geq 1}\LRl{{\Delta}(h, \sigma_{\ell}) + \|f_h\|_{L_2} \LRs{2\delta_{n,\gG_{\varphi_{\sigma_{\ell}}}^-} + \sqrt{\frac{2\log x}{n}}} + \frac{\log x}{3n}}.
    \end{align}
    Combining Eqs. \eqref{eq:diff_positive} and \eqref{eq:diff_negative}, we can prove the conclusion.
\end{proof}

\begin{lemma}\label{lemma:critical_radius_infinite_class}
    Suppose the marginal density of $T:=T(h(X),Y)$ is bounded, i.e. $p_{T}(t)\le \overline{C}_T$ with some constant $\overline{C}_T>0$ and assume for any $\vu_1, \vu_2\in\R^m$, $|T(\vu,y)-T(\vu,y)|\le \overline{L}_T \|\vu_1-\vu_2\|$ with some constant $\overline{L}_T>0$. If we choose $\gF=\sK_{\gK_\gZ}$ and $\gH$ is defined in \eqref{eq:gH_rkhs}, then we have $\overline{\Delta}(h, \sigma_{\ell}), {\Delta}(h, \sigma_{\ell}) \leq \overline{C}_T\sigma_{\ell}$, and 
    \begin{align*}
        \delta_{n,{\gG}_{\overline{\varphi}}^+},\delta_{n,\gG_{\overline{\varphi}}^-},\delta_{n,\gG_{\varphi}^+},\delta_{n,\gG_{\varphi}^-}\lesssim\frac{1}{\sqrt{n}}\cdot\sqrt{{d_{\gZ}\log^{d_{\gZ}+1}(n/d_{\gZ})}\vee m d_{\gX}\log^{d_{\gX}+1}(\frac{mn}{\sigma^2_l d_{\gX}})}.
    \end{align*}
\end{lemma}
\begin{proof}
    By definition of $\Delta$ and $\bar{\Delta}$ in \eqref{eq:smooth_gap_low} and \eqref{eq:smooth_gap_low} and the assumption,
    \begin{align*}
        \overline{\Delta}(h, \sigma_{\ell}) 
        &= P\{-\sigma_{\ell}\leq T(h) \leq 0\}=\int_{-\sigma_\ell}^0f_{T}(t)dt\leq \overline{C}_T\sigma_{\ell},\\ 
        {\Delta}(h, \sigma_{\ell}) 
        &= P\{0 \leq T(h) \leq \sigma_{\ell}\}=\int_{0}^{\sigma_\ell} f_{T}(t)dt\leq \overline{C}_T\sigma_{\ell},
    \end{align*}
    The critical part is to calculate the rate of $\gG_{\overline{\varphi}}^+,\gG_{\overline{\varphi}}^-,\gG_{{\varphi}}^+, \gG_{{\varphi}}^-$. We first consider the $\gG_{\overline{\varphi}}^+$ in \eqref{eq:function_class_positive_up}.
    Since for any $\vu_1, \vu_2\in\R^m$, $|T(\vu,y)-T(\vu,y)|\le \overline{L}_T \|\vu_1-\vu_2\|$, then $\overline{\varphi}\circ T(\vu,y)$ is a Lipschitz function w.r.t. $\vu$ with Lipschitz constant $\overline{L}_T/\sigma_l$.
    For fixed $h^{\rm ora}\in\gH$, define
    \begin{align}
        \widetilde{\gG}_{\overline{\varphi}}^+ &= \{g_{\overline{\varphi}}^+(h)=[f]^+ \LRs{\overline{\varphi}\circ T(h) -(1-\mathbbm{1}\{ T(h^{\rm ora}) \leq 0\})}:h\in \gH, f\in\gF_{U_\gZ}\}.\nonumber
    \end{align}
    Here we denote $U_{\gZ} := U$ to distinguish it from $U_{\gX}$.
    Note that ${\gG}_{\overline{\varphi}}^+\subseteq\widetilde{\gG}_{\overline{\varphi}}^+$, then we only need to bound the covering number of $\widetilde{\gG}_{\overline{\varphi}}^+$. By Lemma \ref{lemma:minus_fix_covering_number}, \ref{lemma:Lipschitz_covering_number} and \ref{lemma:product_covering_number}, we have
    \begin{align}
         \log N(\rho, \widetilde{\gG}_{\overline{\varphi}}^+, \|\cdot\|_{L_2}) \le \log N(\rho/2, \gF_{U_\gZ}, \|\cdot\|_{L_2}) + \log N\LRs{\frac{\sigma_l\rho}{2\overline{L}_T}, \gH, \|\cdot\|_{L_2}}.\nonumber
    \end{align}
    Let $\gK_\gZ(\cdot,\cdot):\gZ\times\gZ\to \R$, $\gK_\gX(\cdot,\cdot):\gX\times\gX\to \R$ be Gaussian kernel
    and $\sK_{\gK_\gZ}$, $\sK_{\gK_\gX}$ denote RKHSs equipped with kernel $\gK_\gZ$, $\gK_\gX$ respectively.
    Let $\gF=\sK_{\gK_\gZ}$ and $\gH=\{h=(h_1,\ldots, h_m)^\top: h_j\in\sK_{\gK_\gX}, \|h_j\|_{\sK_{\gK_\gX}}^2\le U_\gX, j\in[m]\}$.
    By Lemma \ref{lemma:Multi_covering_number} and \eqref{eq:covering_number_rkhs}, we have
    \begin{align*}
        \log N(\rho, \gF_{U_\gZ}, \|\cdot\|_{L_2}) \leq 4^{d_{\gZ}}(6d_Z+2)\LRs{\log({U_\gZ}/\rho)}^{d_{\gZ}+1},\\
        \log N(\rho, \gH, \|\cdot\|_{L_2}) \leq m 4^{d_{\gX}}(6d_Z+2)\LRs{\log(m{U_\gX}/\rho)}^{d_{\gX}+1}.
    \end{align*}
    \begin{align}
        \overline{\gR}_{n}(\delta;{\gG}_{\overline{\varphi}}^+ )&\le \overline{\gR}_{n}(\delta;\widetilde{\gG}_{\overline{\varphi}}^+ )\le \frac{C_0}{\sqrt{n}} \int_0^\delta \sqrt{\log N(\rho, \widetilde{\gG}_{\overline{\varphi}}^+ , \|\cdot\|_{L_2})} d\rho\nonumber\\
        &\leq \frac{2}{\sqrt{n}}\int_0^{\delta/2} \sqrt{\log N(\rho, \gF_{U_\gZ}, \|\cdot\|_{L_2})} d\rho + \frac{2\overline{L}_T}{\sigma_l\sqrt{n}}\int_0^{\sigma_l\delta/(2\overline{L}_T)} \sqrt{\log N(\rho, \gH, \|\cdot\|_{L_2})} d\rho \nonumber\\
        &\lesssim \delta\sqrt{\frac{(2 d_{\gZ}+1)(4\log({U_\gZ}^2/\delta^2))^{d_{\gZ}+1}}{n}}\nonumber\\
        &\qquad + \delta\sqrt{\frac{m(2 d_{\gX}+1)(4\log\LRl{({mU_\gX}\overline{L}_T)^2/(\sigma_l\delta)^2})^{d_{\gX}+1}}{n}}.
    \end{align}
    Similar with the proof of Theorem \ref{thm:critical_radius}, since we require $\overline{\mathcal{R}}(\delta ; \widetilde{\gG}_{\overline{\varphi}}^+)\lesssim\delta^2$, we can have
    \begin{align*}
        \begin{cases}
            \frac{d_{\gZ}}{n}\log^{{d_{\gZ}+1}}(U_{\gZ}^2/\delta^2)\lesssim \frac{1}{2}\delta^2,\\[4pt]
            \frac{md_{\gX}}{n}\log^{{d_{\gX}+1}}((mU_{\gX}\overline{L}_T)^2/(\sigma_l\delta)^2)\lesssim \frac{1}{2}\delta^2.
        \end{cases} 
    \end{align*}
    For the first inequality, using \eqref{eq:delta_gF_RKHS} we get $\delta\asymp\sqrt{\frac{d_{\gZ}\log^{d_{\gZ}+1}(n/d_{\gZ})}{n}}$.
    For the second inequality, using Lemma \ref{lemma:critical_equation}, let $t=\frac{(mU_{\gX}\overline{L}_T)^2}{\sigma^2_l\delta^2}$ and check the root of the equation $t(\log t)^{d_{\gX}+1}=\frac{nm(U_X\overline{L}_T)^2}{\sigma^2_l d_{\gX}}$. Taking the logarithm on both sides, we have:
    \begin{align*}
        \log t+{(d_{\gX}+1)}\log\log t =\log (\frac{nm(U_X\overline{L}_T)^2}{\sigma^2_l d_{\gX}}).
    \end{align*}
    By Lemma \ref{lemma:critical_equation} with $L_n = \log (\frac{nm(U_X\overline{L}_T)^2}{\sigma^2_l d_{\gX}})$, we have
    \begin{align*}
        \frac{(mU_X\overline{L}_T)^2}{\sigma^2_l\delta^2}&\asymp\exp\LRl{\log \LRs{\frac{nm(U_X\overline{L}_T)^2}{\sigma^2_l d_{\gX}}}-{(d_{\gX}+1)}\log \log \LRs{\frac{nm(U_X\overline{L}_T)^2}{\sigma^2_l d_{\gX}}}}\\
        &=\frac{\frac{nm(U_X\overline{L}_T)^2}{\sigma^2_l d_{\gX}}}{\log^{d_{\gX}+1}(\frac{nm(U_X\overline{L}_T)^2}{\sigma^2_l d_{\gX}})}.
    \end{align*}
    Hence, 
    $$
    \delta_{n,{\gG}_{\overline{\varphi}}^+}\lesssim
     \delta_{n,\widetilde{\gG}_{\overline{\varphi}}^+}\asymp\frac{1}{\sqrt{n}}\cdot\sqrt{{d_{\gZ}\log^{d_{\gZ}+1}(n/d_{\gZ})}\vee m d_{\gX}\log^{d_{\gX}+1}(\frac{mn}{\sigma^2_l d_{\gX}})}.
    $$
    The above derivation applies equally to 
    $\gG_{\overline{\varphi}}^-,\gG_{{\varphi}}^+, \gG_{{\varphi}}^-$
    since $\varphi$ and $\overline{\varphi}$, $[f]^+$ and $[f]^-$ share the same Lipschitz constant. 
    Similarly, we can bound 
    $\delta_{n,\gG_{\overline{\varphi}}^-}$, 
    $\delta_{n,\gG_{\varphi}^+}$, and 
    $\delta_{n,\gG_{\varphi}^-}$ 
    in the same manner. 
    Therefore, 
    $$
    \delta_{n,{\gG}_{\overline{\varphi}}^+},\delta_{n,\gG_{\overline{\varphi}}^-},\delta_{n,\gG_{\varphi}^+},\delta_{n,\gG_{\varphi}^-}\lesssim\frac{1}{\sqrt{n}}\cdot\sqrt{{d_{\gZ}\log^{d_{\gZ}+1}(n/d_{\gZ})}\vee m d_{\gX}\log^{d_{\gX}+1}(\frac{mn}{\sigma^2_l d_{\gX}})}.
    $$
\end{proof}

\section{Auxiliary Lemmas}
\begin{lemma}\label{lemma:critical_equation}
    Let $\log t + d\log\log t = L_n$ be a equation w.r.t. $t$ with $d>0$ and $L_n\to\infty$ as $n\to\infty$. Then the solution $t^*$ of the equation satisfies $t^*\asymp\exp\{L_n-d\log L_n\}$.
\end{lemma}
\begin{proof}
    Consider $\log t = L-d\log L +\Delta$ and we prove $\Delta=o(1)$. Using $\log(1+u)=u+O(u^2)$, we have
    \begin{align}
    \label{eq:Delta_order}
        &\qquad(L-{d}\log L +\Delta)+{d}\log(L-{d}\log L +\Delta)=L\nonumber\\
        &\Longleftrightarrow(L-{d}\log L +\Delta)+{d}\log L+{d}\log\LRs{1-{d}\frac{\log L}{L} +\frac{\Delta}{L}}=L\nonumber\\
        &\Longleftrightarrow\Delta+{d}\left(-\frac{{d} \log L}{L}+O\left(\frac{(\log L)^{2}}{L^{2}}\right)+\frac{\Delta}{L}\right)=0\nonumber\\
        &\Longleftrightarrow\Delta\left(1+\frac{{d}}{L}\right)={d}^{2} \frac{\log L}{L}+O\left(\frac{(\log L)^{2}}{L^{2}}\right) .
    \end{align}
    As $n\to\infty$, $L\to\infty$, the right hand side of \eqref{eq:Delta_order} converge to 0, therefore $\Delta=o(1)$. Therefore, solution $t^*\asymp\exp\{L_n-d\log L_n\}$.
\end{proof}

\begin{lemma}\label{lemma:minus_fix_covering_number}
    Let $\gF$ be a function class with $\rho$-covering number $ N(\rho, \gF, \|\cdot\|_{L_2})$, and $g$ be any measurable function. Let $\gF-g:=\{f-g:f\in\gF\}$, then for any $\rho>0$, $N(\rho, \gF, \|\cdot\|_{L_2})=N(\rho, \gF-g, \|\cdot\|_{L_2})$.
\end{lemma}
\begin{proof}
    Let $\{f_{\ell}\}_{\ell=1}^{N(\rho,\gF,\|\cdot\|_{L_2})}$ be a $\rho$-covering of $\gF$ under the $L_2$ norm. 
    For any $h\in \gF-g$, there exists $f\in\gF$ such that $h=f-g$. 
    By the definition of the covering number, there exists some $f_{\ell_1}$ satisfying
    \[
    \|(f-g)-(f_{\ell_1}-g)\|_{L_2}
    = \|f-f_{\ell_1}\|_{L_2}
    \le \rho .
    \]
    Hence, the collection $\{f_{\ell}-g\}_{\ell=1}^{N(\rho,\gF,\|\cdot\|_{L_2})}$ forms a $\rho$-covering of $\gF-g$, which implies $N(\rho,\gF-g,\|\cdot\|_{L_2}) \le N(\rho,\gF,\|\cdot\|_{L_2})$.
    Conversely, observe that $\gF=(\gF-g)+g$. Applying the above argument with $-g$ in place of $g$ yields
    $N(\rho,\gF,\|\cdot\|_{L_2}) \le N(\rho,\gF-g,\|\cdot\|_{L_2})$ .
    Combining the two inequalities completes the proof.
\end{proof}

\begin{lemma}\label{lemma:Lipschitz_covering_number}
    Let $\gF$ be a function class with $\rho$-covering number $ N(\rho, \gF, \|\cdot\|_{L_2})$, and $\varphi$ be a $L$-Lipschitz function. Let $\varphi\circ\gF:=\{\varphi\circ f:f\in\gF\}$, then for any $\rho>0$, $N(\rho, \varphi\circ\gF, \|\cdot\|_{L_2})\leq N(\rho/L, \gF, \|\cdot\|_{L_2})$.
\end{lemma}
\begin{proof}
    Let $\{f_{\ell}\}_{\ell=1}^{N(\rho/L,\gF,\|\cdot\|_{L_2})}$ be a $\rho/L$-covering of $\gF$ under the $L_2$ norm. By the definition of the covering number, there exists some $f_{\ell_1}$ satisfying $\|f-f_{\ell_1}\|_{L_2}\le \rho/L$. For any $\varphi\circ f\in\varphi\circ\gF$,
    \begin{align*}
        \|\varphi(f)-\varphi(f_{\ell_1})\|^2_{L_2}&=\int\LRs{\varphi(f))-\varphi(f_{\ell_1})}^2 dP\nonumber\\
        &\le L^2\int\LRs{f-f_{\ell_1}}^2 dP=L^2\|f-f_{\ell_1}\|_{L_2}^2\le \rho^2
    \end{align*}
    Hence, the collection $\{\varphi\circ f_{\ell}\}_{\ell=1}^{N(\rho/L,\gF,\|\cdot\|_{L_2})}$ forms a $\rho$-covering of $\varphi\circ\gF$, which implies $N(\rho,\varphi\circ\gF,\|\cdot\|_{L_2}) \le N(\rho/L,\gF,\|\cdot\|_{L_2})$.
\end{proof}

\begin{lemma}\label{lemma:Multi_covering_number}
    Let $\gF$ be a function class with $\rho$-covering number $ N(\rho, \gF, \|\cdot\|_{L_2})$, and let $\mathfrak{F}=\{\vf=(f_1,\ldots,f_m)^\top:f_1,\ldots f_m\in\gF\}$, then for any $\rho>0$, $N(\rho, \mathfrak{F}, \|\cdot\|_{L_2(P,\R^m)})\leq \LRs{N(\rho/\sqrt{m}, \gF, \|\cdot\|_{L_2(P,\R)})}^m$.
\end{lemma}
\begin{proof}
    Let $\{f_{\ell}\}_{\ell=1}^{N(\rho/\sqrt{m},\gF,\|\cdot\|_{L_2})}$ be a $\rho/\sqrt{m}$-covering of $\gF$ under the $L_2$ norm. By the definition of the covering number, there exists some $f_{\ell_i}$ satisfying $\|f-f_{\ell_i}\|_{L_2}\le \rho/\sqrt{m}$. Let 
    \begin{align*}
        \overline{\mathfrak{F}}=\{(f_{\ell_1},\ldots,f_{\ell_m})^\top:\ell_1,\ldots,\ell_m\in[N(\rho/\sqrt{m},\gF,\|\cdot\|_{L_2})]\}
    \end{align*}
    For any $\vf\in\mathfrak{F}$, there exist some $\bar{\vf}\in\overline{\mathfrak{F}}$ such that
    \begin{align*}
        \|\vf-\bar{\vf}\|^2_{L_2}&=\int\sum_{i=1}^m\LRs{f_i-\bar{f}_{\ell_i}}^2 dP\leq m\cdot\frac{\rho^2}{m}=\rho^2
    \end{align*}
    Hence, the collection $\overline{\mathfrak{F}}$  forms a $\rho$-covering of $\mathfrak{F}$, which implies $N(\rho,\mathfrak{F},\|\cdot\|_{L_2}) \le |\overline{\mathfrak{F}}|\le (N(\rho/\sqrt{m},\gF,\|\cdot\|_{L_2}))^m$.
\end{proof}


\begin{lemma}\label{lemma:product_covering_number}
    Let $\gF$, $\gG$ be two function classes with the $\rho$-covering number $ N(\rho, \gF, \|\cdot\|_{L_2})$ and $ N(\rho, \gG, \|\cdot\|_{L_2})$, respectively. In addition, $\sup_{f\in \gF}\|f\|_{\infty} \leq B_{\gF}$ and $\sup_{g\in\gG}\|g\|_{\infty} \leq B_{\gG}$. Then the product function class $\gM = \{fg: f\in \gF,g\in\gG\}$ satisfies
    \begin{align*}
         N(\rho, \gM, \|\cdot\|_{L_2}) \leq  N(\rho/(2B_{\gG}), \gF, \|\cdot\|_{L_2})\cdot N(\rho/(2B_{\gF}), \gG, \|\cdot\|_{L_2}).
    \end{align*}
\end{lemma}
\begin{proof}
    We first notice that for any $f_1,f_2\in\gF$ and $g_1,g_2\in \gG$, it holds that
    \begin{align*}
        \|f_1g_1 - f_2g_2\|_{L_2} &\leq \|(f_1 - f_2)g_1\|_{L_2} + \|(g_1 - g_2)f_2\|_{L_2}\\
        &\leq B_{\gG}\|f_1 - f_2\|_{L_2} + B_{\gF}\|g_1-g_2\|_{L_2}.
    \end{align*}
    Let $\{f_{\ell}\}_{\ell=1}^{ N(\rho/(2B_{\gG}), \gF, \|\cdot\|_{L_2})}$ and $\{g_{\ell}\}_{\ell=1}^{ N(\rho/(2B_{\gF}), \gG, \|\cdot\|_{L_2})}$ be $\rho/(2B_{\gG})$-covering of $\gF$ and $\rho/(2B_{\gF})$-covering $\gG$, respectively. It means that for any $f\in \gF$ and $g\in \gG$, there exist some $f_{\ell_1}$ and $g_{\ell_2}$ such that $\|f - f_{\ell_1}\|_{L_2} \leq \rho/(2B_{\gG})$ and $\|g - g_{\ell_2}\|_{L_2} \leq \rho/(2B_{\gF})$. Hence we have
    \begin{align*}
        \|fg - f_{\ell_1}g_{\ell_2}\|_{L_2} \leq \rho.
    \end{align*}
    Therefore, $\{f_{\ell_1}g_{\ell_2}:\ell_1\in [ N(\rho/(2B_{\gG}), \gF, \|\cdot\|_{L_2})],\ell_2\in [ N(\rho/(2B_{\gF}), \gG, \|\cdot\|_{L_2})]\}$ is a $\rho$-covering for $\gM$, then the conclusion follows immediately.
\end{proof}

\begin{lemma}[Theorem 14.1, \citet{wainwright2019high}]\label{lemma:quadratic_EP_bound}
    Given a star-shaped function class $\gF$, suppose $\|f\|_{\infty} \leq 1$ holds for any $f\in \gF$. Let $\delta_n$ be any positive solution of the inequality
    \begin{align*}
        \E\LRm{\sup_{f\in \gF,\|f\|_{L_2} \leq \delta} \left|\frac{1}{n} \varepsilon_i f(X_i)\right|} \leq \delta^2,
    \end{align*}
    where $\{\varepsilon_i\}_{i=1}^n$ are i.i.d. Rademacher variables.
    For any $t \geq \delta_n$, we have
    \begin{align*}
        \left|(P_n - P)f^2\right|\leq \frac{1}{2}Pf^2 + \frac{t^2}{2},\quad \forall f\in \gF,
    \end{align*}
    with probability at least $1-c_1 e^{-c_2 n\delta_n^2/b^2}$.
    If in addition $n\delta_n^2 \geq \frac{2}{c_2} \log(4 \log(1/\delta_n))$, then
    \begin{align*}
        \left|P_n f^2 - Pf^2\right|\leq c_0 \delta_n,\quad \forall f\in \gF,
    \end{align*}
    with probability at least $1-c_1^{\prime} e^{-c_2^{\prime} n\delta_n^2/b^2}$.
\end{lemma}

\begin{lemma}\label{lemma:concentration_with_variance}
    Let $\gG$ be a star-shaped function class. Denote the critical radius of $\gG$ as $\delta_{n, \gG}$. If $\sup_{g\in \gG}\|g\|_{\infty} \leq 1$, with probability at least $1-\zeta$, we have
    \begin{align*}
        |(P_n-P) g| \leq \|g\|_{L_2} \LRs{2 \delta_{n,\gG} + \sqrt{\frac{2 \log(R_n/\zeta)}{n}}} + \frac{\log(R_n/\zeta)}{3n}
    \end{align*}
    simultaneously holds for any $g\in \gG$, where $R_n = \lceil\log_2(1/\delta_{n,\gG})\rceil$.
\end{lemma}

\begin{proof}
    Define the following function space
    \begin{align}
        \gG(r) = \LRl{g\in \gG: \|g\|_{L_2} \leq r},\nonumber
    \end{align}
    where $r > 0$. Denote the random variable $Z_n(r) = \sup_{g\in \gG(r)} |(P_n - P) g|$.
    Applying the functional version of Bennett's inequality (e.g., Theorem 7.3 in \citet{bousquet2003concentration}), we have
    \begin{align}\label{eq:concentartion_Bennett}
        \sP\LRl{Z_n(r) \geq \E\LRm{Z_n(r)} + r\sqrt{\frac{2  x}{n}} + \frac{x}{3 n}} \leq e^{-x}.
    \end{align}
    Let $\varepsilon_i \sim \operatorname{Uniform}\{\pm 1\}$ be i.i.d Rademacher variables.
    Using the standard symmetrization technique (e.g., Lemma A.5 in \citet{Bartlett2005local}), for any $r \geq \delta_{n, \gG}$ we have
    \begin{align}\label{eq:rademacher_com_r}
        \E\LRm{Z_n(r)} & \leq 2\E\LRm{\sup_{g\in \gG(r)} \LRabs{\frac{1}{n}\sum_{i=1}^n \varepsilon_i g(X_i)}} = 2\overline{\mathcal{R}}(r ; \gG) \leq 2r \delta_{n,\gG},
    \end{align}
    where we used the fact $\overline{\mathcal{R}}(r ; \gG)/r$ is non-increasing and the definition of critical radius, then for any $r \geq \delta_{n, \gG}$,
    \begin{align*}
        \frac{\overline{\mathcal{R}}(r ; \gG)}{r} \leq \frac{\overline{\mathcal{R}}(\delta_{n,\gG}; \gG)}{\delta_{n, \gG}} \leq \delta_{n, \gG}.
    \end{align*}
    Substituting \eqref{eq:rademacher_com_r} into \eqref{eq:concentartion_Bennett}, we can have
    \begin{align}\label{eq:concentartion_var_r}
        \sP\LRl{Z_n(r) \geq 2r \delta_{n,\gG} + r\sqrt{\frac{2 x}{n}} + \frac{x}{3 n}} \leq e^{-x}.
    \end{align}
    Next, we employ the peeling argument by defining $\gG^{(0)} = \LRl{g\in \gG: \|g - g^*\|_{L_2} \leq \delta_{n, \gG}}$ and
    \begin{align}
        \gG^{(k)} = \LRl{g\in \gG: 2^{k-1}\delta_{n, \gG}^2s \leq \|g\|_{L_2} \leq 2^{k} \delta_{n, \gG}}, k \geq 1.\nonumber
    \end{align}
    Let $R_n = \lceil\log_2(1/\delta_{n,\gG})\rceil+1$, then it holds that $\cup_{k=0}^{R_n} \gG^{(k)} = \gG$. Letting $r_k = 2^{k} \delta_{n, \gG}$ and using \eqref{eq:concentartion_var_r}, we can guarantee that
    \begin{align}
        \sP\LRl{\forall k \in [K], Z_n(r_k) \geq 2r_k \delta_{n,\gG} + r_k\sqrt{\frac{2x}{n}} + \frac{x}{3n}} \leq R_ne^{-x}.\nonumber
    \end{align}
    And we also have
    \begin{align}
        \sP\LRl{Z_n(r_0) \geq 2\delta_{n,\gG} + \delta_{n,\gG}\sqrt{\frac{2x}{n}} + \frac{x}{3n}} \leq e^{-x}.\nonumber
    \end{align}
    Notice that $g\in \gG^{(k)}$ implies that $\|g\|_{L_2} \leq r_{k}$ for $k \geq 1$ and $g\in \gG^{(0)}$ implies that $\|g\|_{L_2} \leq \delta_{n,\gG}$. Hence we can conclude that, for any $g\in \gG$ such that $\|g\|_{L_2} \geq \delta_{n,\gG}$,
    \begin{align}\label{eq:concentration_var_tail}
        \sP\LRl{\sup_{g\in \gG}|(P_n - P)g| \geq 2 \|g\|_{L_2} \delta_{n,\gG} + \|g\|_{L_2}\sqrt{\frac{2 x}{n}} + \frac{x}{3n}} \leq R_ne^{-x}.
    \end{align}
    Taking $x = \log(R_n/\zeta)$ in \eqref{eq:concentration_var_tail} and denoting $v(g,g^*) = \|g-g^*\|_{L_2}$, we can prove that with probability at least $1-\zeta$,
    \begin{align}
        |(P_n - P)g| \leq \|g\|_{L_2}\LRs{2 \delta_{n,\gG} + \sqrt{\frac{2 \log(R_n/\zeta)}{n}}} + \frac{\log(R_n/\zeta)}{3n},\nonumber
    \end{align}
    simultaneously holds for any $g\in \gG$.
\end{proof}

\begin{lemma}\label{lemma:weight_cover_EP_bound}
    Let $\gG$ be the function class defined in \eqref{eq:def_G}, $R_n = \lceil \log_2(1/\delta_{n,\gG})\rceil + 1$ and $v(C, C^{\rm ora}) = \|\alpha(C) - \alpha(C^{\rm ora})\|_{L_2}$, then with probability at least $1-\delta$, we have
    \begin{align}
    |(P_n - P)(\psi(C, f_C) - \psi(C^{\rm ora}, f_C))| &\leq v(C,C^{\rm ora})\LRs{4 \delta_{n,\gG} + 2\sqrt{\frac{2 \log(R_n/\zeta)}{n}}} + \frac{\log(R_n/\zeta)}{3n},\nonumber
\end{align}
simultaneously holds for any $C\in \fC$.
\end{lemma}

\begin{proof}
   By the definition of $f_C$ in \eqref{eq:def_fh} and 
   \[
   \gG = \LRl{(x,z,y)\mapsto f_C(z)(\one\{y\notin C(x)\}-\one\{y\notin C^{\rm ora}(x)\}): C\in \fC},
   \]
   we know $\sup_{g\in \gG}\|g\|_{\infty} \leq \sup_{C\in \fC} \|f_C\|_{\infty} \leq 1$.
    Applying Lemma \ref{lemma:concentration_with_variance} on $\gG$, we can prove the conclusion by using the fact {$\|g\|_{L_2} \leq \|f_C\|_{L_2} \leq 2v(C,C^{\rm ora})$} due to the definition $f_C = \argmin_{f\in \gF}\|f - (\alpha(C) - \alpha(C^{\rm ora}))\|_{L_2}^2$ and $\|f_C\|_{L_2}\le \|f_C - (\alpha(C) - \alpha(C^{\rm ora}))\|_{L_2}+\|\alpha(C) - \alpha(C^{\rm ora})\|_{L_2}$.
\end{proof}

\section{Additional Synthetic Data Results}
\label{sec:addi_simu}

\subsection{Experiment details of Section \ref{subsec:multi_Y}}
To obtain a reliable estimate of the MSCE, we set $|\gD_{\rm test}|=10{,}000$. 
The \texttt{CC} method \cite{gibbs2023conformal} is a full conformal procedure that requires recomputation for each test point and therefore cannot be efficiently implemented with matrix operations, resulting in substantial computational cost. 
In our experiments, its runtime is approximately 10--100 times larger than that of \texttt{CC(non-full)}, which only performs quantile regression using $\gD_{\rm cal}$ (see Appendix \ref{appen:compute_eff}). 
To mitigate the excessive computational burden caused by the large number of test samples, we therefore adopt \texttt{CC(non-full)} as the \texttt{CC} baseline in this section. 

In this experiment, we set the number of equidistant partition sets to $|\gJ|=100$ when computing $\sqrt{\rm{MSCE}}$ and the number of balls to $n_B=50$ when evaluating worst-case coverage. 
For the metric Set Volume, we compute the set size $|\widehat{C}(x)|$ in closed form for both ellipsoidal and box prediction sets.

\paragraph*{Ellipsoidal sets.}
For the baseline methods, the prediction set takes the form
$\widehat{C}(x) = \left\{ y \in \mathbb{R}^{d_{\gY}} : (y-\mu_0(x))^{\top}\Sigma_0^{-1}(y-\mu_0(x)) \le q(x) \right\}$,
whose volume is given by
\[
|\widehat{C}(x)| 
= \frac{\pi^{d_{\gY}/2}}{\Gamma\!\left(\frac{d_{\gY}}{2}+1\right)} \, q(x)^{d_{\gY}/2} \, {\rm det}(\Sigma_0)^{1/2},
\]
where $\Gamma(\cdot)$ is the standard Gamma function.
For our \texttt{MOPI}, the volume of prediction set is
\[
|\widehat{C}(x)| 
= \frac{\pi^{d_{\gY}/2}}{\Gamma\!\left(\frac{d_{\gY}}{2}+1\right)} \, {\rm det}(\Sigma_0 \Sigma(x))^{1/2}.
\]

\paragraph*{Box sets.}
For the baseline methods, 
$\widehat{C}(x) = \left\{ y \in \mathbb{R}^{d_{\gY}} : \|(y-\mu_0(X))/\sigma_0(X)\|_\infty \le q(x) \right\}$,
which corresponds to an axis-aligned box with side lengths $2 q(x)\sigma_{0,j}(x)$ along each coordinate. Its volume is
\[
|\widehat{C}(x)| = \prod_{j=1}^{d_{\gY}} 2 q(x)\sigma_{0,j}(x).
\]
For our \texttt{MOPI}, the volume of prediction set is
\[
|\widehat{C}(x)| = \prod_{j=1}^{d_{\gY}} 2 \sigma_{0,j}(x)\,\sigma_j(x).
\]

For each method, all tuning parameters are selected by minimizing the average $\sqrt{\rm{MSCE}}$ over $10$ replications of the entire experiment with different random seeds. 
These seeds are distinct from those used in the reported experiments.

\subsection{Experiment details of Section \ref{sec:simu_equal_coverage}}\label{appen:details_equal_coverage}
We first describe Setting 2 in detail: $Y_i\mid \{X_i, Z_i=z\}\sim N(\mu_z(X_i),\sigma_z^2(X_i))$, $Z_i\mid X_i\sim \text{Cat}(\pi_1(X_i),\ldots,\pi_4(X_i))$, where $\text{Cat}$ is the categorical distribution. The mixing weights
$\{\pi_z(X_i)\}_{k=1}^4$ are defined through a softmax model
\[
\pi_z(X_i)
= \frac{\exp\bigl( X_i^\top \gamma_z \bigr)}
{\sum_{\ell=1}^4 \exp\bigl( X_i^\top \gamma_\ell \bigr)},
\qquad z=1,\ldots,4,
\]
with coefficient vectors $\gamma_z \in \mathbb{R}^{d_\gX}$.
\[
\mu_z(X_i) = X_i^\top a_z + 0.5 \sin^2(X_{i,1}), 
\qquad
\sigma_z(X_i) = \bigl| X_i^\top \delta_z \bigr| + c_z.
\]
Here, $a_z, \delta_z \in \mathbb{R}^{d_\gX}$ are fixed coefficient vectors and $c_z > 0$ are component-specific constants.

In this experiment, since the support of $Z$ is relatively small, we set the test sample size to $|\gD_{\rm test}|=500$. 
Due to the relatively small $|\gD_{\rm test}|$, we can adopt the full conformal version of the \texttt{CC} method in this section, following \cite{gibbs2023conformal}. All tuning parameters are selected by minimizing the average $\sqrt{\rm{MSCE}}$ over $10$ replications of the entire experiment with different random seeds, which are distinct from those used in the reported experiments. Here,
\begin{align*}
    \sqrt{\rm{MSCE}}=\LRl{|\gZ|^{-1}\sum_{z\in\gZ}\LRs{n_z^{-1} \sum_{i\in \gD_{\rm test}} \mathbbm{1}\{Z_i=z, Y_{i} \notin \widehat{C}\left(X_{i}\right)\}-\alpha}^2}^{1/2},
\end{align*}
where $n_z=\sum_{i\in \gD_{\rm test}} \mathbbm{1}\{Z_i=z\}$.
Since $Z$ is a sensitive variable, our simulation setting restricts its usage as follows: 
$Z$ is not available during pretraining or when constructing prediction intervals. It is only used in the calibration stage, where it helps learn a more accurate threshold function $h(x)$.
In our framework, we can leverage the sensitive variable $Z$ by solving \eqref{eq:sample_minimax_CP}. However, since $Z$ is unobservable for the test samples, algorithms such as \texttt{CC}, \texttt{SCP}, and \texttt{RLCP} cannot leverage $Z$ in their calibration procedures. Hence, for \texttt{CC}, \texttt{SCP}, and \texttt{RLCP}, prediction sets can only be constructed using $X$.

\subsection{Simulation results for one-dimensional label}\label{sec:test_cond_one_dim}
In this section, we consider the sublevel set \eqref{eq:pretrain_level_set} for all baseline methods and \texttt{MOPI}. 
In each replication, we independently generate pretraining, calibration, and test datasets from the same distribution. 
We use a pretraining set of size $|\gD_{\rm pre}|=3{,}000$ to train $\hat{\mu}$ via linear regression with score $s(x,y)=|y-\hat{\mu}(x)|$, a calibration set of size $|\gD_{\rm cal}|=1{,}500$ to construct prediction sets, and a test set of size $|\gD_{\rm test}|=5{,}000$ for evaluation.
\subsubsection{Test-conditional coverage}
{In the multivariate covariate setting below, we focus on one-dimensional feature-conditional coverage with respect to $X_1$, that is, we take $Z=X_1$. This is a special case of the general $Z$-conditional coverage framework studied in the paper. Since the prediction sets still depend on the full covariate vector, we additionally report the worst-case local coverage over random balls as a complementary summary of the overall multivariate heterogeneity.} Consider a heteroscedastic regression model $Y = \sum_{j=1}^{6} \frac{1}{\sqrt{6}} X_{j} + e(X)$ with $X_{1}$ independently drawn from $\mathrm{Unif}(0,5)$, and $X_{2},\ldots,X_{6}$ independently drawn from $\mathrm{N}(0,1)$. The noise term is generated as $e(X) \sim \mathrm{N}\left(0, \sigma^*(X)^2\right)$ and we consider the following three different settings:
\begin{itemize}
    \item \textbf{Setting 1($x$)}: $\sigma^*(X)^2=1 + X_{1}^2 + 0.5 \times \sin\LRs{\sum_{j=2}^{6} X_{j}}$;

    \item \textbf{Setting 2($x$)}: $\sigma^*(X)^2=1 + 0.2 \times \sum_{j=1}^{6} \sin^2(X_{j})+\sum_{j=1}^{6} X^2_{j}\cdot\one\LRl{\sum_{j=1}^6 X_{j}>3}$.

    \item \textbf{Setting 3($x$)}: $\sigma^2(X_i)=2-\one\{X_{i,1}<1.2\}+5\one\{X_{i,1}>3\}+2\one\{X_{i,3}>0\}+3\one\{X_{i,5}>0.5\}$.
\end{itemize}
In Setting 1($x$), the variance depends only on the first coordinate $X_1$, whereas in Setting 2($x$), it depends on all coordinates. 
In Setting 3($x$), the variance follows a discrete structure. 
We choose $\gF$ and $\gH$ as the same RKHS. All tuning parameters are selected by minimizing the average $\sqrt{\rm{MSCE}}$ over $10$ replications of the entire experiment with different random seeds. 
These seeds are distinct from those used in the reported experiments.
\paragraph*{Computation of evaluation metrics.}
{For the multi-dimensional settings with $d_{\gX}=6~\text{or}~11$, the $\sqrt{\rm{MSCE}}$ is intentionally computed from equidistant partition sets constructed only along the first coordinate $X_1$, so that the target matches the feature-conditional coverage objective $Z=X_1$.} For worst-case coverage evaluation, we construct $|\mathcal{B}| = 50$ random balls, each centered at a randomly selected test point, with radii independently sampled from $\mathrm{Unif}(0.5\sqrt{2d_{\mathcal{X}}},1.5\sqrt{2d_{\mathcal{X}}})$, where $d_{\mathcal{X}}$ is the dimension of $X$.

Table \ref{tab:setting12_d6} shows that \texttt{MOPI} basically maintains marginal coverage validity while achieving the highest worst-case coverage and the lowest MSCE. 
Figure \ref{fig:simu_1d_12b} presents prediction intervals for the one-dimensional $X$. Compared with \texttt{CC}, \texttt{MOPI} yields intervals that more closely align with the oracle, exhibit fewer uncovered cases, as highlighted in the zoomed-in region of Figure \ref{fig:simu_1d_12b}. Compared with \texttt{RLCP}, \texttt{MOPI} and \texttt{CC} produce smoother intervals as $x$ varies, enhancing practical usability. 



\begin{table}[ht]
\centering
\setlength{\tabcolsep}{3pt}
\caption{Coverage metrics and set size under different variance settings with $d_{\gX}=6$.}
\label{tab:setting12_d6}
\resizebox{\textwidth}{!}{
\begin{tabular}{lcccccccccccc}
\toprule
& \multicolumn{4}{c}{Setting 1($x$)} 
& \multicolumn{4}{c}{Setting 2($x$)} 
& \multicolumn{4}{c}{Setting 3($x$)} \\
Methods 
& \textit{Marginal} & \textit{Worst-case} & $\sqrt{\rm{MSCE}}$ & \textit{Set size} 
& \textit{Marginal} & \textit{Worst-case} & $\sqrt{\rm{MSCE}}$ & \textit{Set size} 
& \textit{Marginal} & \textit{Worst-case} & $\sqrt{\rm{MSCE}}$ & \textit{Set size} \\
\midrule

MOPI 
& 0.888 & 0.865 & 0.050 & 4.492
& 0.891 & 0.865 & 0.049 & 3.046
& 0.888 & 0.865 & 0.054 & 4.075 \\

CC   
& 0.883 & 0.858 & 0.056 & 4.480
& 0.883 & 0.857 & 0.055 & 3.005
& 0.882 & 0.859 & 0.062 & 4.069 \\

RLCP 
& 0.900 & 0.815 & 0.094 & 4.955
& 0.901 & 0.811 & 0.076 & 3.320
& 0.900 & 0.846 & 0.066 & 4.249 \\

SCP  
& 0.900 & 0.796 & 0.113 & 5.090
& 0.900 & 0.788 & 0.095 & 3.351
& 0.899 & 0.831 & 0.075 & 4.251 \\

\bottomrule
\end{tabular}}
\end{table}

\begin{figure}[H]
    \centering
    \includegraphics[width=\linewidth]{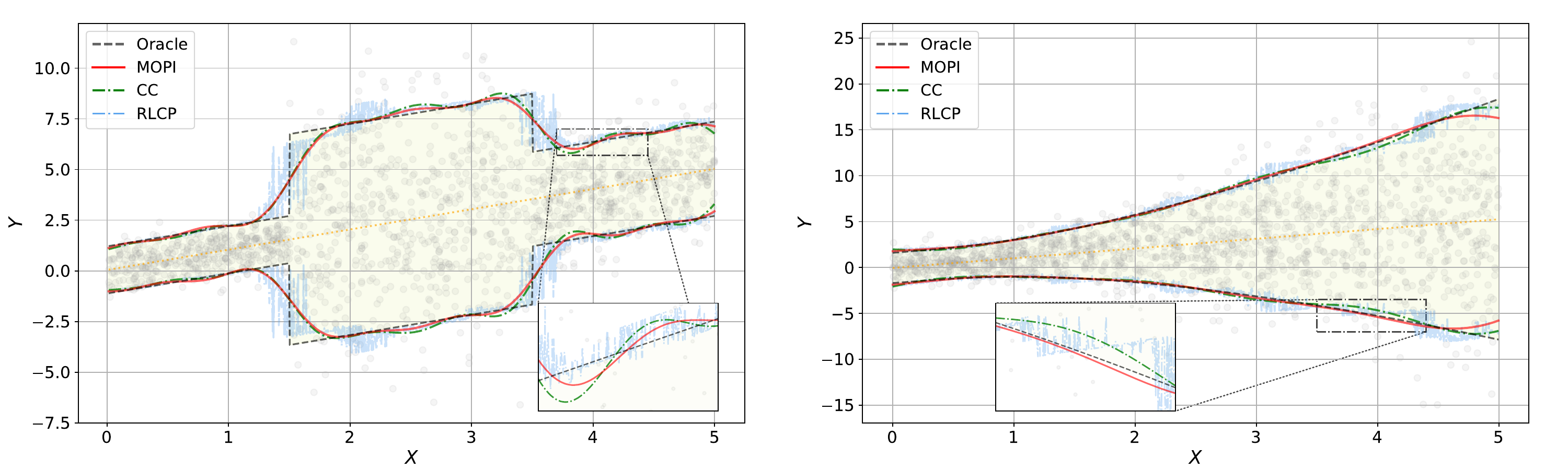}
    \caption{Prediction intervals produced by each method under the univariate settings. 
    We generate $X \sim U(0,5)$ and $Y \mid X \sim N(X,\sigma^2(X))$. Left: $\sigma^2(X)=0.5\cdot\one\{X<1.5\}+10\cdot\one\{1.5\le X<3.5\}+2\cdot\one\{3.5\le X\}$; Right: $\sigma^2(X)=1+X^3/2$. The orange solid 
    line represents the linear regression model $\hat{\mu}$ fitted on the training data; the 
    colored dashed lines denote the prediction intervals constructed by the different methods; 
    and the black dashed line corresponds to the oracle prediction interval.
    }

    \label{fig:simu_1d_12b}
\end{figure}

Additionally, we consider Setting~1$(x)$, Setting~2$(x)$ and Setting~3$(x)$ with $d_{\gX}=11$. From Tables \ref{tab:d11_all_settings}, we observe that \texttt{MOPI} achieves the highest worst-case coverage and the smallest MSCE while maintaining basically valid marginal coverage. These results indicate that, for the settings studied in our experiments, \texttt{MOPI} provides improved test-conditional coverage performance under heterogeneous variance structures and varying covariate dimensions.

\begin{table}[ht]
\centering
\setlength{\tabcolsep}{3pt}
\caption{Coverage metrics and set size for $d_{\gX}=11$ under Settings 1($x$), 2($x$), and 3($x$).}
\label{tab:d11_all_settings}
\resizebox{\textwidth}{!}{
\begin{tabular}{lcccccccccccc}
\toprule
& \multicolumn{4}{c}{Setting 1($x$)} 
& \multicolumn{4}{c}{Setting 2($x$)} 
& \multicolumn{4}{c}{Setting 3($x$)} \\
Methods 
& \textit{Marginal} & \textit{Worst-case} & $\sqrt{\rm{MSCE}}$ & \textit{Set size} 
& \textit{Marginal} & \textit{Worst-case} & $\sqrt{\rm{MSCE}}$ & \textit{Set size} 
& \textit{Marginal} & \textit{Worst-case} & $\sqrt{\rm{MSCE}}$ & \textit{Set size} \\
\midrule

MOPI 
& 0.896 & 0.884 & 0.056 & 4.546 
& 0.896 & 0.885 & 0.053 & 3.159
& 0.895 & 0.885 & 0.054 & 4.155 \\

CC   
& 0.885 & 0.868 & 0.056 & 4.466 
& 0.882 & 0.870 & 0.056 & 3.033
& 0.888 & 0.875 & 0.057 & 4.067 \\

RLCP 
& 0.901 & 0.823 & 0.100 & 5.027
& 0.900 & 0.836 & 0.069 & 3.313
& 0.901 & 0.852 & 0.068 & 4.263 \\

SCP  
& 0.900 & 0.813 & 0.112 & 5.099 
& 0.900 & 0.826 & 0.078 & 3.332
& 0.900 & 0.844 & 0.075 & 4.258 \\

\bottomrule
\end{tabular}}
\end{table}

\subsubsection{Group-conditional coverage}\label{sec:simu_group_cond}
Align with the simulation of test-conditional coverage, we investigate Setting 1($x$) and Setting 2($x$). To evaluate group-conditional coverage, we consider a collection of overlapping complex groups $\mathcal{G} = \{ G_1, \ldots, G_4 \}$, where $G_1 = \{x : x_1 + x_6 \geq 3 \}$, $G_2 = \{x : \sin(x_5) + x_1^2 \geq 3.5 \}$, $G_3 = \left\{ x : |x_2| + 3|x_{4}| + |x_6| + 2|x_{3}| \leq 3\sqrt{x_1} \right\}$, and  $G_4 = \left\{ x : 4x_3^2 + x_5^2 + 2x_{2}^2 - x_{1}^2 \leq x_1 \right\}$.
Let $Z=(\mathbbm{1}\{X\in G_1\},\ldots, \mathbbm{1}\{X\in G_4\})^T$. For a fair comparison, we choose the solving space of \texttt{CC}, $\gF$ and $\gH$ of \texttt{MOPI} to be $\{\sum_{z\in \gZ} \beta_z\mathbbm{1}\{Z = z\}: \beta_z \in \sR, z\in \{0,1\}^{2^{|\gG|}}\}$. 
For \texttt{RLCP}, we also choose the Gaussian kernel. 
From Figure \ref{fig:group_cond_1}, \ref{fig:group_cond_2} and \ref{fig:group_cond_3}, we observe that both \texttt{MOPI} and \texttt{CC} achieve accurate group-conditional coverage under both the disjoint and overlapping group structures, whereas \texttt{RLCP} and \texttt{SCP} exhibit comparatively inferior performance.

\begin{figure}[ht]
    \centering
    \includegraphics[width=\linewidth]{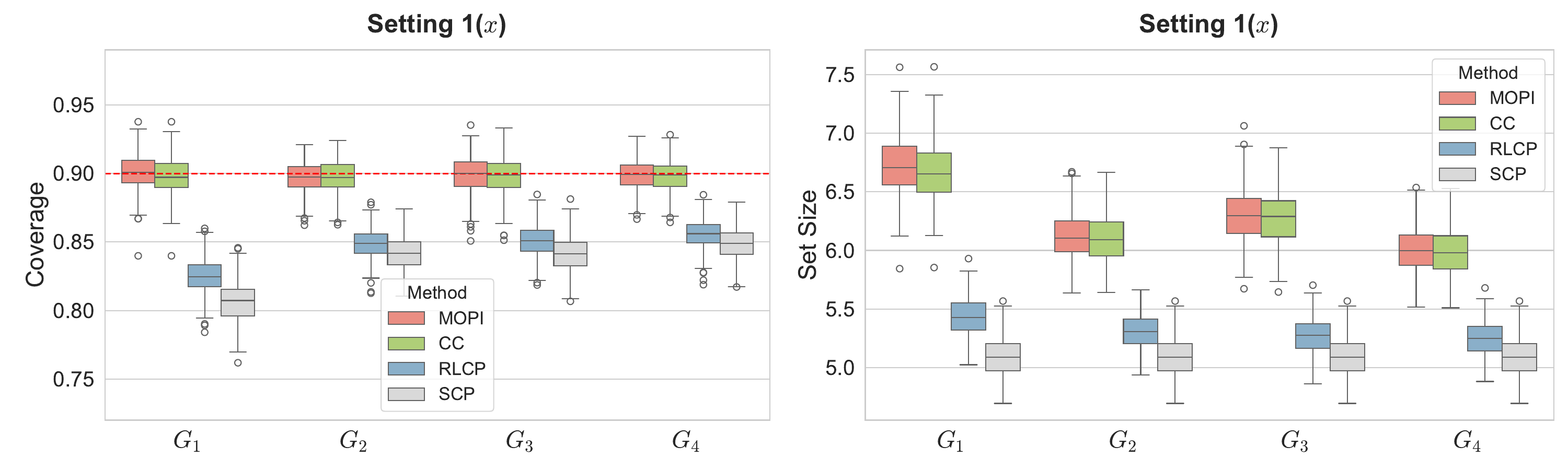}
    \caption{Group-conditional coverage and set size under Setting 1($x$) with $d_{\gX}=6$.}
    \label{fig:group_cond_1}
\end{figure}

\begin{figure}[ht]
    \centering
    \includegraphics[width=\linewidth]{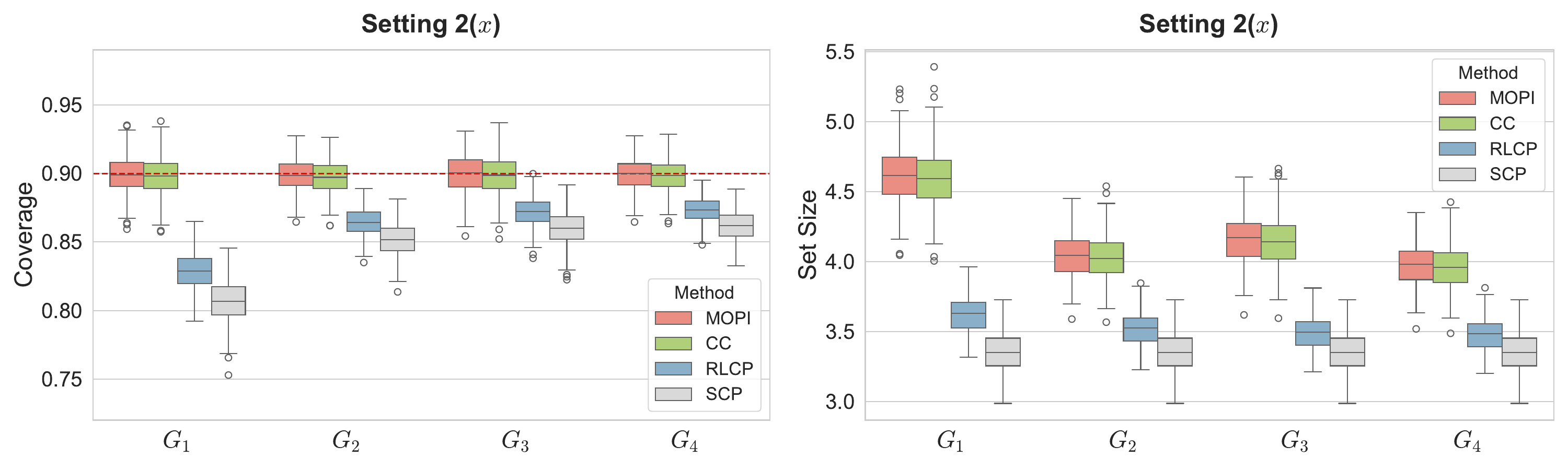}
    \caption{Group-conditional coverage and set size under Setting 2($x$) with $d_{\gX}=6$.}
    \label{fig:group_cond_2}
\end{figure}

\begin{figure}[ht]
    \centering
    \includegraphics[width=\linewidth]{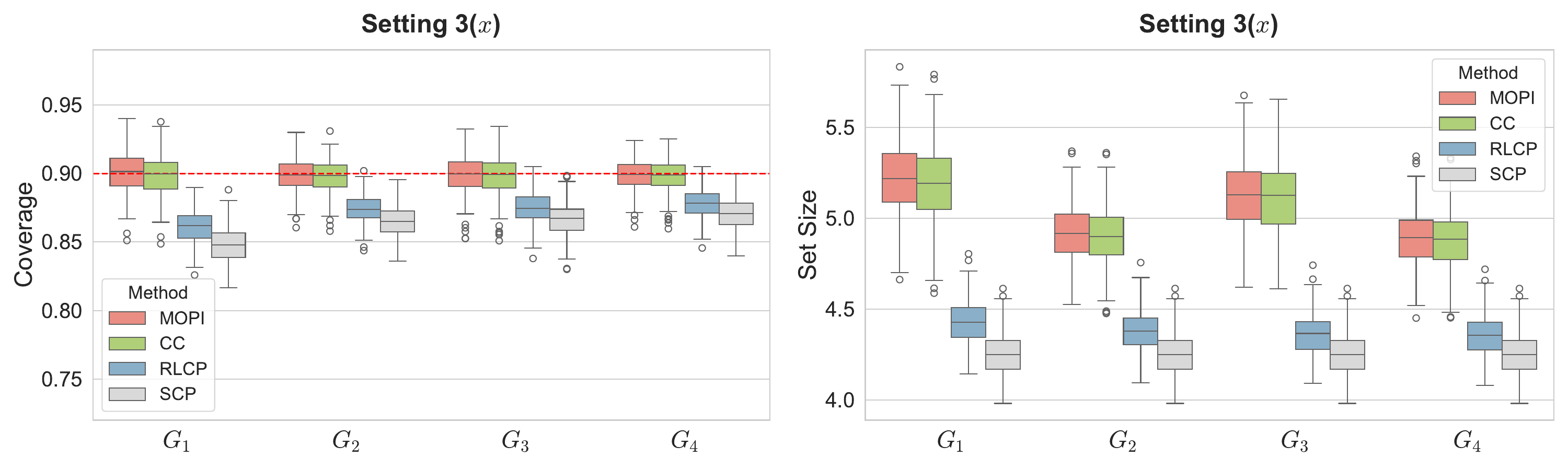}
    \caption{Group-conditional coverage and set size under Setting 3($x$) with $d_{\gX}=6$.}
    \label{fig:group_cond_3}
\end{figure}

\section{Additional Real Data Results}\label{appen:additional_realdata}
\subsection{Experiment details of Section \ref{sec:real_data}}
\label{subsec:Imple_multiY}

\paragraph*{Pretrained models.}
For \textit{ellipsoidal sets}, \texttt{SCP}, \texttt{CC}, and \texttt{RLCP} use the score $s(x,y)=(y-{\mu_0}(x))^\top{\Sigma_0}^{-1}(x)(y-{\mu_0}(x))$ and prediction sets of the form $\widehat{C}(x)=\{y:s(x,y)\le \hat{h}(x)\}$. 
Here the regression function $\mu_0(x)$ is pretrained using a random forest, while
$\Sigma_0^{-1}(x)=L_0(x)L_0(x)^\top$ is obtained from a neural network that predicts
the Cholesky factor $L^*(x)$ of $(\Sigma^*)^{-1}(x)$.
The network takes $x$ as input and outputs the lower-triangular entries of $L_0(x)$,
with positivity of the diagonal enforced via a Softplus transformation.
The network $L_0$ is pretrained by minimizing the Gaussian negative log-likelihood
\begin{align}
    \min_{L}\frac{1}{n_{\rm pre}}\sum_{i=1}^{n_{\rm pre}} \left[\frac{1}{2}\LRs{Y_i-\mu_0(X_i)}^\top \LRs{L(X_i)L(X_i)^\top} \LRs{Y_i-\mu_0(X_i)}-\frac{1}{2}\log\det\!\big(L(X_i)L(X_i)^\top\big)\right],\nonumber
\end{align}
where $\log\det(L(x)L(x)^\top) = 2\sum_j \log L_{jj}(x)$ since $L(x)$ is lower triangular.

\subsection{Medical Insurance Dataset} 
This dataset \footnote{\url{https://www.kaggle.com/datasets/rajgupta2019/medical-insurance-dataset}}
contains 3,630 records of medical charges and demographic attributes (age, BMI, number of children, sex, and smoking status). We use $Y=\log(\text{charges})$ as the label variable and take $X_{\rm all} = (X_{\rm age}, X_{\rm BMI}, X_{\rm children}, X_{\rm sex}, X_{\rm smoking})$ as the feature vector.
We consider the sublevel set \eqref{eq:pretrain_level_set} for all baseline methods and \texttt{MOPI} and split the data into three parts. The first $2000$ samples are used to pretrain  a random forest predictor $\mu$, and we define the score function as $s(x,y) = |y - \mu(x)|$. The second $1080$ samples are used to construct the prediction set and the rest $550$ samples are used as test data.
Here, we consider two experimental settings, the \emph{unmasked} case and the \emph{masked} case. The experiment is repeated 100 times with random train-calibration-test splits, and the target coverage level is $1-\alpha = 0.9$.

In the \emph{unmasked} case, all features are available at prediction time, corresponding to $X = Z = X_{\rm all}$. In this case, we evaluate whether each demographic subgroup achieves the target coverage level.
For a fair comparison, we set the function classes $\gH$, $\gF$ in \texttt{MOPI}, and the search space in \texttt{CC} as RKHS with a Gaussian kernel. To avoid uninformative intervals in \texttt{RLCP}, we tune the kernel bandwidth such that their proportion is below 10\%, and exclude these intervals when reporting lengths.
We consider marginal coverage and group-conditional coverage, which include two simple groups and two complex groups: (1) \textit{Smoker}: $X_{\rm smoking} = 1$; (2) \textit{Male}: $X_{\rm sex} = 1$; (3) \textit{$\text{BFR} > 40\%$}: individuals with a body fat ratio (BFR) exceeding 40\%, corresponding to the extremely obese group. $\text{BFR} = \left(1.2 \cdot X_{\rm BMI} + 0.23 \cdot X_{\rm age} - 10.8 \cdot X_{\rm sex} - 5.4\right)\%$, see \citet{deurenberg1991body}; (4) \textit{Male, $\text{BFR} > 24\%$}: obese males, defined as those with $X_{\rm sex} = 1$ and $\text{BFR} > 24\%$.
Figure \ref{fig:MI_data} shows that \texttt{MOPI} achieves near-perfect coverage guarantees with reasonable set sizes.
\begin{figure}[ht]
    \centering
    \includegraphics[width=\linewidth]{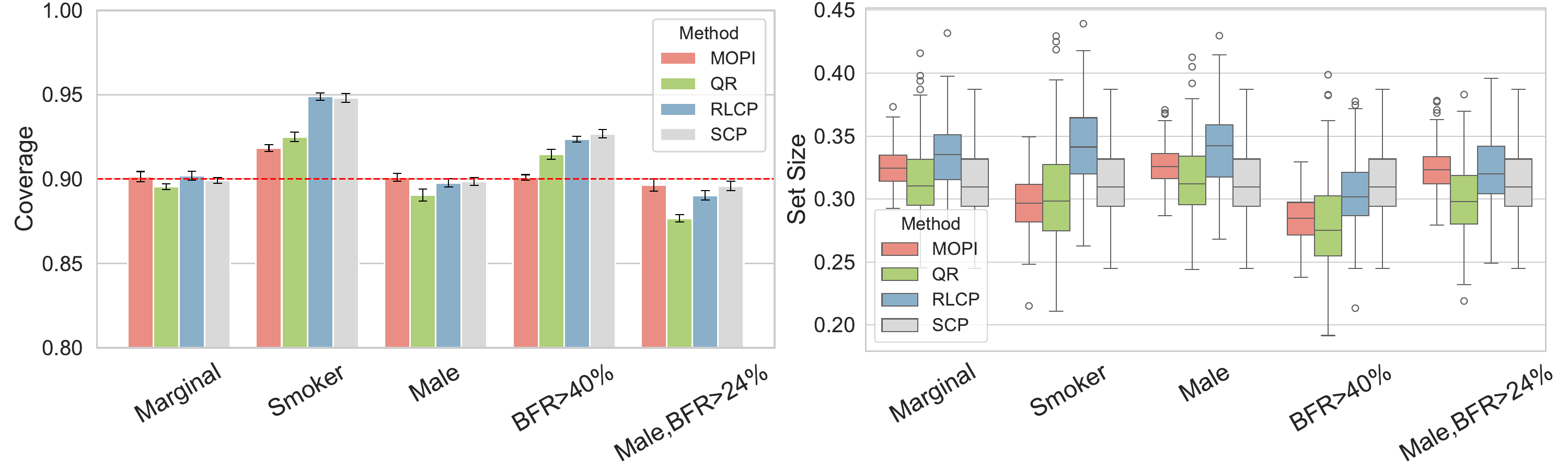}
    \caption{Performance comparison under medical insurance charges dataset. Left: marginal and group-conditional coverage; Right: corresponding set size.}
    \label{fig:MI_data}
\end{figure}

In the \emph{masked} case, the gender attribute is treated as a sensitive variable that is unavailable for the test points. Here, we let $Z = X_{\rm sex}$ and $X = X_{\rm all} \setminus X_{\rm sex}$, and assess whether the methods achieve \textit{equalized coverage} across the two gender groups despite the sensitive attribute being unobserved at prediction time. 
Figure \ref{fig:MI_data_gender} reports the coverage for the \texttt{Male} and \texttt{Female} groups under both experimental settings. 
Our \texttt{MOPI} method achieves nearly perfect coverage in both the unmasked and the masked settings, whereas the other baselines exhibit noticeable coverage gaps when the sensitive attributes are masked.

\begin{figure}[ht]
    \centering
    \includegraphics[width=\linewidth]{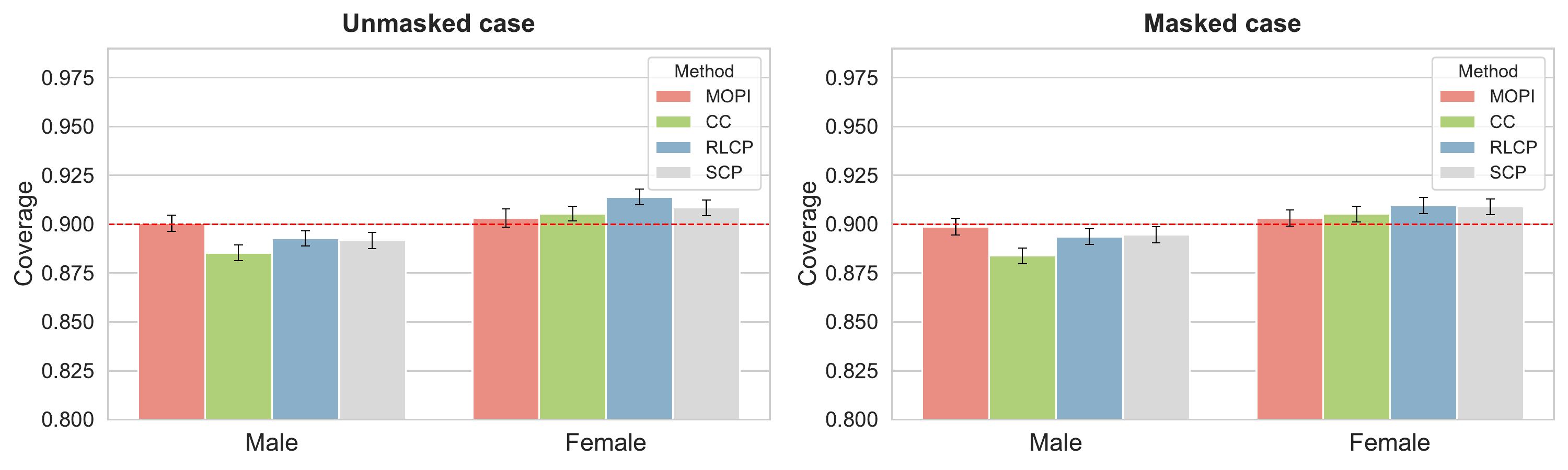}
    \caption{Coverage on the Male and Female groups for Medical Insurance Dataset.}
    \label{fig:MI_data_gender}
\end{figure}

\subsection{FashionMNIST Dataset} 
The dataset \citep{xiao2017online} contains $70{,}000$ grayscale images of clothing items ($28\times28$ pixels), evenly distributed across ten categories such as T-shirts, shoes, and bags. 
We consider the sublevel set \eqref{eq:pretrain_level_set} for all baseline methods and \texttt{MOPI}. 
The data are partitioned into training, calibration, and test sets in a ratio of $4:1:1$.
We train a convolutional neural network (CNN) $\pi(x): \gX \mapsto [0,1]^{|\gY|}$, consisting of two convolutional layers (each followed by max-pooling, ReLU, and dropout), and two fully connected layers of sizes $320 \to 20$ and $20 \to 10$, mapping the flattened features to 10 logits for classification. The 20-dimensional output of the penultimate layer is used as the feature representation. We use the adaptive prediction sets (APS) score \citep{romano2019conformalized}, $S(x,y) := \sum_{i:\,\pi_i(x) > \pi_y(x)} \pi_i(x)$. 

We set $\gH$, $\gF$ in \texttt{MOPI}, and the solve space in \texttt{CC} to be RKHSs with a Gaussian kernel.
For \texttt{PLCP} proposed by \cite{kiyani2024conformal}, the prediction set is obtained by solving
\begin{equation}
g^{*}, \boldsymbol{q}^{*}=\underset{\substack{\boldsymbol{q} \in \mathbb{R}^{m} \ g \in \mathcal{G}}}{\operatorname{argmin}} \frac{1}{n} \sum_{j=1}^{n} \sum_{i=1}^{m} g_{i}\left(X_{j}\right) \cdot\ell_{\alpha}\left(q_{i}, S_{j}\right),
\end{equation}
where $\ell_\alpha(\cdot,\cdot)$ is pinball loss. 
Following the \texttt{PLCP} setting in \citet{kiyani2024conformal}, we parameterize $g$ with $m=8$ groups using a CNN with three convolutional layers and two fully connected layers.
In addition to marginal coverage and worst-case conditional coverage, we also evaluate the absolute cluster-conditional coverage gap (\textit{Cluster-cond-gap}) and the absolute predicted-class conditional coverage gap (\textit{Class-cond-gap}). For \textit{Cluster-cond-gap}, we apply $k$-means clustering (with $k=5$) to the feature representations and compute the average absolute coverage gap across the resulting clusters. For \textit{Class-cond-gap}, we measure coverage within each predicted class and then take the mean of the average absolute coverage gap over all predicted classes. From Table \ref{tab: res_FashionMNIST} and the  Figure \ref{fig:FashionMNIST}, our method achieves marginal validity while attaining the best performance across all metrics that reflect conditional coverage. 

\begin{table}[h]
\centering
\setlength{\tabcolsep}{3pt}
\caption{Coverage metrics comparison under FashionMNIST dataset.}
\label{tab: res_FashionMNIST}
\begin{tabular}{lccccc}
\toprule
Methods & \textit{Marginal} & \textit{Worst-case} & \textit{Cluster-cond-gap} & \textit{Class-cond-gap} & \textit{Set size}\\
\midrule
MOPI & 0.910  & 0.851  & 0.021  & 0.036  & 2.294\\
CC   & 0.893  & 0.814  & 0.028  & 0.048  & 1.922\\
PLCP & 0.907  & 0.825   & 0.021  & 0.050 & 2.049\\
SCP  & 0.894  & 0.802  & 0.036  & 0.053  & 1.928\\
\bottomrule
\end{tabular}
\end{table}

\begin{figure}[H]
    \centering
    \includegraphics[width=0.8\linewidth]{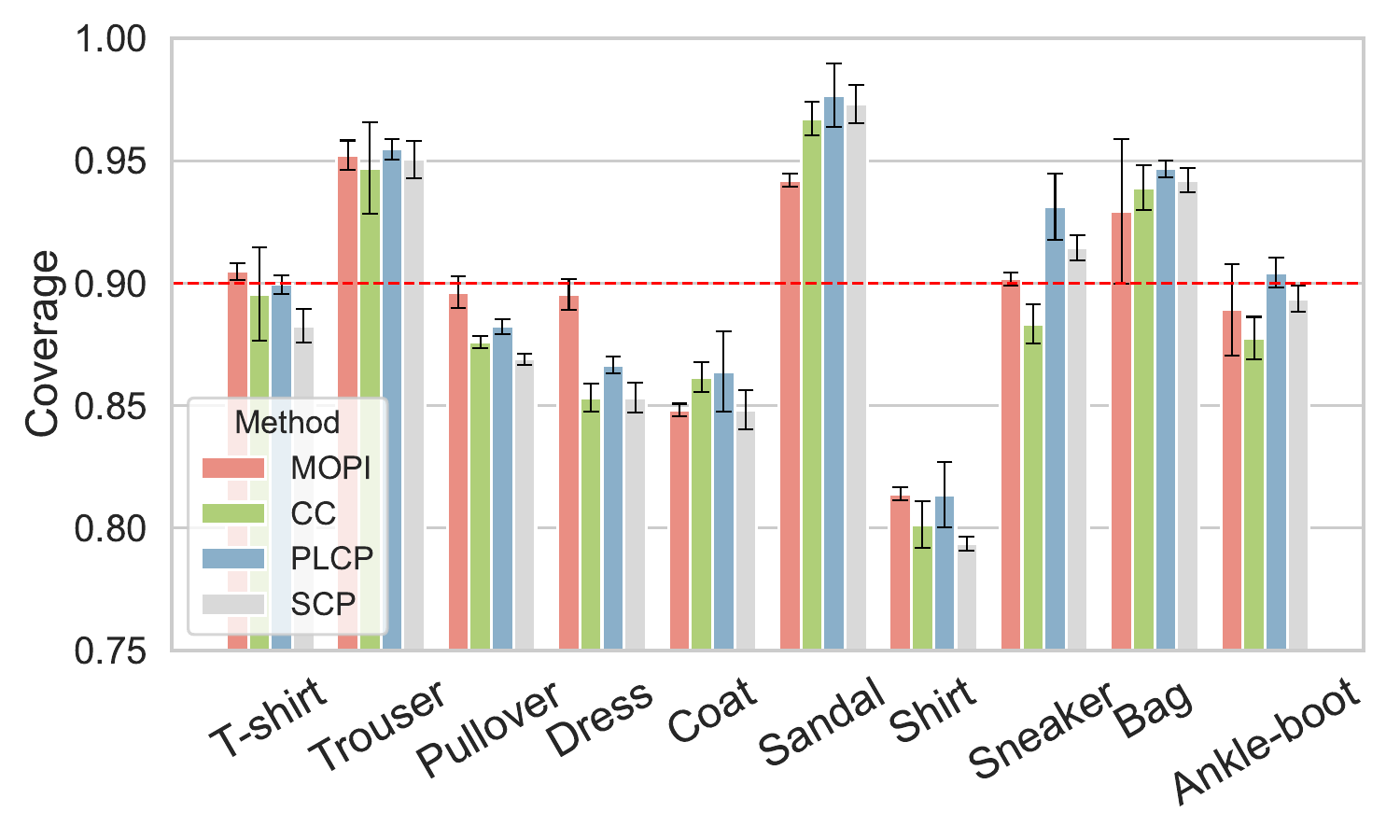}
    \caption{Predicted class-conditional coverage of the FashionMNIST dataset.}
    \label{fig:FashionMNIST}
\end{figure}